\documentclass[11pt,a4paper,leqno]{amsart}
\usepackage{amsmath,amsfonts,amssymb,amsthm}
\usepackage[alphabetic]{amsrefs}
\usepackage{hyperref, url}
\usepackage{graphicx, color}
\usepackage{enumerate}
\usepackage{fullpage}
\usepackage{xspace}
\usepackage[utf8]{inputenc}
\usepackage[T1]{fontenc}
\usepackage[english]{babel}
\usepackage{xkeyval}
\usepackage[final]{pdfpages}
\usepackage{amsmath}
\usepackage{amsfonts}
\usepackage{amssymb}
\usepackage{graphicx}
\usepackage{amsthm}
\usepackage{blindtext}
\usepackage{wrapfig}
\usepackage{scrextend}
\usepackage{verbatim}
\usepackage{appendix}
\usepackage{scalerel}
\usepackage{stix}

\usepackage{hyperref}
\usepackage[top=3cm, bottom=2.54cm, left=2.57cm, right=2.51cm]{geometry}

\usepackage{xspace}
\usepackage{tikz}

\usetikzlibrary{arrows}
\usetikzlibrary{patterns}
\usetikzlibrary{calc}
\usetikzlibrary{shapes}
\usetikzlibrary{positioning}
\usetikzlibrary{math}
\usetikzlibrary{external}
\usepackage{ifthen}

\usepackage[scr=boondox,scrscaled=1.05]{mathalfa}

\newtheorem{theorem}{Theorem}[section]
\newtheorem{proposition}[theorem]{Proposition}

\newtheorem{corollary}[theorem]{Corollary}

\newtheorem{lemma}[theorem]{Lemma}

\theoremstyle{definition}
\newtheorem{claim}[theorem]{Claim}

\newtheorem{question}[theorem]{Question}

\theoremstyle{remark}
\newtheorem{remark}[theorem]{Remark}
\newtheorem{example}[theorem]{Example}

\newcommand*\sg[1]{\{ #1 \}}

\newcommand*\cro[1]{\left[ #1\right]}

\newcommand{\Triangle}{\scaleobj{0.8}{\triangle}}

 \newcommand{\PT}{\mathrm{PT}}
 \newcommand{\PPu}{\mathrm{PP}_1}
 \newcommand{\PPd}{\mathrm{PP}_2}
 \newcommand{\INC}{\mathrm{INC}}
 \newcommand{\INCz}{\mathrm{INC}^0}
 \newcommand{\INCp}{\mathrm{INC}^+}
 \newcommand{\TPC}{\mathrm{TPC}}
 \newcommand{\TPCz}{\mathrm{TPC}^0}
 \newcommand{\TPCu}{\mathrm{TPC}^1}
 \newcommand{\TPCd}{\mathrm{TPC}^2}
 \newcommand{\TPCp}{\mathrm{TPC}^+}
 \newcommand{\TC}{\mathrm{TC}}
 \newcommand{\PC}{\mathrm{PC}}
 \newcommand{\PCz}{\mathrm{PC}^0}
 \newcommand{\PCu}{\mathrm{PC}^1}
 \newcommand{\PCd}{\mathrm{PC}^2}
 \newcommand{\PCp}{\mathrm{PC}^+}
 \newcommand{\TODO}{ \textbf{\textcolor{red}{TODO}} }
 \newcommand{\Ri}{($R_i$) }
 \newcommand{\Rip}{($R_{i+1}$) }
 \newcommand{\Si}{($S_i$) }
 \newcommand{\Ti}{($T_i$) }
 \newcommand{\Tip}{($T_{i+1}$) }
 \newcommand{\Bi}{\ensuremath{\widetilde{B}}_i\xspace}
 \newcommand{\Bip}{\ensuremath{\widetilde{B}}_{i+1}\xspace}
 \newcommand{\Biu}{\ensuremath{\widetilde{B}}_{i-1}\xspace}

\newcommand{\vz}{\ensuremath{\widetilde{v}}_0\xspace}

 \newcommand{\mr}{\mathrm}

\newcommand{\ta}{\ensuremath{\widetilde{a}}\xspace}
\newcommand{\tb}{\ensuremath{\widetilde{b}}\xspace}
\newcommand{\tc}{\ensuremath{\widetilde{c}}\xspace}
\newcommand{\td}{\ensuremath{\widetilde{d}}\xspace}
\newcommand{\te}{\ensuremath{\widetilde{e}}\xspace}
\newcommand{\tp}{\ensuremath{\widetilde{p}}\xspace}
\newcommand{\tr}{\ensuremath{\widetilde{r}}\xspace}
\newcommand{\ttt}{\ensuremath{\widetilde{t}}\xspace}
\newcommand{\tq}{\ensuremath{\widetilde{q}}\xspace}
\newcommand{\tu}{\ensuremath{\widetilde{u}}\xspace}
\newcommand{\tup}{\ensuremath{\widetilde{u}}'\xspace}
\newcommand{\tv}{\ensuremath{\widetilde{v}}\xspace}

\newcommand{\tw}{\ensuremath{\widetilde{w}}\xspace}

\newcommand{\tx}{\ensuremath{\widetilde{x}}\xspace}

\newcommand{\ty}{\ensuremath{\widetilde{y}}\xspace}
\newcommand{\tz}{\ensuremath{\widetilde{z}}\xspace}
\newcommand{\ts}{\ensuremath{\widetilde{s}}\xspace}
\newcommand{\tB}{\ensuremath{\widetilde{B}}\xspace}
\newcommand{\tS}{\ensuremath{\widetilde{S}}\xspace}
\newcommand{\tG}{\ensuremath{\widetilde{G}}\xspace}
\newcommand{\tX}{\ensuremath{\widetilde{X}}\xspace}

\newcommand{\tpi}{\ensuremath{\widetilde{\pi}}\xspace}

\newcommand{\du}[3]{\ensuremath{{#1} \bowtie_{#3} {#2}}\xspace}

\DeclareMathOperator{\cSt}{cSt}
\DeclareMathOperator{\conv}{conv}
\DeclareMathOperator{\diam}{diam}

\begin{document}

\sloppy

\centerline{\bf\Large Graphs with convex balls}

\bigskip
\centerline{\sc \large Jérémie Chalopin$^{1}$, Victor Chepoi$^{1}$, and Ugo Giocanti$^{2}$}

\medskip
\centerline{$^{1}$LIS, Aix-Marseille Universit\'e, CNRS, and Universit\'e de Toulon}
\centerline{Facult\'e des Sciences de Luminy, F-13288 Marseille Cedex 9, France}
\centerline{ {\sf \{jeremie.chalopin,victor.chepoi\}@lis-lab.fr} }

\medskip
\centerline{$^{2}$G-SCOP, Univ. Grenoble Alpes, CNRS}
\centerline{Grenoble INP, 38000 Grenoble, France,}
\centerline{ {\sf ugo.giocanti@ens-lyon.fr} }

\bigskip \noindent
{\footnotesize {\bf Abstract}
In this paper, we investigate the graphs in which all balls are convex
and the groups acting on them geometrically (which we call CB-graphs and CB-groups).
These graphs have been introduced and
characterized by Soltan and Chepoi (1983) and Farber and Jamison (1987).
CB-graphs and CB-groups generalize systolic (alias bridged) and weakly systolic graphs and groups,
which play an important role in geometric group theory. We present metric and  local-to-global characterizations of CB-graphs. Namely, we characterize CB-graphs $G$ as graphs whose triangle-pentagonal complexes
$X_{\Triangle, \pentagon}(G)$ are simply connected and balls of radius at most $3$ are convex.  Similarly to systolic and weakly systolic graphs, we prove
a dismantlability result for CB-graphs $G$: we show that their squares $G^2$ are dismantlable. This implies that the Rips complexes of CB-graphs
are contractible. Finally, we adapt and extend the approach of  Januszkiewicz and Swiatkowski (2006) for systolic groups and of Chalopin et al. (2020)
for Helly groups, to show that the CB-groups are biautomatic.
}

\tableofcontents
\setcounter{tocdepth}{2}

\section{Introduction}
In this paper, we investigate the graphs in which all balls are convex
and the groups acting on them geometrically. The (local and global)
convexity of balls is one of the fundamental features of geodesic
metric spaces, which are (locally or globally) nonpositively curved
\cite{BrHa}. The graphs with convex balls have been introduced and
characterized in \cite{FaJa,SoCh} as graphs without embedded isometric
cycles of lengths different from 3 and 5 and in which for all pairs of
vertices $u$ and $v$, all neighbors of $u$ lying on a shortest
$(u,v)$-path form a clique.  One of their important subclass is the
class of bridged graphs: these are the graphs without embedded
isometric cycles of length greater than 3 and they are exactly the
graphs in which the balls around convex sets are convex
\cite{FaJa,SoCh}.

CAT(0) (alias nonpositively curved geodesic metric spaces), introduced by Gromov in his seminal paper \cite{Gr},
and groups acting on them are fundamental objects of study in metric geometry and geometric group theory.
Graphs with strong metric properties often arise as 1-skeletons of CAT(0) cell complexes. Gromov \cite{Gr} gave a nice
combinatorial characterization of CAT(0) cube complexes as simply connected cube complexes in which the links of vertices
are simplicial flag complexes. Based on this
result, \cite{Ch_CAT, Roller1998} established a bijection between the 1-skeletons of CAT(0) cube
complexes and the median graphs, well-known in metric graph theory \cite{BaCh_survey}. A similar characterization of CAT(0) simplicial
complexes with regular Euclidean simplices as cells seems impossible. Nevertheless, Chepoi \cite{Ch_CAT} characterized the simplicial complexes having bridged graphs as 1-skeletons as the simply connected simplicial complexes
in which the links of vertices are flag complexes without embedded 4- and 5-cycles.  Januszkiewicz and Swiatkowski \cite{JS}, Haglund \cite{Haglund}, and Wise \cite{Wise}
rediscovered this class of simplicial complexes, called them {\it systolic complexes}, and used them in the context of geometric group theory. The convexity of balls around convex sets, which characterizes the 1-skeletons of
systolic complexes, is a fundamental geometric property of CAT(0) spaces.
Systolic complexes turned out to be good combinatorial analogs of CAT(0) metric spaces; cf. \cite{Elsner2009-flats,Elsner2009-isometries,Haglund,JS,JanuszkiewiczSwiatkowski2007,Pr3,Wise}.  One of the main results
of \cite{JS} is that systolic groups (i.e., groups acting geometrically on systolic complexes) are biautomatic. For other results about systolic groups, see the papers \cite{HuaOsa1,Osajda2007,OsaPrzy2009,OsajdaPrytula,Pr2}. From the results of \cite{Ch_CAT,JS} it follows that systolic complexes are contractible. Metrically systolic complexes, which are metric analogs of systolic complexes, have been introduced and investigated in \cite{HuOs_ms}. Bridged graphs have also been further investigated in several graph-theoretical  papers; cf.
\cite{AnFa, BaCh95, Ch_bridged,Po_bridged1,Po_bridged2} and the survey \cite{BaCh_survey}.

Weakly bridged/systolic graphs and complexes were introduced by Osajda \cite{Osajda} and have been thoroughly investigated in \cite{ChepoiOsajda}. The initial motivation of \cite{Osajda} was to exhibit a class of  simplicial complexes with some kind of
simplicial nonpositive curvature that will include the systolic complexes  and some other classes of complexes appearing in geometric group theory. Examples
of weakly systolic groups (i.e., groups acting geometrically on weakly systolic complexes) which are not systolic were presented in  \cite{Osajda}. The results of \cite{ChepoiOsajda}
show that weakly systolic complexes behave much like systolic complexes: their 1-skeletons are dismantlable and thus contractible, they satisfy the fixed simplex property, and can be characterized in the local-to-global way. It was shown in \cite{ChepoiOsajda} that the 1-skeletons of weakly bridged graphs are exactly the weakly modular graphs with convex balls. Weakly modular graphs are defined by two metric conditions, the triangle and quadrangle conditions \cite{BaCh95,Ch_metric} (see also Subsection \ref{ss:t-p&inc} below) and contain as subclasses most of the classes of graphs investigated in metric graph theory \cite{BaCh_survey}. For example, bridged graphs are exactly the weakly modular graphs without induced $C_4$ and $C_5$ \cite{Ch_metric}. For a general theory of weakly modular graphs and their cell complexes, see the recent paper \cite{CCHO}.

Except the papers \cite{FaJa} and \cite{SoCh}, there are no other papers investigating graphs with convex balls in full generality (notice however that convexity of balls
occurs in the characterization of 1-hyperbolic graphs \cite{BaCh_hyp}). We believe that this is due to the technical difficulties arising from the presence in such graphs of pentagons that are not paved by triangles (such pentagons do not exist in  bridged and weakly bridged graphs). For example, all $C_4$-free graphs of diameter 2 are graphs with convex balls. Hoffman and Singleton \cite{HoSi} proved that regular $C_3$-free such graphs exist only for degrees $2,3,7$, and possibly for degree $57$. The graph of degree 3 is the \emph{Petersen graph} and the graph of degree 7 is nowadays called the \emph{Hoffman-Singleton graph} and has 50 vertices and 175 edges.  Notice also that any $C_4$-free graph may occur as an induced subgraph of a graph with convex balls (and diameter 2). This shows that at small scale the graphs with convex balls can be quite arbitrary.

In this paper, we present a systematic structural study of graphs with convex balls  (which we abbreviately call \emph{CB-graphs}) and of groups acting on them geometrically
(which we call \emph{CB-groups}). Roughly speaking, we show  that at large scale the CB-graphs behave like
weakly modular and bridged graphs. Similarly to weakly modular graphs, we characterize the CB-graphs via two metric conditions, the \emph{Triangle-Pentagon} and
the \emph{Interval Neighborhood} conditions $\TPC$ and $\INC$. In analogy with weakly modular graphs, which were characterized as graphs with strongly equilateral metric triangles \cite{Ch_metric}, we show
that in CB-graphs all metric triangles are either strongly equilateral or pentagons. We use this result to extend the approach of \cite{BaCh95} for weakly modular graphs
and prove a local-to-global Helly-type theorem for convex sets in CB-graphs. Namely, we show that the Helly number of such a graph $G$ equals the size of a largest Helly independent set of diameter 1 or 2 of $G$.
These metric properties and characterizations of CB-graphs are used to provide geometric and
local-to-global (topological) properties and characterizations. Namely, we characterize CB-graphs $G$ as graphs whose triangle-pentagonal complexes
$X_{\Triangle, \pentagon}(G)$ are simply connected and balls of radius at most $3$ are convex. Consequently, we obtain the following discrete Cartan-Hadamard result:
if small balls (i.e. balls of radius 2 and 3) of a graph $G$ are convex, then the 1-skeleton of the universal cover of $X_{\Triangle, \pentagon}(G)$ is a
CB-graph.  Similarly to bridged and weakly bridged graphs, we prove
a dismantlability result for CB-graphs $G$: we show that their squares $G^2$ are dismantlable. This implies that the Rips complexes of CB-graphs
are contractible. Finally, we adapt and extend the approach of  \cite{JS} for systolic groups and of \cite{ChChGeHiOs} for Helly groups, to show that the CB-groups
are biautomatic. Such a result was not yet known for weakly systolic groups. Summarizing, here is the list of
the main results of the paper and the sections where they are proved:
\begin{itemize}
\item $G$ is a CB-graph iff $G$ satisfies the conditions $\TPC$ and $\INC$ (Section \ref{sec: INCTPC}).
\item $G$ is a CB-graph  iff its triangle-pentagon complex $X_{\Triangle, \pentagon}(G)$ is simply connected and balls of radius at most $3$ in $G$ are convex (Section~\ref{sec: loc2glob}).
\item if $G$ is a CB-graph, then $G^2$ is dismantlable. Consequently,
  all Rips complexes $X_k(G)$, $k\ge 2$ are contractible and when $G$
  is finite, any automorphism of $G$ stabilizes a convex set of
  diameter 2 (Section~\ref{sec: dism}).
\item CB-groups are biautomatic (Section~\ref{sec: NCP}).
\item if $G$ is a CB-graph, then all metric triangles are strongly equilateral or pentagons (Section~\ref{sec: metriangle}).
\item if $G$ is a CB-graph and $h(G)$ is its Helly number, then $h(G)=h_2(G)$, where $h_2(G)$ is the size of a largest $h$-independent set of diameter $\le 2$ (Section~\ref{sec: helly}).
\end{itemize}
Additionally, we show that CB-graphs satisfy the falsification by fellow traveler property (FFTP). Namely, we show that for any
non-geodesic $(u,v)$-path $\gamma$ of a CB-graph $G$ there exists a shorter $(u,v)$-path $\eta$, such that $\gamma$ and $\eta$ 2-fellow travel. FFTP  was initially introduced for Cayley graphs of
groups by Neumann and Shapiro \cite{NeSh}. In \cite{NeSh} and  \cite{Elder02} it was shown that the groups satisfying FFTP have many strong properties, in particular, they are almost convex in
the sense of Cannon \cite{Cannon}.  While convexity of balls obviously implies almost convexity, we do not know if any CB-group admits a set of generators with respect to which the Cayley graph
of the group is almost convex or satisfies FFTP. Finally note that CB-graphs and CB-groups are subclasses of shortcut graphs and shortcut groups, introduced and investigated by
Hoda \cite{Ho_shortcut}. Shortcut graphs are graphs  for which there is an upper bound on the lengths of isometrically  embedded  cycles. Shortcut groups are groups acting geometrically
on shortcut graphs.

 \section{Preliminaries}
\subsection{Graphs}
Recall that a graph $G=(V,E)$ consists of a set of vertices $V$ and a
set of edges $E \subseteq \binom{V}{2}$. In this article, all graphs
are simple, undirected and connected, but are not necessarily finite
or locally-finite.  Nevertheless, in all results we will consider only
\emph{graphs in which all cliques are finite}.  In some results, we
will additionally assume that the graphs have uniformly bounded
degrees.  We write $u \sim v$ if the vertices $u$ and $v$ are adjacent in $G$
and $u\nsim v$ if $u$ and $v$ are not adjacent.  For every $n \geq 0$,
$C_n$ denotes the cycle on $n$ vertices. We call an induced $C_3$ a
\emph{triangle}, an induced $C_4$ a \emph{square} and an induced $C_5$
a \emph{pentagon}. A \emph{wheel} $W_k$ is the graph obtained by
connecting a single vertex -- the \emph{central vertex} $c$ -- to all
vertices of a $k$--cycle $(x_1,x_2, \ldots, x_k)$.  For a set
$A\subseteq V$, we denote by $G[A]$ the subgraph of $G$ induced by
$A$. Given a graph $H$, the graph $G$ is called \emph{$H$-free} if $G$
does not contain an induced subgraph isomorphic to $H$.

We suppose that $G$ is endowed with the standard \emph{graph-distance}  $d=d_G$ so that for every $u, v \in V$, $d(u,v)$ is the length of a shortest path between $u$ and $v$ in $G$.
The \emph{ distance} from a vertex $u$ to a set $C\subseteq V$ is $d(u,C):=\min \{ d(u,x): u\in C\}$.
For any two vertices $u,v\in V$, the \emph{interval} between $u$ and $v$ is defined by $I(u,v):=\{x \in V:  d(u,v) = d(u,x) + d(x,v)\}$.

For every integer $r\geq 0$ and every vertex $u$ of $G=(V,E)$,
$B_r(u,G) := \{v \in V: d(u, v) \leq r\}$ denotes the \emph{ball} of
radius $r$ centered at $u$, and $S_r(u,G) := \{v \in V: d(u, v) = r\}$
denotes the \emph{sphere} of radius $r$ centered at $u$. When there is
no ambiguity about the graph $G$, we denote the ball $B_r(u,G)$ by
$B_r(u)$ and the sphere $S_r(u,G)$ by $S_r(u)$.  We say that a subset
$C \subseteq V$ is at \emph{uniform distance} from a vertex $u\in V$
if for every $x \in C$, we have $d(u,x) = d(u,C)$. More generally, two
sets of vertices $C,C'\subseteq V$ are at \emph{uniform distance} $k$
(notation \du{C}{C'}{k}) if for every $(x,y)\in C\times C'$ we have
$d(x,y)=k$.  The \emph{ball of radius $r$ around a set $C$} is the set
$B_r(C)=\{ v\in V: d(v,C)\le r\}$. The \emph{diameter} of a set
$C\subset V$ is $\diam(C)=\sup\{ d(u,v): u,v\in C\}$. Recall also that
the \emph{$k$-power} of a graph $G$ is the graph $G^k$ having the same
set of vertices as $G$ and two vertices $u,v$ are adjacent in $G^k$ if
and only if $d_G(u,v) \le k$. Three vertices $u,v,w$ of a graph $G$
form a \emph{metric triangle} $uvw$ if $I(u,v)\cap I(u,w)=\sg{u},$
$I(u,v)\cap I(v,w)=\sg{v}$, and $I(u,w)\cap I(v,w)=\sg{w}$. A metric
triangle $uvw$ is \emph{equilateral} when $d(u,v)=d(v,w)=d(w,u)$.

Let $H$ be a subgraph of $G$. We say that $H$ is an \emph{isometric subgraph} of $G$ if $d_H(u,v)=d_G(u,v)$ for any two vertices $u,v$ of $H$, i.e., if any two vertices $u, v$  of $H$ can be connected by a shortest path of $G$ totally included in $H$. A cycle $C$ is an \emph{isometric cycle} if $C$ is an isometric subgraph of $G$. A subset of vertices $C \subseteq V$ (or the subgraph $G[C]$  induced by $C$) is \emph{convex} in $G$ if for every two vertices $u, v \in C$, we have $I(u,v) \subseteq C$.
Since the intersection of convex sets is convex, for any set of vertices $S$ of $G$, there exists the smallest convex set $\conv(S)$ containing $S$, called the \emph{convex hull} of $S$. A weaker condition is $k$-convexity:  for an integer $k\geq 2$, we say that $C \subseteq V$ is  \emph{$k$-convex}  in $G$ if for every two vertices $u, v \in C$ such that $d(u,v) \leq k$, we have $I(u, v) \subseteq C$. We will call a graph $G$ a \emph{graph with convex balls} (or a \emph{convex balls graph}, abbreviately, a \emph{CB-graph}) if all balls $B_r(v)$ of $G$ are convex for every $v\in V$ and $r\geq 1$. Similarly, we will call graphs whose all balls are $k$-convex for some fixed $k\geq0$ \emph{graphs with $k$-convex balls}. A graph is called \emph{bridged} (or \emph{systolic})
if all balls $B_r(C)$ around convex sets $C$ are convex. Finally, we will say that the convexity of a graph $G$ \emph{preserves the diameters of sets}
if for any set $S\subset V$, $\diam(\conv(S))=\diam(S)$. Clearly, bridged graphs have convex balls. It is also easy to show that if $G$ is a graph with convex balls, then the convexity of $G$ preserves the diameters of sets \cite{SoCh}.

\subsection{Cell complexes and group actions}

All complexes considered in this paper are CW complexes. Following
\cite{Hat}*{Chapter 0}, we call them \emph{cell complexes} or just
\emph{complexes}. If all cells are simplices and the nonempty
intersections of two cells is their common face, then $X$ is called a
{\em simplicial complex}. For a cell complex $X$, by $X^{(k)}$ we
denote its \emph{$k$--skeleton}.  All cell complexes considered in
this paper will have graphs (that is, one-dimensional simplicial
complexes) as their $1$--skeleta. Therefore, we use the notation
$G(X):=X^{(1)}$. The \emph{star} of a vertex $v$ ($0$--dimensional
cell) in a complex $X$, denoted $\mr{St}(v,X)$, is the set of all
cells containing $v$. Note that the star of a vertex is not
necessarily a cell complex.  The \emph{closed star}
$\mr{cSt}(v,X)$ of $v$ in $X$ is the smallest  subcomplex of
$X$ containing $\mr{St}(v,X)$. It consists of all cells $\sigma'$
such that there exists $\sigma \in \mr{St}(v,X)$ with
$\sigma' \subseteq \sigma$.

An {\em abstract simplicial complex} $X$ on a set $V$
is a set of nonempty subsets of $V$ such that each member of $X$, called a {\em simplex},
is a finite set, and any nonempty subset of a simplex is also a simplex.
A simplicial complex $X$ naturally gives rise
to an abstract simplicial complex $X'$ on the set of vertices  of $X$
by: $U \in X'$ if and only if there is a simplex in $X$ having $U$ as its vertices.
Let $X$ be an abstract simplicial complex with ground set $V$. The \emph{barycentric subdivision} $\beta(X)$ of $X$
is the abstract simplicial complex on ground set $2^V \setminus \sg{\varnothing}$, whose vertices are the simplices
of $X$, and whose simplices are chains $\sigma_1 \subseteq \ldots \subseteq \sigma_k$ of simplices of $X$.

The {\it clique complex} of a graph $G$ is the abstract simplicial complex $X(G)$ having the cliques (i.e.,
complete subgraphs) of $G$ as simplices. A simplicial complex $X$ is a {\it flag complex} if $X$ is the clique
complex of its $1$--skeleton. The {\it $k$-Rips complex} $X_k(G)$ of a graph $G$ has the subsets $\sigma$
of vertices of $G$ such that $d_G(u,v)\le k$ for all $u,v\in \sigma$, as simplices.
Obviously, the $k$-Rips complex of $G$ is exactly the clique complex $X(G^k)$ of $G^k$. We also
define the \emph{triangle-pentagon complex} $X_{\Triangle, \pentagon}(G)$
of a graph $G$  as a two-dimensional cell complex with $1$--skeleton $G$, and such that the two-cells
are (solid) triangles and pentagons whose boundaries are identified by isomorphisms
with (graph) triangles and pentagons  in $G$.

As morphisms between cell complexes we consider all \emph{cellular
  maps}, i.e., maps sending (linearly) cells to cells. An
\emph{isomorphism} is a bijective cellular map being a linear
isomorphism (isometry) on each cell.  A cell complex $X$ is called
{\it simply connected} if it is connected and if every continuous
mapping of the 1-dimensional sphere $S^1$ into $X$ can be extended to
a continuous mapping of the disk $D^2$ with boundary $S^1$ into $X$.
Note that $X$ is connected iff $G(X)=X^{(1)}$ is connected, and $X$ is
simply connected iff $X^{(2)}$ is simply connected. The definition of
simple connectivity of cell complexes can be equivalently reformulated
in the following more combinatorial way. Let $X$ be a cell complex and
$C$ be a cycle in the $1$--skeleton of $X$.  Then a cell complex $D$
is called a {\it singular disk diagram} (or Van Kampen diagram) for
$C$ if the $1$--skeleton of $D$ is a plane graph whose inner faces are
exactly the $2$--cells of $D$ and there exists a cellular map
$\varphi:D\rightarrow X$ such that $\varphi|_{\partial D}=C$, where $\partial D$ denotes the boundary of the external face of $D$ (for more
details see \cite{LySch}*{Chapter V}).  According to Van Kampen's
lemma (\cite{LySch}, pp.\ 150--151), a cell complex $X$ is simply
connected if and only for each cycle $C$ of $X^{(1)}$ one can
construct a singular disk diagram $D=D(C)$. All 2-faces of $D$ will be
isomorphic to 2-faces of $X$ (for example, if $X$ is a
triangle-pentagonal complex, then all 2-faces of $D(C)$ will be
triangles or pentagons). Equivalently, $\partial D$ and its pre-image
$C$ in $X$ are null-homotopic, i.e., they can be transformed to a
single point by a sequence of elementary homotopies, i.e., there
exists a sequence $C_0 = C, C_1, \ldots, C_{k-1}, C_k$ of circuits
such that $C_k$ is a single point and the difference between two
consecutive circuits $C_{i-1}$ and $C_i$ is a cell of dimension $0$,
$1$, or $2$. By a \emph{circuit} of $X$, we mean a sequence
$(u_0,u_1,\ldots,u_k)$ of vertices (i.e., 0-cells) of $X$ such that
$u_i= u_{i+1}$ or $u_i \sim u_{i+1}$ for any $0 \leq i \leq k$ (where
$u_{k+1}=u_0$).
A \emph{covering (map)} of a cell complex $X$ is a cellular surjection $p\colon
\widetilde{X} \to X$ such that $p|_{\mbox{St}(\tv,\widetilde{X})}\colon \mbox{St}(\tv,\widetilde{X})\to \mbox{St}(p(\tv),X)$ is
an isomorphism for every vertex $\tv$ in $\tX$; compare \cite{Hat}*{Section 1.3}.
The space $\widetilde{X}$ is then called a \emph{covering space}.
A \emph{universal cover} of $X$ is a simply connected covering space $\widetilde{X}$. It
is unique up to isomorphism. In particular, if $X$ is simply connected, then
its universal cover is $X$ itself.

A group $\Gamma$ \emph{acts by automorphisms} on a cell complex $X$ if there is a
homomorphism $\Gamma\to \mbox{Aut}(X)$ called an \emph{action of $\Gamma$}. The action is
\emph{geometric} (or \emph{$\Gamma$ acts geometrically}) if it is proper (i.e., cells
stabilizers are finite) and cocompact (i.e., the quotient $X/\Gamma$ is compact). In the current paper we usually consider geometric
actions on graphs, viewed as one-dimensional complexes. Namely, we say that a group $\Gamma$ acts on a graph $G$ when it acts on $G$ if we consider it as a $1$-dimensional simplicial complex. Observe that if $\Gamma$ acts on $G$, then it induces a group action of $\Gamma$ on its clique complex $X(G)$, and thus we also get an induced action of $\Gamma$ on $\beta(X(G))$. Note that this also induces an action of $\Gamma$ on $G^k$ for any $k \geq 1$, as well as an action of $\Gamma$ on the triangle-pentagon $X_{\Triangle, \pentagon}(G)$. We call a group $\Gamma$ a \emph{$CB$-group} if $\Gamma$ acts geometrically on a CB-graph and a \emph{weakly systolic group} (respectively, \emph{systolic group}) if $\Gamma$ acts geometrically on a weakly systolic graph (respectively, on a systolic graph).

\subsection{The characterization of graphs with convex balls from \cite{SoCh} and \cite{FaJa}}
We recall the known structural characterization of CB-graphs:  

\begin{theorem}[\cite{SoCh,FaJa}]\label{thm: well-bridged}
 For a graph $G=(V,E)$ the following conditions are equivalent:
 \begin{itemize}
 \item[(i)] $G$ has convex balls;
 \item[(ii)] $G$ has $3$-convex balls;
 \item[(iii)]  any isometric cycle of  $G$ has length 3 or 5 and for any two vertices $u,v$ of $G$ the neighbors of $u$ in $I(u,v)$ form a clique. \end{itemize}
\end{theorem}

Although CB-graphs have no isometric
cycles of length other than $3$ and $5$, the converse is not always true. For example, the graph obtained by gluing a triangle and a pentagon
along an edge has a ball which is not convex. Hence, the condition (iii)  of Theorem \ref{thm: well-bridged} cannot be relaxed.
 \section{A metric characterization}
\label{sec: INCTPC}
In this section, we provide a new metric characterization of graphs with convex balls. This characterization will be very useful in subsequent sections.  We also characterize graphs with 2-convex balls.

\subsection{Triangle-Pentagon and Interval Neighborhood Conditions} \label{ss:t-p&inc}
Consider the following metric conditions on a graph $G$ (see Figure \ref{fig:TCQC}, upper row and the first graph of Figure \ref{fig:TPC}):
\begin{itemize}
 \item \emph{Triangle Condition} ($\TC$): for any three vertices $v,x,y$ such that $d(v,x)=d(v,y)=k$ and $x\sim y$, there exists a vertex  $z\in B_{k-1}(v)$ such that $xzy$ is a triangle of $G$;
 \item\emph{Quadrangle Condition} ($\mathrm{QC}$): for any four vertices $v,x,y,u$ such that $d(v,x)=d(v,y)=k=d(v,u)-1$ and $u\sim x,y$, $x\nsim y$, there exists
 a vertex $z\in B_{k-1}(v)$ such that $xzyu$ is a square of $G$;
 \item \emph{Pentagon Condition} ($\PC$)=($\PCz$): for any three vertices $v,x,y$ such that $d(v,x)=d(v,y)=k$ and $x\sim y$, there exist vertices $z,w,w'$ such that $z\in B_{k-2}(v)$ and  $xwzw'y$ is a pentagon of $G$.
\end{itemize}
A graph $G$ is called \emph{weakly modular} if it satisfies the
triangle and the quadrangle conditions \cite{BaCh95, Ch_metric}. A
graph $G$ is \emph{weakly systolic} (or \emph{weakly bridged}) if $G$
is weakly modular and does not contain induced 4-cycles. Equivalently,
weakly systolic graphs are the weakly modular graphs in which all
balls are convex~\cite{ChepoiOsajda}. The systolic (or bridged) graphs
are precisely the weakly modular graphs that do not contain induced
4-cycles and 5-cycles~\cite{Ch_metric}.  Equivalently, $G$ is
systolic if all isometric cycles of $G$ have length 3; this was the
initial definition of bridged/systolic graphs from \cite{FaJa,SoCh}.

\tikzexternaldisable
  \begin{figure}[h]\label{fig:TQPconditions}
    \centering
    \begin{tikzpicture}[scale=1]

    \begin{scope}[xshift = 3cm, yshift= 6cm]
     \tikzstyle{every node}=[draw, circle,fill=black,minimum size=4pt,
                            inner sep=0pt]
    \node (lab) at (0, -1)[draw=white, fill=white] {$\TC(v, xy)$};
    \node (v) at (0,0) [label=below: $v$] {};
    \node (x) at (-1,4) [label=above: $x$] {};
    \node (y) at (1,4) [label=above: $y$] {};
    \node (z) at (0,3) [red, label=right: $z$] {};

    \draw[very thick] (x)--(y);
    \draw[very thick, red] (x)-- (z)-- (y);
    \draw[very thick, red, dashed] (v) -- (z) node [draw=white, right, midway, fill, opacity =0, text opacity=1]{$k-1$};
    \draw[very thick, dashed] (v) -- (x) node [draw=white, left, midway, fill, opacity =0, text opacity=1]{$k$};
    \draw[very thick, dashed] (v) -- (y) node [draw=white, right, midway, fill, opacity =0, text opacity=1]{$k$};

    \end{scope}

    \begin{scope}[xshift = 7cm, yshift= 6cm]
     \tikzstyle{every node}=[draw, circle,fill=black,minimum size=4pt,
                            inner sep=0pt]

    \node (v) at (0,0) [label=below: $v$] {};
    \node (u) at (0,4) [label=above: $u$] {};
    \node (w) at (-1,3) [label=left: $x$] {};
    \node (w') at (1,3) [label=right: $y$] {};
    \node (z) at (0,2) [red, label=right: $z$] {};
    \node (lab) at (0, -1)[draw=white, fill=white, opacity =0, text opacity=1] {$\mathrm{QC}(u,v)$};

    \draw[very thick] (u)--(w);
    \draw[very thick] (u)--(w');
    \draw[very thick, red] (w') -- (z) -- (w);
    \draw[very thick, dashed, red] (v) -- (z) node [draw=white, right, midway, fill, opacity =0, text opacity=1]{$k-1$};
    \draw[very thick, dashed] (v) -- (w) node [draw=white, left, midway, fill, opacity =0, text opacity=1]{$k$};
    \draw[very thick, dashed] (v) -- (w') node [draw=white, right, midway, fill, opacity =0, text opacity=1]{$k$};

    \end{scope}

        \begin{scope}
     \tikzstyle{every node}=[draw, circle,fill=black,minimum size=4pt,
                            inner sep=0pt]

    \node (v) at (0,0) [label=below: $v$] {};
    \node (u) at (0,4) [label=above: $u$] {};
    \node (w) at (-1,3) [label=left: $w$] {};
    \node (w') at (1,3) [label=right: $w'$] {};
    \node (lab) at (0, -1)[draw=white, fill=white, opacity =0, text opacity=1] {$\INCz(u,v)$};

    \draw[very thick] (u)--(w);
    \draw[very thick] (u)--(w');
    \draw[very thick, red] (w)-- (w');
    \draw[very thick, dashed] (v) -- (w) node [draw=white, left, midway, fill = white, opacity =0, text opacity=1]{$k$};
    \draw[very thick, dashed] (v) -- (w') node [draw=white, right, midway, fill=white, opacity =0, text opacity=1]{$k$};
    \end{scope}

    \begin{scope}[xshift=5cm]
     \tikzstyle{every node}=[draw, circle,fill=black,minimum size=4pt,
                            inner sep=0pt]

    \node (v) at (0,0) [label=below: $v$] {};
    \node (u) at (0,4) [label=above: $u$] {};
    \node (w) at (-1,3) [label=left: $w$] {};
    \node (w') at (1,3) [label=right: $w'$] {};
    \node (z) at (0,2) [red, label=right: $z$] {};
    \node (lab) at (0, -1)[draw=white, fill=white, opacity =0, text opacity=1] {$\INC(u,v)$};

    \draw[very thick] (u)--(w);
    \draw[very thick] (u)--(w');
    \draw[very thick, red] (w)-- (w') -- (z) -- (w);
    \draw[very thick, dashed, red] (v) -- (z) node [draw=white, right, midway, fill, opacity =0, text opacity=1]{$k-1$};
    \draw[very thick, dashed] (v) -- (w) node [draw=white, left, midway, fill, opacity =0, text opacity=1]{$k$};
    \draw[very thick, dashed] (v) -- (w') node [draw=white, right, midway, fill, opacity =0, text opacity=1]{$k$};

    \end{scope}

    \begin{scope}[xshift=10cm]
     \tikzstyle{every node}=[draw, circle,fill=black,minimum size=4pt,
                            inner sep=0pt]

    \node (v) at (0,0) [label=below: $v$] {};
    \node (u) at (0,4) [label=above: $u$] {};
    \node (w) at (-1,3) [] {};
    \node (w') at (1,3) [] {};
    \node (z) at (0,2) [red, label=right: $z$] {};
    \node (w'') at (0.33,3) [] {};
    \node (w''') at (-0.33,3) [] {};

    \node (lab) at (0, -1)[draw=white, fill=white, opacity =0, text opacity=1] {$\INCp(u,v)$};

    \draw[very thick] (u)--(w);
    \draw[very thick] (u)--(w');
    \draw[very thick] (u)--(w'');
    \draw[very thick] (u)--(w''');

    \draw[very thick, red] (z)--(w'');
    \draw[very thick, red] (z)--(w''');

    \draw[very thick, red] (w) -- (w''') -- (w'') -- (w') -- (z) -- (w);
    \draw[very thick, red] (w) to[bend right] (w'');
    \draw[very thick, red] (w''') to[bend right] (w');
    \draw[very thick, red] (w) to[bend right] (w');

    \draw[very thick, dashed] (v) -- (w'') node [draw=white, right, midway, fill, opacity =0, text opacity=1]{$k$};
     \draw[very thick, dashed] (v) -- (w''') node [draw=white, left, midway, fill, opacity =0, text opacity=1]{$k$};
    \draw[very thick, dashed] (v) -- (w) node [draw=white, left, midway, fill, opacity =0, text opacity=1]{$k$};
    \draw[very thick, dashed] (v) -- (w') node [draw=white, right, midway, fill, opacity =0, text opacity=1]{$k$};

    \draw[very thick, dashed, red] (v) -- (z) node [draw=white, right, midway, fill, opacity =0, text opacity=1]{$k-1$};

    \end{scope}

    \end{tikzpicture}
    \caption{Illustration of the Triangle and  Quadrangle conditions and the variants of the Interval Neighborhood Condition. The red vertices and edges are implied by these conditions.}  \label{fig:TCQC}
\end{figure}
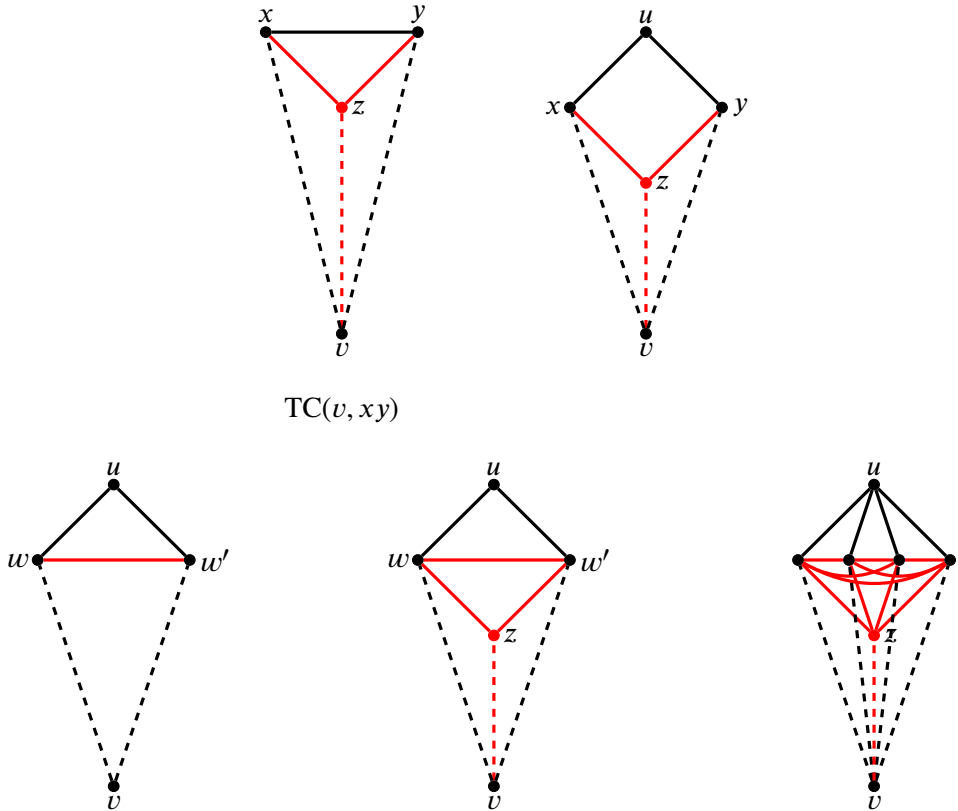

We define the \emph{Interval Neighborhood Condition} $\INCz$ and its stronger versions (see the lower row of Figure \ref{fig:TCQC}):
\begin{itemize}
 \item $\big(\INCz\big)$: for any two distinct vertices $u,v\in V$, the neighbors of $u$ in $I(u,v)$ form a clique.
 \item $\big(\INC\big)$: $G$ satisfies $\INCz$ and for any
two distinct vertices $u,v$ and any $w,w'\in S_1(u) \cap I(u,v)$  there exists a vertex  $z\in I(w,v)\cap I(w',v)$ such that $z\sim w,w'$.
\item $\big(\INCp\big)$: $G$ satisfies $\INCz$ and for any
two vertices $u,v$ with $d(u,v)=k\ge 2$ there exists a vertex $z\in I(u,v)$ at distance $k-2$ from $v$ and adjacent to all vertices of
the clique $S_1(u) \cap I(u,v)$.
\end{itemize}
Note that  $\mathrm{INC}^0$ occurs in the characterization (iii) of CB-graphs in Theorem \ref{thm: well-bridged} and that, by an easy observation,  $\INCz$ is equivalent to $2$-convexity of balls of $G$.
Note also  that $u$ and $v$ do not play the same role in the above definitions. To make clear how we use these conditions, we will write that $G$ satisfies $\INCz(u,v)$ (respectively $\INC(u,v)$ or $\INCp(u,v)$)
when the condition $\INCz$ (respectively $\INC$ or $\INCp$) is satisfied for the pair of vertices $u$ and $v$. We will use the notation $\INCz(v)$ (respectively $\INC(v)$ or $\INCp(v)$) when $\INC(u,v)$ (respectively $\INC(u,v)$ or $\INCp(u,v)$) is satisfied for every $u\in V\setminus \sg{v}$. Finally for every $k\geq 1$, we will use the notation $\INC_{\leq k}$ for the ``local version of $\INC$'', whenever the condition $\INC(u,v)$ holds for every pair of vertices $u,v$ such that $d(u,v)\leq k$.

We will also consider the following stronger versions of $\PCz$ (see Figure \ref{fig:TPC}):
\begin{itemize}
\item $\left(\PCu\right)$: for any three vertices $v,x,y$ such that $d(v,x)=d(v,y)=k\geq 2$ and $x\sim y$ and for any neighbor $w\in I(x,v)$ of $x$ there exists a neighbor $w'\in I(y,v)$ of $y$ and $z \in B_{k-2}(v)$ such that $xwzw'y$ is a pentagon of $G$.
\item $\left(\PCd\right)$: $\left(\PCu\right)$ holds and moreover for
  any $v,x,y$ such that $d(v,x)=d(v,y)=k\geq 2$ and $x\sim y$, we
  have: $B_2(x)\cap B_{k-2}(v)= B_2(y)\cap B_{k-2}(v)$.
\item $\left(\PCp\right)$: for any three vertices $v,x,y$ such that $d(v,x)=d(v,y)=k\geq 2$ and $x\sim y$, there exist a neighbor $w'\in I(y,v)$ of $y$ and $z \in B_{k-2}(v)$ such that for all neighbors $w\in I(x,v)$ of $x$, $xwzw'y$ is a pentagon of $G$.
\end{itemize}

\tikzexternaldisable
  \begin{figure}[h]
    \centering
    \begin{tikzpicture}[scale=1]

    \begin{scope}[]
     \tikzstyle{every node}=[draw, circle,fill=black,minimum size=4pt,
                            inner sep=0pt]
    \node (lab) at (0, -1)[draw=white, fill=white] {$\PCz(v, xy)$};
    \node (v) at (0,0) [label=below: $v$] {};
    \node (x) at (-1,4) [label=above: $x$] {};
    \node (y) at (1,4) [label=above: $y$] {};
    \node (w) at (-1,3) [red, label=left: $w$] {};
    \node (w') at (1,3) [red, label=right: $w'$] {};
    \node (z) at (0,2) [red, label=right: $z$] {};

    \draw[very thick] (x)--(y);
    \draw[very thick, red] (x)-- (w) -- (z)-- (w') -- (y);
    \draw[very thick, red, dashed] (v) -- (z) node [draw=white, right, midway, fill, opacity =0, text opacity=1]{$k-2$};
    \draw[very thick, red, dashed] (v) -- (w) node [draw=white, left, midway, fill, opacity =0, text opacity=1]{$k-1$};
    \draw[very thick, red, dashed] (v) -- (w') node [draw=white, right, midway, fill, opacity =0, text opacity=1]{$k-1$};
    \end{scope}

    \begin{scope}[xshift=4cm]
     \tikzstyle{every node}=[draw, circle,fill=black,minimum size=4pt,
                            inner sep=0pt]
    \node (lab) at (0, -1)[draw=white, fill=white] {$\PCu(v, xy)$};
    \node (v) at (0,0) [label=below: $v$] {};
    \node (x) at (-1,4) [label=above: $x$] {};
    \node (y) at (1,4) [label=above: $y$] {};
    \node (w) at (-1,3) [label=left: $w$] {};
    \node (w') at (1,3) [red, label=right: $w'$] {};
    \node (z) at (0,2) [red, label=right: $z$] {};

    \draw[very thick] (w) -- (x)--(y);
    \draw[very thick, red] (w) -- (z)-- (w') -- (y);
    \draw[very thick, red, dashed] (v) -- (z) node [draw=white, right, midway, fill, opacity =0, text opacity=1]{$k-2$};
    \draw[very thick, dashed] (v) -- (w) node [draw=white, left, midway, fill, opacity =0, text opacity=1]{$k-1$};
    \draw[very thick, red, dashed] (v) -- (w') node [draw=white, right, midway, fill, opacity =0, text opacity=1]{$k-1$};
    \end{scope}

    \begin{scope}[xshift=8cm]
     \tikzstyle{every node}=[draw, circle,fill=black,minimum size=4pt,
                            inner sep=0pt]
    \node (lab) at (0, -1)[draw=white, fill=white] {$\PCd(v, xy)$};
    \node (v) at (0,0) [label=below: $v$] {};
    \node (x) at (-1,4) [label=above: $x$] {};
    \node (y) at (1,4) [label=above: $y$] {};
    \node (w) at (-1,3) [label=left: $w$] {};
    \node (w') at (1,3) [red, label=right: $w'$] {};
    \node (z) at (0,2) [label=right: $z$] {};

    \draw[very thick] (z) -- (w) -- (x)-- (y);
    \draw[very thick, red] (z)-- (w') -- (y);
    \draw[very thick, red, dashed] (v) -- (z) node [draw=white, right, midway, fill, opacity =0, text opacity=1]{$k-2$};
    \draw[very thick, dashed] (v) -- (w) node [draw=white, left, midway, fill, opacity =0, text opacity=1]{$k-1$};
    \draw[very thick, red, dashed] (v) -- (w') node [draw=white, right, midway, fill, opacity =0, text opacity=1]{$k-1$};
    \end{scope}

    \begin{scope}[xshift=12cm]
     \tikzstyle{every node}=[draw, circle,fill=black,minimum size=4pt,
                            inner sep=0pt]
    \node (lab) at (0, -1)[draw=white, fill=white] {$\PCp(v, xy)$};
    \node (v) at (0,0) [label=below: $v$] {};
    \node (x) at (-1,4) [label=above: $x$] {};
    \node (y) at (1,4) [label=above: $y$] {};
    \node (w) at (-1,3) [] {};
    \node (w'') at (-1.5,3) [] {};
    \node (w''') at (-0.5,3) [] {};
    \node (w') at (1,3) [red, label=right: $w'$] {};
    \node (z) at (0,2) [red, label=right: $z$] {};

    \draw[very thick] (z) -- (w) -- (x)-- (y);
    \draw[very thick] (z) -- (w'') -- (x)-- (w''') -- (z);
    \draw[very thick, red] (z)-- (w') -- (y);
    \draw[very thick, red, dashed] (v) -- (z) node [draw=white, right, midway, fill, opacity =0, text opacity=1]{$k-2$};
    \draw[very thick, dashed] (v) -- (w) node [draw=white, left, midway, fill, opacity =0, text opacity=1]{$k-1$};
    \draw[very thick, dashed] (v) -- (w'') node [draw=white, left, midway, fill, opacity =0, text opacity=1]{};
    \draw[very thick, dashed] (v) -- (w''') node [draw=white, left, midway, fill, opacity =0, text opacity=1]{};
    \draw[very thick, red, dashed] (v) -- (w') node [draw=white, right, midway, fill, opacity =0, text opacity=1]{$k-1$};
    \end{scope}

    \end{tikzpicture}
    \caption{Illustration of the four variants of the Pentagon Condition. The red vertices and edges are implied by these conditions.}  \label{fig:TPC}
\end{figure}
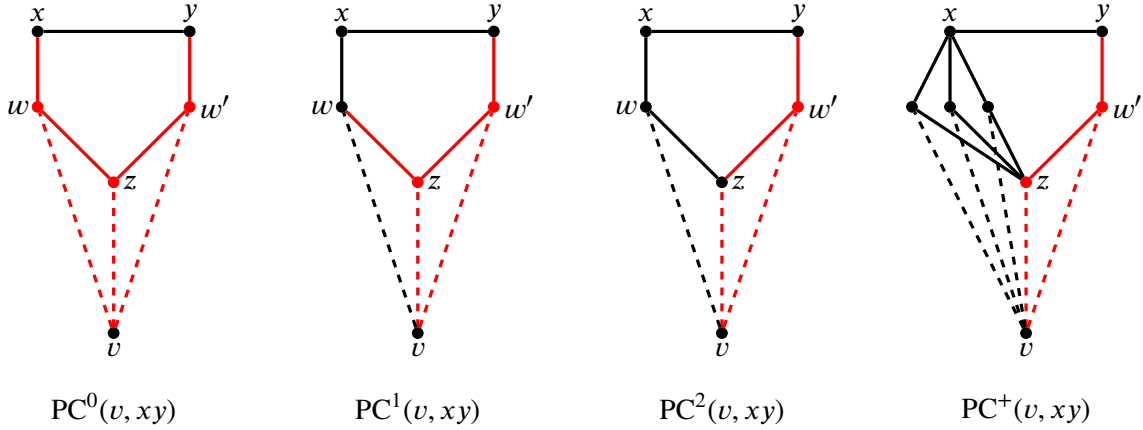

Recall that a graph $G$ is called \emph{weakly modular} \cite{BaCh95, Ch_metric} if it satisfies $\TC$ and $\mathrm{QC}$. For any $i\in \sg{0,1,2,+}$, it will be useful to write $\PC^i(v,xy)$ if the condition $\PC^i$ is satisfied for the vertex $v$ and edge $xy$.
We say that a graph $G$ satisfies the \emph{Triangle-Pentagon Condition}  $\TPC^i$ if for any three vertices $v,x,y$ such that $d(v,x)=d(v,y)=k$ and $x\sim y$ either $\TC(v,xy)$ or $\PC^i(v,xy)$ holds. Thus we will write $\TPC^i(v,xy)$ when either $\TC(v,xy)$ or $\PC^i(v,xy)$ holds. We will use the notation $\TPC^i(v)$ (respectively $\TC(v)$, $\PC^i(v)$) whenever $\TPC^i(v,xy)$ (respectively $\TC(v,xy)$, $\PC^i(v,xy)$) holds for $v$ and any edge $xy$ at uniform distance from $v$.
Finally for every $k\geq 1$, we will use $\TPC^i_{\leq k}$ to denote the ``local version of $\TPC^i$, i.e. whenever $\TPC^i(v,xy)$ holds for every $v,x,y$ such that the edge $xy$ is at uniform distance at most $k$ from $v$.

\subsection{Graphs with $2$-convex balls}
We start with a simple characterization of graphs with $2$-convex balls.

\begin{theorem}\label{thm: INC}
For a graph $G=(V,E)$, the following conditions are equivalent:
\begin{itemize}
  \item[(i)] $G$ has $2$-convex balls;
  \item[(ii)]  $G$ satisfies $\INC$;
  \item[(iii)] $G$ satisfies $\INCp$.
\end{itemize}
\end{theorem}

\begin{proof}
  By definition, a graph satisfying $\INC$ has $2$-convex
  balls, thus (ii)$\Rightarrow$(i). Suppose now that $G$ has $2$-convex balls and we will show that
  $G$ satisfies $\INC$.   Pick any vertices
  $u,v$ at distance $k+1\ge 2$ and two neighbors $w,w'$ of $u$ in the
  interval $I(u,v)$.  We proceed by induction on
  $k=d(v,u)-1\geq1$. The $2$-convexity of $B_k(v)$ implies
  that $w\sim w'$. It remains to show that $w$ and $w'$ have a common
  neighbor $z$ at distance $k-1$ from $v$.  If $k=1$, we can set
  $z:=v$. If $k\geq 2$, we may assume thanks to the induction
  hypothesis that $I(v,w)\cap I(v,w')=\sg{v}$. Let
  $a\in I(v,w), b\in I(v,w')$ such that $a,b\sim v$. The $2$-convexity
  of the ball $B_k(u)$ implies that $a\sim b$. Let $c$ be a neighbor
  of $w'$ in $I(b,w')$. Then $d(a,w)=k-1$ and $d(a,c)\leq k-1$.  Since
  $I(v,w)\cap I(v,w')=\sg{v}$, we also have $d(a,w')=k$. By
  $2$-convexity of the ball $B_{k-1}(a)$, we conclude that $c\sim w$,
  which contradicts the assumption that $I(v,w)\cap
  I(v,w')=\sg{v}$. This establishes that (i)$\Rightarrow$(ii).

  Clearly, any graph satisfying $\INCp$ satisfies $\INC$, thus (iii)$\Rightarrow$(ii).
  Therefore, it suffices to show that (ii)$\Rightarrow$(iii).
Let $G$ be a graph that satisfies  $\INC$. To show that $G$ satisfies $\INCp$, we use a
  maximality argument: let $u,v\in V$ be two distinct vertices at distance $k+1:=d(u,v)\geq 2$ and assume by way of contradiction that there is no
  vertex at distance $k-1= d(v,u)-2$ from $v$ which is adjacent to all
  vertices of the clique $C:=S_1(u) \cap I(u,v)$.
  Choose $z\in I(v,u)$ at distance $k-1$ from $v$ which is adjacent to a
  maximum number of vertices of $C$. Denote by $C'$ the set of
  neighbors of $z$ in $C$.  By $\INC(u,v)$, $C'$ contains at least
  two vertices, unless $|C|\leq 1$, in which case we are trivially done. On the other hand, there is a vertex
  $x \in C\setminus C'$. Pick any $y \in C'$. By $\INC(u,v)$,
  there exists a vertex $z'\sim x,y$ at distance $k-1$ from
  $v$. Applying $\INCz(y,v)$, we conclude that $z\sim z'$.
  We assert that $z'$ is adjacent to all vertices of $C'$. Indeed, if
  $y'\in C'$, then $zy'xz'$ is a $4$-cycle, which cannot be
  induced. Since $z\nsim x$, necessarily $z'\sim y'$. Consequently,
  $z'$ is adjacent to $x$ and all vertices of $C'$, contrary to the
  maximality choice of $z$. Thus we proved that $G$ satisfies $\INCp$, establishing
  (ii)$\Rightarrow$(iii).  \end{proof}

\begin{remark} A graph with 2-convex balls has no isometric cycles of even length $2k$:  if $G$ contains such a cycle $C$, then for any vertex $v$ of $C$ the
ball $B_{k-1}(v)$ is not 2-convex.  Indeed, if $u$ is the vertex of $C$ antipodal to $v$ and $x,y$ are its neighbors in $C$, then $x,y\in B_{k-1}(v)$ and $u\notin B_{k-1}(v)$.
The converse however is not true: let $G$ be obtained by gluing two 5-cycles along an edge $e$. Then  the unique isometric cycles of $G$ are the two 5-cycles.
Also $G$ has diameter 4 and the unique diametral pair is the pair $u,v$, where $u$ and $v$ are the vertices of the two 5-cycles opposite to $e$. Then the ball $B_3(v)$  is not 2-convex because
it contains the neighbors of $u$ but not the vertex $u$.
\end{remark}

\subsection{Graphs with convex balls}

Now, we show that in full analogy with weakly modular graphs, the conditions $\INC$ and $\TPCz$ as well as
their stronger versions characterize the CB-graphs.
\begin{theorem}
\label{thm: TPC}
For a graph $G$, the following conditions are equivalent:
 \begin{itemize}
  \item[(i)] $G$ has convex balls;
  \item[(ii)] $G$ satisfies $\INC$ and $\TPCz$;
  \item[(iii)] $G$ satisfies $\INC$ and $\TPCu$;
  \item[(iv)] $G$ satisfies $\INC$ and $\TPCd$.
\item[(v)] $G$ satisfies $\INCp$ and $\TPCp$.
  \end{itemize}
\end{theorem}

\begin{proof}
We first prove the implication $(i) \Rightarrow (iii)$. Let $G$ be a CB-graph. By Theorem \ref{thm: INC}, $G$ satisfies $\INC$. We show that it also satisfies $\TPCu$: let $v \in V$ be a vertex of $G$ and $x, y \in S_k(v)$ such that $x \sim y$. We may assume that $k\geq 2$. Suppose  that $\TC(v,xy)$ does not hold, i.e.,  there is no $z \in B_{k-1}(v)$ such that $z \sim x, y$. We assert that in this case $\PCu(v,xy)$ holds. Let $w$ be any neighbor of $x$ in $I(x,v)$. If $k=2$, then we can set $z:=v$ and take as $w'$ any common neighbor of $y$ and $v$.
Let now $k \geq 3$ and proceed by induction on $k$. Let $a$ be a neighbor of $v$ in $I(w,v)$. If $a\in I(v,y)$, then $d(a,x)=d(a,y)=k-1$ and we can apply the induction hypothesis to the vertex $a$ and the edge $xy$. Therefore, $d(a,y)=k$. Let $w'$ be a neighbor of $y$ in $I(y,v)$. Since $\TC(v,xy)$ does not hold, $w\ne w'$. Since $w,w'\in B_{k-1}(v)$ and $x,y\notin B_{k-1}(v)$, by convexity of  $B_{k-1}(v)$, $(w,x,y,w')$ cannot be a shortest path. Thus $d(w,w')\le 2$. If $w \sim w'$, then we obtain a 4-cycle $xww'y$, which cannot be induced. Hence, either $w \sim y$ or $w'\sim y$, contrary to the assumption that $\TC(v,xy)$ does not hold. Therefore, $d(w,w')=2$. Now, let $c$ be a neighbor of $a$ in $I(a,w)$  (if $k=3$, then $c=w$). Then $d(w',c)\le d(w',w)+d(w,c)=2+k-3=k-1$. Since $d(v,w')=k-1$ and $a\in I(c,v)$, the convexity of the ball $B_{k-1}(w')$ implies that $d(w',a)\leq k-1$. Since $w'\sim y$ and $w'\nsim x$, we have $y\in I(w',x)$.
Since $d(a,y)=k$ and $x,w'\in B_{k-1}(a)$ we obtain a contradiction with the convexity of the ball $B_{k-1}(a)$.  This establishes $\TPCu$.

The implication $(iii)\Rightarrow (ii)$ is trivial. Now we prove the
implication $(ii) \Rightarrow (i)$. Let $G$ be a graph satisfying
$\INC$ and $\TPCz$. By Theorem \ref{thm: well-bridged} $(iii)$, it
suffices to show that $G$ has no isometric cycles of length different
from $3$ and $5$ (as $\INCz$ follows immediately from $\INC$). Let
$C$ be an isometric cycle of $G$ of length $n$. By $\INCz$, $n$ must be
odd, say $n=2k+1\geq 7$ and let
$C=(v_0,v_1,\ldots,v_k,v_{k+1},\ldots,v_{2k})$.  Pick $v=v_0$ of $C$
and let $x=v_k, y=v_{k+1}$, i.e., $xy$ is the (unique) edge of $C$
opposite to $v$. Then $d(v,x)=d(v,y)=k$.  Let $w=v_{k-1}$ and
$w'=v_{k+2}$ be the second neighbors of $x$ and $y$ (respectively) in
$C$. Since $C$ is an isometric cycle, $d(v,w)=d(v,w')=k-1$ and
$d(w,w')=3$.  By $\TPCz(v,xy)$, there
exists a vertex $z\in I(x,v)\cap I(y,v)$ which has distance 1 or 2 to
both vertices $x$ and $y$. If $z\sim x,y$, then $w,z\in I(x,v)$ and
$w',z\in I(y,v)$, thus by $\INCz$, either $z$ is adjacent to both
$w,w'$ or $z$ coincides with one of the vertices $w,w'$ and is
adjacent to the second one. In both cases, we conclude that
$d(w,w')\le 2$, which is impossible. Thus $d(z,x)=d(z,y)=2$ and
$d(z,v)=k-2$ and we can suppose that there is no vertex in
$I(x,v)\cap I(y,v)$ adjacent to $x$ and $y$. 

First suppose that $z$ is adjacent to one of the vertices $w,w'$, say $z\sim w$. Since $d(w,w')=3$, $z\nsim w'$. Let $t'$ be a common neighbor of $z$ and $y$.
Let $p=v_{k-2}$. Since $z,p$ belong to $I(w,v)$ and are both adjacent to $w$, by $\INCz$, the vertices $z$ and $p$ are adjacent or coincide. If $z=p$,
then $d(p,y)=2$, which is impossible because $y=v_{k+1}$ and $p=v_{k-2}$ and whence $d(p,y)=3$. Hence $z\sim p$.  Since $x,t'\in I(y,p)$ and are both adjacent to $y$,
by $\INCz$, $x\sim t'$. Since $d(t',v)=k-1$, $t'$ belongs to $I(x,v)\cap I(y,v)$ and is adjacent to $x$ and $y$, contrary to our assumption.

Thus, we can suppose that any vertex of $I(x,v)\cap I(y,v)$ at  distance 2 from $x$ and $y$ is not adjacent to any of the vertices $w$ and $w'$. Let $t\ne  w$
be a common neighbor of $x$ and $z$ and $t'\ne w'$ be a common neighbor of $y$ and $z$. By our assumption, $t\nsim y$ and $t'\nsim x$. Since $w,t\in I(x,v)$ and $x\sim w,t$, by $\INC$, $w\sim t$ and there exists a common neighbor $q$ of $w$ and $t$ at distance $k-2$ from $v$. Since $q,z \in I(t,v)$, by $\INCz$, $q \sim z$. Then $d(q,x)=2$ because $w\sim x,q$ and $d(q,t')=2$ because $z\sim q,t'$. Therefore, if $d(q,y)=3$, then $x,t'\in I(y,q)$ and $x,t'\sim y$, thus  by $\INCz$ the vertices $x$ and $t'$ must be adjacent.
Since we proved that $x\nsim t'$, we conclude that $d(y,q)=2$. Since $d(q,v)=k-2$, $q$ is a vertex of $I(x,v)\cap I(y,v)$ at distance 2 from $x$ and $y$ and adjacent to $w$, contrary to our assumption
that such vertex does not exists. This concludes the proof of the implication $(ii) \Rightarrow (i)$ and proves the equivalence $(i)\Leftrightarrow (ii)\Leftrightarrow (iii)$.

The implication $(iv)\Rightarrow (iii)$ is immediate. We now prove $(iii)\Rightarrow (iv)$. We only have to show that if $v,x,y \in V$ are such that $d(v,x)=d(v,y)=k\geq 2$, $x\sim y$ and $\TC(v,xy)$ does not hold, then $\PCd(v,xy)$ holds. We let $z\in B_2(x)\cap B_{k-2}(v)$ and $w\sim x,z$. We will show that $z \in B_2(y)\cap B_{k-2}(v)$. As we assumed $(iii)$, $\PCu(v,xy)$ holds, so there exist $z', w'\in V$ such that $xwz'w'y$ form a pentagon and $z' \in B_{k-2}(v)$. If $z=z'$, then we are immediately done. Otherwise by $\INCz(w,v)$, $z'\sim z$. In particular we have $x,w'\in B_2(z)$. As $xwz'w'y$ is a pentagon, $x\nsim w'$. Thus by $\INCz(z)$ we must have $d(z,y)\leq 2$, which is the desired result. We proved that $B_2(x)\cap B_{k-2}(v)\subseteq B_2(y)\cap B_{k-2}(v)$, and by symmetry we are done.

Since the implication $(v)\Rightarrow (iii)$ is trivial, it remains to
show that $(iv) \Rightarrow (v)$. By Theorem~\ref{thm: INC}, $G$
satisfies $\INCp$.  Let $v \in V$ be a vertex of $G$ and
$x, y \in S_k(v)$ such that $x \sim y$. Suppose that $\TC(v,xy)$ does
not hold.  By $\INCp$, there exists a vertex $z$ at distance $k-2$
from $v$ adjacent to all vertices in $S_1(x)\cap I(x,v)$. By $(iv)$,
$d(z,y) = 2$ and thus there exists $w' \sim y,z$ such that $xwzw'y$ is
a pentagon for every $w \in S_1(x)\cap I(x,v)$.
This establishes $\TPCp$ and concludes the
proof of the theorem.
\end{proof}

\begin{remark}
 Observe that to prove the implication $(i)\Rightarrow (iii)$, we only needed to use that $G$ has $3$-convex balls. Thus we also get an alternative proof of the equivalence between items $(i)$ and $(ii)$ of Theorem \ref{thm: well-bridged}.
\end{remark}

\begin{remark}\label{rem: PC+_tight}
The condition $\TPCu$ is tight, in the sense that we cannot require $z \in B_{k-2}(v)$ to be adjacent to arbitrary $w\in I(x,v), w\sim x$ and $w'\in I(y,v), w'\sim y$. The example drawn in Figure \ref{fig: ctreex} illustrates this. To see that the described graph has convex balls, note that the only pairs of vertices at distance $>2$  are $(v,x)$ and $(v,y)$. From this remark one can easily deduce that balls of radius at least $2$ are convex. On the other hand, balls of radius $1$ are also convex because $G$ does not contain induced $C_4$.

  \tikzexternaldisable
  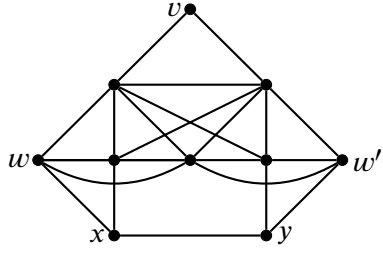
\begin{figure}[h]
    \centering
    \begin{tikzpicture}[scale=1]
    \tikzstyle{every node}=[draw,circle,fill=black,minimum size=4pt,
                            inner sep=0pt]

    \node (v) at (-1,0) [label=left: $x$] {};
    \node (w) at (1,0) [label=right: $y$] {};
    \node (11) at (-2,1) [label=left: $w$] {};
    \node (12) at (-1,1) {};
    \node (13) at (0,1) {};
    \node (14) at (1,1) {};
    \node (15) at (2,1) [label=right: $w'$] {};
    \node (21) at (-1,2) {};
    \node (22) at (1,2) {};
    \node (u) at (0,3) [label=left: $v$] {};
    \draw[thick] (v) -- (11) -- (12) -- (13) -- (11) -- (21) -- (22) -- (12) -- (21) -- (13) -- (14) -- (21) -- (u) -- (22) -- (14) -- (w) -- (15) -- (22) -- (13);
    \draw[thick] (w) -- (v) -- (12);
    \draw[thick] (14) -- (15);
    \draw[thick] (11) to[bend right] (13);
    \draw[thick] (13) to[bend right] (15);
    \end{tikzpicture}
    \caption{A graph with convex balls and $k:=d(v,x)=d(v,y)=3$. $\TC(v,xy)$ does not apply and for vertices $w\in I(x,v), w\sim x$ and $w'\in I(y,v), w'\sim y$ there is no $z \in B_{k-2}(v)$ adjacent to $w$ and $w'$.}  \label{fig: ctreex}
\end{figure}
\end{remark}

\subsection{Pairs of intersecting pentagons}
As observed earlier,  two pentagons glued along one common edge do not define a CB-graph. In particular, such pairs of pentagons cannot occur as isometric subgraphs
of CB-graphs. In this subsection, we metrically describe the subgraphs of CB-graphs induced by pairs of pentagons intersecting in at least two vertices.  If $\pi=abcde$ is a pentagon of a graph $G$, then we call a vertex $x$ a \emph{universal vertex of $\pi$} if $x$ is adjacent to all the  vertices of $\pi$.  

We denote by $\PT$ the graph obtained by gluing a pentagon and a
triangle along a common edge (see Figure~\ref{fig: PTPP},
left). Similarly we denote by $\PPu$ the graph obtained by gluing two
pentagons together along a common edge (see Figure~\ref{fig: PTPP},
center) and by $\PPd$ the graph obtained by gluing two pentagons
together along a common path of length 2 (see Figure~\ref{fig: PTPP},
right).

\tikzexternaldisable
  \begin{figure}[h]
    \centering
    \begin{tikzpicture}[scale=1]

    \begin{scope}[xshift = 0cm]
     \tikzstyle{every node}=[draw, circle,fill=black,minimum size=4pt,
                            inner sep=0pt]
    \node (lab) at (0, -1.5)[draw=white, fill=white] {$\PT$};
    \node (1) at (0:1) {};
    \node (2) at (60:1) {};
    \node (3) at (120:1) {};
    \node (4) at (180:1) {};
    \node (5) at (240:1) {};
    \node (6) at (300:1) {};

    \draw[very thick] (1)--(2)--(3)--(4) --(5) --(6)--(1);
    \draw[very thick] (3)--(5);
    \end{scope}

    \begin{scope}[xshift = 6cm]
     \tikzstyle{every node}=[draw, circle,fill=black,minimum size=4pt,
                            inner sep=0pt]
    \node (lab) at (-0.5, -1.5)[draw=white, fill=white] {$\PPu$};
    \node (a) at (-1.75,1) {};
    \node (b) at (-2.5,0) {};
    \node (c) at (-1.75,-1) {};
    \node (d) at (-0.5,-0.6) {};
    \node (e) at (-0.5,0.6) {};
    \node (f) at (0.75,1) {};
    \node (g) at (1.5,0) {};
    \node (h) at (0.75,-1) {};

    \draw[very thick] (a)--(b)--(c)--(d)--(e)--(a);
    \draw[very thick] (e)--(f)--(g)--(h)--(d);
    \end{scope}

        \begin{scope}[xshift = 11cm]
     \tikzstyle{every node}=[draw, circle,fill=black,minimum size=4pt,
                            inner sep=0pt]
    \node (lab) at (0, -1.5)[draw=white, fill=white] {$\PPd$};
    \node (1) at (30:1) {};
    \node (2) at (90:1) {};
    \node (3) at (150:1) {};
    \node (4) at (210:1) {};
    \node (5) at (270:1) {};
    \node (6) at (330:1) {};
    \node (7) at (0:0) {};

    \draw[very thick] (1)--(2)--(3)--(4) --(5) --(6)--(1);
    \draw[very thick] (2)--(7)--(5);

    \end{scope}

    \end{tikzpicture}
    \caption{The graphs $\PT$, $\PPu$ and $\PPd$.}  \label{fig: PTPP}
\end{figure}

\begin{lemma}
\label{lem: deuxpent} Let $G$ be a graph in which $\TPCz_{\leq3}$, $\INC_{\leq 3}$ hold and such that every induced subgraph isomorphic to $\PPu$ has diameter at most $3$.
Let $\pi_1$ and $\pi_2$ be two pentagons of $G$ intersecting in at least
two vertices. Then either $\diam(\pi_1\cup \pi_2)=2$ or one of the pentagons $\pi_1$ or $\pi_2$ has a universal vertex. \end{lemma}

\begin{proof}
Observe that if we look back at the proof of Theorem \ref{thm: TPC}, we can assume that the variants of $\TPCz$ also apply when the distances involved are at most $3$. We also can use the $2$-convexity of balls of radius at most $2$.

If $\pi_1$ and $\pi_2$ intersect in two non-adjacent vertices, then using the $2$-convexity of balls of radius 2, one can conclude that $d(x,y)\le 2$ for any vertex $x\in \pi_1\setminus \pi_2$ and $y\in \pi_2\setminus \pi_1$.

Now, suppose that the pentagons $\pi_1$ and $\pi_2$ intersect in
exactly two adjacent vertices, say $\pi_1=abcde$ and $\pi_2=defgh$.

\begin{claim}\label{claim-PPinduced}
  If $\pi_1\cup \pi_2$ is not induced, then
  $\diam(\pi_1\cup \pi_2)=2$.
\end{claim}

\begin{proof}
  Observe that to establish this claim, it is sufficient to show that
  $d(a,g), d(a,h), d(b,f), d(b,g),$ $ d(b,h), d(c,f), d(c,g) \leq 2$.
  Suppose first that $a \sim f$. Then since $d(a,g) \leq 2 = d(a,d)$,
  by the 2-convexity of $B_2(a)$, we have $d(a,h) \leq 2$. By
  symmetry, we also have $d(b,f), d(c,f) \leq 2$. Since
  $d(c,f), d(c,h) \leq 2$, by the 2-convexity of $B_2(c)$, we have
  $d(c,g) \leq 2$ and similarly, we have $d(b,h) \leq 2$. Since
  $d(b,f), d(b,h) \leq 2$, by the 2-convexity of $B_2(b)$, we have
  $d(b,g) \leq 2$. Consequently, $\diam(\pi_1\cup \pi_2)=2$ if
  $a \sim f$, or for similar reasons, if $c \sim h$. If $a \sim g$,
  then $agfe$ is a 4-cycle that cannot be induced. Since $e\nsim g$,
  necessarily $a \sim f$ and thus $\diam(\pi_1\cup \pi_2)=2$ by the
  previous case. Similarly, if $c \sim g$, $b \sim f$, or $b \sim h$,
  we get $\diam(\pi_1\cup \pi_2)=2$. If $a \sim h$, then $aedh$ is a
  4-cycle that cannot be induced and thus either $a\sim d$ or
  $e \sim h$, which is impossible. For the same reasons, $c \nsim
  f$. Finally, suppose that $b \sim g$. Then
  $d(a,g), d(b,f), d(b,h),d(c,g) \leq 2$. By the 2-convexity of
  $B_2(a)$, we get $d(a,h) \leq 2$ and similarly, we have
  $d(c,f) \leq 2$, establishing that $\diam(\pi_1\cup \pi_2)=2$. This
  ends the proof of the claim. 
\end{proof}

If $\diam(\pi_1\cup \pi_2)=2$, then we are done. By
Claim~\ref{claim-PPinduced}, we can assume that $\pi_1 \cup \pi_2$ is
induced and isomorphic to $\PPu$. Consequently,
$\diam(\pi_1\cup \pi_2)=3$. If $d(b,g)=2$, by convexity of $B_2(b)$
and $B_2(g)$ we will conclude that
$d(b,f)=d(b,h)=d(g,a)=d(g,c)=2$. Then, by $2$-convexity of $B_2(a)$
and $B_2(c)$, we conclude that $d(c,f) = d(a,h) = 2$, whence
$\pi_1\cup \pi_2$ has diameter 2. Thus, $d(b,g)=3$. The $2$-convexity
of $B_2(b)$ implies that $d(b,f)=3$ or $d(b,h)=3$. Analogously, the
$2$-convexity of $B_2(g)$ implies that $d(a,g)=3$ or
$d(c,g)=3$. 

 First suppose that $d(b,f)=d(b,h)=3$. This implies that $d(g,a)=d(g,c)=3$. Indeed, suppose by way if contradiction that $d(g,c)=2$ and $d(g,a)=3$. Then, $e,g\in B_2(c)$ and $f\in I(e,g)$. By $2$-convexity of $B_2(c)$ we conclude that $d(c,f)\le 2$. But this is impossible because $a,c\in B_2(f)$ and $b\in I(a,c)\setminus B_2(f)$, contrary to the $2$-convexity of $B_2(f)$. Consequently, $d(g,a)=d(g,c)=3$. We apply $\TPCu(b,fg)$. If $\TC$ applies, then there exists $z\sim f,g$ such that $d(b,z)=2$. By $\INCz(b)$ this means that $z \sim e$. Hence we get $d(z,d)\leq 2$ and $d(b,z)=2$, so by $2$-convexity of $B_2(z)$ we must have $d(c,z)\leq 2$. Thus by $\INCz(c)$ we get $z\sim h$. We get the last adjacency $z\sim d$ because the $4$-cycle $zedh$ cannot be induced, showing that $z$ is universal for pentagon $\pi_2=defgh$. If $\PCu$ applies for neighbor $e$ of $f$, then there exists $z \sim b, e$ with $d(z,g)=2$. Then the $4$-cycle $abze$ cannot be induced and we get $z\sim a$. As $d(z,g)=2$ and $d(z,d)\leq 2$, $2$-convexity of $B_2(z)$ implies that $d(z,h)=2$. Thus by $\INCz(h)$ we must have $z \sim c$. As the $4$-cycle $zbcd$ cannot be induced we get $z\sim d$, so $z$ is universal for pentagon $\pi_1=abcde$.

 Therefore, we can suppose without loss of generality that $d(b,f)=3$ and $d(b,h)=2$. By analogy, we also conclude that one of the distances $d(g,a),d(g,c)$ is 3 and another one is 2. First, let $d(g,a)=3$ and $d(g,c)=2$  and apply $\TPCu(b,fg)$. If $\TC$ applies, then there exists $z\sim f,g$ such that $d(b,z)=2$. By $\INCz(b)$ this means that $z \sim e$ and $z\sim h$. As the $4$-cycle $ezhd$ cannot be induced, we conclude that $z\sim h$, showing that $z$ is universal for $\pi_2$.  If $\PCu$ applies with respect to the neighbor $e$ of $f$, then there exists $z \sim b, e$ with $d(z,g)=2$. Again, when considering the $4$-cycle $aezb$ we get $z\sim a$. Then $\INCz(g)$ implies that $z \sim c$. Since the $4$-cycle $zcde$ cannot be induced, so $z\sim d$, whence $z$ is universal for $\pi_1$.

 By previous cases and symmetric arguments, it remains to deal with the case where $d(b,f)=d(b,g)=d(c,g)=3$ and $d(b,h)=d(a,g)=2$.
 First note that $d(c,f)=3$, otherwise $d(c,f)=2$ and the $2$-convexity of $B_2(c)$ would imply that $d(c,g)=2$, contradicting our hypothesis.
 Hence we may apply $\TPCp(b,fg)$ and $\TPCp(c,fg)$. Assume first that $\TC(b,fg)$ applies, i.e. that there exists some $z\sim f,g$ with $d(b,z)=2$. Then because $d(b,h)=d(b,e)=2$, $\INCz(b)$ gives us $z\sim h$ and $z\sim e$, and the last adjacency $z\sim  d$ immediately follows when considering the $4$-cycle $ezhd$, hence $z$ is a universal vertex for $defgh$. Observe that we can make a similar reasonning when $\TC(c,fg)$ applies. By symmetry we are also done when $\TC(g,bc)$ or $\TC(f,bc)$ apply. Hence we may assume now that $\PCp$ and $\PCu$ respectively hold when we apply $\TPCp(b,fg), \TPCp(c,fg),\TPCu(g,bc)$ and $\TPCu(f,bc)$.
 The application of $\PCp(b,fg)$ gives us a vertex $s\sim b$ such that $d(s,f)= d(s,g)=2$ and $s$ is adjacent to every vertex of the clique $B_1(g)\cap B_2(b)$. Applying $\PCp(c,fg)$, we similarly find a vertex $t\sim c$ such that $d(t,f)=d(t,g)=2$ and $t$ is adjacent to every vertex of $B_1(f)\cap B_2(c)$. Now the applications of $\PCu(g,bc)$ and $\PCu(f,bc)$ to $b \sim c$ with respect to $t$ and $s$ give us vertices $u\sim g,t$ and $v\sim f,s$ such that $d(u,b)=d(u,c)=d(v,b)=d(v,c)=2$. Thanks to one of the previous cases (when $\TC$ applies somewhere), we may assume that $s\neq t$ and $u\neq v$. By definition of $s$, we must have $s\sim u$. By definition of $t$ we get $t\sim v$.
 Hence $svtu$ form a $4$-cycle, which implies that $s\sim t$ or $u\sim v$. Observe that these two adjacencies create the $4$-cycles $bstc$ or $uvfg$, which ultimately imply that $s\sim c$ or $t\sim b$ or $v\sim g$ or $u\sim f$. Any of these four adjacencies correspond to one of the applications of $\TC$ that we already covered, concluding the proof; see Figure \ref{fig: lastcase} for an illustration of the last case.
\end{proof}

  \tikzexternaldisable
  \begin{figure}[h]
    \centering
    \begin{tikzpicture}[scale=1]
    \tikzstyle{every node}=[draw,circle,fill=black,minimum size=4pt,
                            inner sep=0pt]

    \node (a) at (-2,2) [label=above: $a$] {};
    \node (b) at (-3,1) [label=left: $b$] {};
    \node (c) at (-2,0) [label=below: $c$] {};
    \node (d) at (-0.5,0.25) [label=below: $d$] {};
    \node (e) at (-0.5,1.75) [label=above: $e$] {};
    \node (f) at (1,2) [label=above: $f$] {};
    \node (g) at (2,1) [label=right: $g$] {};
    \node (h) at (1,0) [label=below: $h$] {};
    \node (s) at (-1.5,1.25) [label=above: $s$] {};
    \node (t) at (-1.5,0.75) [label=below: $t$] {};
    \node (u) at (0.5,0.75) [label=below: $u$] {};
    \node (v) at (0.5,1.25) [label=above: $v$] {};
    \draw (e) -- (a) -- (b) -- (c) -- (d) -- (e) -- (f) -- (g) -- (h) -- (d);
    \draw (b) -- (s) -- (v) -- (t) -- (c);
    \draw (f) -- (v);
    \draw (s) -- (u) -- (g);
    \draw (t) -- (u);
    \draw (s) -- (t)[dashed];
    \draw (u) -- (v)[dashed];
    \draw (b) -- (t)[dashed];
    \draw (c) -- (s)[dashed];
    \draw (v) -- (g)[dashed];
    \draw (u) -- (f)[dashed];
    \end{tikzpicture}
    \caption{The case $d(b,f)=d(g,c)=3$ and $d(b,h)=d(g,a)=2$ of the proof of Lemma \ref{lem: deuxpent}.
    The dashed edges occur at the end of the proof (not necessarily all of them).}  \label{fig: lastcase}
\end{figure}

\begin{example} \label{CB-pentagons-not-in-W5}
Lemma \ref{lem: deuxpent} is tight in the following sense: there exist CB-graphs with two pentagons glued together along one common edge having diameter $3$ and
such that only one of the two pentagons has a universal vertex; see Figure \ref{fig: diameter3notwm} for an illustration of such a graph. The pentagon $deabc$
is not included in any 5-wheel. 

  \tikzexternaldisable
  \begin{figure}[h]
    \centering
    \begin{tikzpicture}[scale=0.75]
    \tikzstyle{every node}=[draw,circle,fill=black,minimum size=4pt,
                            inner sep=0pt]

    \node (a) at (0,0) [label=below: $a$] {};
    \node (b) at (3,0) [label=below: $b$] {};
    \node (c) at (3,2) [label=right: $c$] {};
    \node (d) at (3,4) [label=right: $d$] {};
    \node (e) at (0,2) [label=left: $e$] {};
    \node (f) at (0,4) [label=left: $f$] {};
    \node (g) at (0,6) [label=above: $g$] {};
    \node (h) at (3,6) [label=above: $h$] {};
    \node (s) at (1.5,4) {};
    \node (t) at (1.5,2) {};
    \draw[very thick, magenta] (e) -- (a) -- (b) -- (c) -- (d) -- (e) -- (f) -- (g) -- (h) -- (d);
    \draw[thick] (g)--(s)--(f);
    \draw[thick] (e)--(s)--(d);
    \draw[thick] (h)--(s)--(t);
    \draw[thick] (s)--(c)--(t)--(f);
    \draw[thick] (t)--(b);
    \end{tikzpicture}
    \caption{A CB-graph of diameter $3$ consisting of two pentagons (in magenta) and of two additional vertices. Only the upper pentagon has a universal vertex.}  \label{fig: diameter3notwm}
\end{figure}
\end{example}

\begin{example}\label{CB-pentagons-infinite}
The first graph from  Figure \ref{fig: infinitegraph} shows that a property similar to Lemma \ref{lem: deuxpent} does not hold if we glue two pentagons to two opposite edges of a $C_4$ with a diagonal: it has diameter $4$, it has convex balls,  and none of its pentagons admit a universal vertex. This graph is also interesting because we can extend this construction to find an infinite $2$-connected CB-graph with an infinite number of pentagons and such that no pentagon has a universal vertex. Clearly, all such CB-graphs are not weakly systolic. The second graph from Figure \ref{fig: infinitegraph} shows how to get such a CB-graph. It it obtained by merging two copies of the first graph by identifying the vertices of one common pentagon from each copy. One can extend this construction by gluing an infinite number of pentagons, in a path-like way, where each connection between three successive pentagons is described by the second graph of Figure \ref{fig: infinitegraph}.
\end{example}

  \tikzexternaldisable
  \begin{figure}[h]
    \centering
    \begin{tikzpicture}[scale=1]
    \tikzstyle{every node}=[draw,circle,fill=black,minimum size=4pt,
                            inner sep=0pt]
    \begin{scope}[yshift=0cm]
    \node (a) at (-2,2) {};
    \node (b) at (-3,1) {};
    \node (c) at (-2,0) {};
    \node (d1) at (-0.5,0.25) {};
    \node (d2) at (0.5,0.25) {};
    \node (e1) at (-0.5,1.75) {};
    \node (e2) at (0.5,1.75) {};
    \node (f) at (2,2) {};
    \node (g) at (3,1) {};
    \node (h) at (2,0) {};
    \node (s) at (-1.5,1) {};
    \node (v) at (1.5,1) {};
    \draw (e1) -- (a) -- (b) -- (c) -- (d1) -- (e1) -- (e2) -- (d2) --(d1);
    \draw (e1) -- (d2) -- (h) -- (g) -- (f) -- (e2);
    \draw (b) -- (s) -- (a) -- (e1) -- (s) -- (d2) -- (e1) -- (v) -- (d2) -- (h) -- (v) -- (g);
    \end{scope}

    \begin{scope}[xshift=-1.5cm, yshift=-5cm]
    \node[draw,circle,fill=green,minimum size=4pt, inner sep=0pt] (a) at (-2,2) {};
    \node[draw,circle,fill=green,minimum size=4pt, inner sep=0pt] (b) at (-3,1) {};
    \node[draw,circle,fill=green,minimum size=4pt, inner sep=0pt] (c) at (-2,0) {};
    \node[draw,circle,fill=green,minimum size=4pt, inner sep=0pt] (d1) at (-0.5,0.25) {};
    \node[draw,circle,fill=magenta,minimum size=4pt, inner sep=0pt] (d2) at (0.5,0.25) {};
    \node[draw,circle,fill=green,minimum size=4pt, inner sep=0pt] (e1) at (-0.5,1.75) {};
    \node[draw,circle,fill=magenta,minimum size=4pt, inner sep=0pt] (e2) at (0.5,1.75) {};
    \node[draw,circle,fill=magenta,minimum size=4pt, inner sep=0pt] (f1) at (2,2) {};
    \node[draw,circle,fill=blue,minimum size=4pt, inner sep=0pt] (f2) at (2.5,2.5) {};
    \node[draw,circle,fill=magenta,minimum size=4pt, inner sep=0pt] (g1) at (3,1) {};
    \node[draw,circle,fill=blue,minimum size=4pt, inner sep=0pt] (g2) at (3.5,1.5) {};
    \node[draw,circle,fill=blue,minimum size=4pt, inner sep=0pt] (i) at (5,2.25) {};
    \node[draw,circle,fill=blue,minimum size=4pt, inner sep=0pt] (j) at (4.75,3.75) {};
    \node[draw,circle,fill=blue,minimum size=4pt, inner sep=0pt] (k) at (3.25,3.75) {};
    \node[draw,circle,fill=magenta,minimum size=4pt, inner sep=0pt] (h) at (2,0) {};
    \node[draw,circle,fill=green,minimum size=4pt, inner sep=0pt] (s) at (-1.5,1) {};
    \node[draw,circle,fill=magenta,minimum size=4pt, inner sep=0pt] (v) at (1.5,1) {};
    \node[draw,circle,fill=blue,minimum size=4pt, inner sep=0pt] (w) at (3.85,2.75) {};
    \draw (e1) -- (a) -- (b) -- (c) -- (d1) -- (e1) -- (e2) -- (d2) --(d1);
    \draw (e1) -- (d2) -- (h) -- (g1) -- (f1) -- (e2);
    \draw (b) -- (s) -- (a) -- (e1) -- (s) -- (d2) -- (e1) -- (v) -- (d2) -- (h) -- (v) -- (g1);
    \draw (f1)--(f2)--(k)--(j)--(i)--(g2)--(g1)--(f2)--(g2);
    \draw (g1)--(w)--(f2)--(k)--(w)--(j);
    \draw (v) -- (f2);
    \end{scope}

    \end{tikzpicture}
    \caption{The above graph shows an elementary gluing of two pentagons with the help of one layer of triangles. The beneath graph shows how to merge two copies of the first one by identifying $6$ common vertices, that are represented in magenta in the figure.}\label{fig: infinitegraph}
\end{figure}

\subsection{Triangle-free CB-graphs}
As mentioned in the introduction, every graph of diameter $2$ and girth $5$ is a CB-graph. In particular, the Petersen graph and the Hoffman-Singleton graph are CB-graphs. It is also easy to construct CB-graphs of arbitrary diameter and of girth $5$.  Namely, if $G_1$ and $G_2$ are two CB-graphs, then the \emph{wedge graph} $G_1\bigvee G_2$
obtained by identifying a single vertex of $G_1$ with a single vertex of $G_2$ also has convex balls. However, the graphs constructed in this way are not $2$-connected. The next result shows that this is the unique way
of producing CB-graphs of girth $5$ and arbitrary diameter.

\begin{proposition}
 \label{thm: Moore}
 The finite $2$-connected triangle-free CB-graphs are exactly the Moore graphs of diameter $2$. Consequently, every block in a finite triangle-free CB-graph is either a pentagon, or the Petersen graph,
 or the Hoffman-Singleton graph, or (if it exists) a Moore graph of diameter $2$ and degree $57$.
\end{proposition}

\begin{proof}
  A graph $G$ is \emph{geodetic} if there exists a unique shortest
  path between any pair of vertices of $G$. Observe that by INC, a
  triangle-free CB-graph is necessarily geodetic.  To prove
  Proposition \ref{thm: Moore} we use the following result of Stemple
  \cite{Stemple} about geodetic graphs of diameter $2$. We use its
  formulation from \cite{Blokhuis1988}.

  \begin{theorem}[\cite{Stemple}]
    \label{thm: Stemple}
    Let $G$ be a diameter $2$ geodetic graph. Then $G$ is of one of
    the following three types:
    \begin{enumerate}[(i)]
    \item $G$ contains a vertex adjacent to all other vertices.
    \item $G$ is regular.
    \item exactly two different degrees occur in $G$, say $a<b$, and
      if $B$ denotes the set of vertices of degree $b$, then it
      contains a clique of size $b-a+2>2$.
    \end{enumerate}
  \end{theorem}

  It is easy to observe that if $G$ satisfies item $(i)$, then it is
  either a star, in which case it is not $2$-connected, or it has a
  triangle. Moreover if item $(iii)$ is satisfied, then $G$ clearly
  has a triangle. Consequently, if $G$ is a 2-connected triangle-free
  CB-graph of diameter 2, then $G$ is regular and by the
  Hoffman-Singleton theorem~\cite{HoSi}, $G$ is either a pentagon, or
  the Petersen graph, or the Hoffman-Singleton graph, or (if it
  exists) a Moore graph of diameter $2$ and degree $57$.  Hence to
  prove Proposition \ref{thm: Moore}, it remains to show the following
  claim:

\begin{claim}\label{diameter2}
Any finite 2-connected  triangle-free  CB-graph $G$ has diameter at most $2$.
\end{claim}

\begin{proof} Suppose by way of contradiction that $\diam(G)\ge 3$ and let $v_0$ be a vertex of eccentricity at least 3, i.e., a vertex for which there exists a vertex $u$ with $d(v_0,u)\ge 3$.
Consider a BFS-ordering of the vertices of $G$ starting from $v_0$ (see Section \ref{sec: dism} for definitions). For a vertex $x$, we denote by $f(x)$ its parent in this BFS-order.
We call an edge $uv$ of $G$ \emph{horizontal} if $d(v_0,u)=d(v_0,v)$ and \emph{vertical} if $d(v_0,u)\ne d(v_0,v)$. We call a path $P$ of $G$ \emph{horizontal} (respectively, \emph{vertical})
if  all its edges are horizontal (respectively, vertical).
Since $G$ is a triangle-free CB-graph,  $G$ satisfies the following properties:
\begin{itemize}
\item[(1)] If $d(v_0,v)=i$ for some $v\in V$, then $v$ has a unique neighbor $f(v)$ at distance $i-1$ to $v_0$.
\item[(2)] If $\pi_1$ and $\pi_2$ are two pentagons of $G$ sharing an edge, then $\diam(\pi_1\cup \pi_2)=2$.
\item[(3)] If $Q=(x_1,x_2,\ldots,x_k)$ is a horizontal path, then $f^2(x_1)=f^2(x_2)=\ldots=f^2(x_k)$.
\item[(4)] If $P=(v_0=u_1,u_2,\ldots,u_k=u)$ is a shortest $(v_0,u)$-path, then all edges of $P$ are vertical and no  consecutive vertices $u_i,u_{i+1}$  of  $P$ are incident to horizontal edges of $G$.
\end{itemize}
Indeed, (1) is an immediate consequence of the geodecity of $G$, (2) follows from Lemma \ref{lem: deuxpent}, and (3) is obtained by applying ($\PC$) to the edges of the path $P$ and by property (1). The first part of (4) is trivial. If $u_iv_i$ and $u_{i+1}v_{i+1}$ are horizontal edges of $G$, then applying ($\PC$) to these edges, we conclude that the vertices $u_{i-2},u_{i-1},u_i,u_{i+1}$ are included in two pentagons sharing the edge $u_{i-1}u_i$. By (2) we deduce that $d(u_{i-2},u_{i+1})=2$, a contradiction.

Let $u$ be a furthest from $v_0$ vertex of $G$ and let $P=(v_0=u_0,u_1,\ldots,u_k=u)$ be the unique shortest $(v_0,u)$-path of $G$. Since $v_0$ has eccentricity $\ge 3$, $k\ge 3$. By (4), all edges of $P$ are vertical. Since $G$ is 2-connected,
$u$ is adjacent to a vertex $v\ne u_{k-1}$. From the maximality choice of $u$, necessarily $uv$ is a horizontal edge. By (3), we conclude that $f^2(v)=u_{k-2}$ and there exists a pentagon $uvv'u_{k-2}u_{k-1}$. Since $G$ is 2-connected, the vertices $u_{k-1}$ and $u_{k-3}$ are connected in $G$ by a path $Q=(u_{k-1}=x_1,x_2,\ldots,x_{m-1},x_m=u_{k-3})$ not passing via $u_{k-2}$. Since $u_kv$ is a horizontal edge, by (4) $u_{k-1}x_2$ must be vertical. Since $x_2\ne u_{k-2}$, $d(v_0,x_2)=d(v_0,u)$. From the choice of $u$, $d(v_0,x_3)\le d(v_0,u)$. Since $x_3\ne u_{k-1}$, we conclude that $x_2x_3$ is a horizontal edge. Let $Q'=(x_2,x_3,\ldots,x_i)$ be a maximal by inclusion subpath of $Q$ in which all edges are horizontal. By (3), for each vertex $x_i$ of $Q'$, we have $f^2(x_i) = f^2(x_2) = u_{k-2}$ and by ($\PC$) each edge $x_{j-1}x_j$ of $Q'$ is included in a pentagon $x_{j-1}y_{j-1}u_{k-2}y_jx_j$ (where $y_2=u_{k-1}$). Consider the next edge $x_ix_{i+1}$ of $Q$. From the maximality choice of $u$ and $Q'$, we conclude that $x_ix_{i+1}$ is vertical with $d(v_0,x_{i+1})<d(v_0,x_i)$. Hence $x_{i+1}=y_i$ and $u_{k-2}$ is the parent of $x_{i+1}$. Since $(v_0=u_1,u_2,\ldots,u_{k-2},y_i,x_i)$ is a shortest $(v_0,x_i)$-path and $x_i$ is incident to the horizontal edge $x_ix_{i-1}$ and $x_{i+2}\ne u_{k-2}$, by (4) we conclude that $x_{i+1}x_{i+2}$ is a vertical edge with $d(v_0,x_{i+2})=d(v_0,x_i)$. Again, from the maximality choice of $u$ we deduce that the edge $x_{i+2}x_{i+3}$ is horizontal. By considering
the maximal subpath $Q''=(x_{i+2},x_{i+3},\ldots,x_j)$ of $Q$ consisting of horizontal edges, by (3) we conclude that all vertices of $Q''$ have $u_{k-2}$ as their grandparent and that $x_{j+1}$ has $u_{k-2}$ as its parent. Continuing in this way, we obtain  that all vertices of the path $Q$ have $u_{k-2}$ either as their grandparent or as their parent. But this is impossible since $Q$ is a $(u_{k-1},u_{k-3})$-path  and $u_{k-3}$ is the parent of $u_{k-2}$. This contradiction establishes that the
eccentricity of any vertex of $G$ is 2 and thus that $\diam(G)=2$.
\end{proof}
This concludes the proof of Proposition \ref{thm: Moore}.
\end{proof}
 \section{A local-to-global characterization}
\label{sec: loc2glob}

The goal of this section is to prove the following topological characterization of CB-graphs:
\begin{theorem}
\label{thm: triangle-pentagon}
 For a graph $G$ the following conditions  are equivalent:
 \begin{enumerate}[(i)]
  \item $G$ is a graph with convex balls;
\item the triangle-pentagon complex $X_{\Triangle, \pentagon}(G)$ is simply connected and every ball of radius at most $3$ is convex.
 \end{enumerate}
 Furthermore, if $G$ is a graph in which all balls of radius at most $3$ are convex, then the 1-skeleton $\widetilde{G}$ of the universal cover of
 $X_{\Triangle, \pentagon}(G)$ is a graph with convex balls.
\end{theorem}

The implication (i)$\Rightarrow$(ii) is proved in the  next subsection. The proof of the implication (ii)$\Rightarrow$(i) is much harder and is based on the
proof of the second assertion of the theorem. We adapt the proof of similar local-to-global results from the papers  \cite{Bresaretal2013,ChalopiChepoiOsajda2015,CCHO,Osajda}.
The main difference with those proofs is the technical difficulty in dealing with pentagons in the inductive construction of the universal
cover of $X_{\Triangle, \pentagon}(G)$. The proof is given in the subsequent four subsections.

\subsection{Simple connectivity of the triangle-pentagon complex}The following lemma establishes the implication (i)$\Rightarrow$(ii) of Theorem \ref{thm: triangle-pentagon}:
\begin{lemma}\label{lem: triangle-pentagon-simply-connected}
 Let $G$ be a graph such that $\INC(v_0)$ and $\TPC(v_0)$ hold for some vertex $v_0$. Then $X_{\Triangle, \pentagon}(G)$ is simply connected.
\end{lemma}

\begin{proof}
  We show that any circuit $C=(u_0, \ldots, u_k)$ of $G$ is homotopic
  in $X_{\Triangle, \pentagon}(G)$ to the constant circuit $(v_0)$.
  In what follows, the operations over the indices of a cycle of
  length $k+1$ will be done modulo $k+1$. Let
  $\delta(C):=\sum_{i=0}^k 5^{d(v_0,u_i)+d(v_0, u_{i+1})} \geq
  k+1$. Then $\delta(C)=1$ if and only if $C$ is the constant circuit
  $(v_0)$ of length $1$. If $\delta(C)>1$, we assert that $C$ is
  homotopic in $X_{\Triangle, \pentagon}(G)$ to a circuit $C'$ with
  $\delta(C')<\delta(C)$. Observe that if $k=0$ and $\delta(C) > 1$,
  then $C=(u_0)$ is homotopic to $C'=(u'_0)$ with
  $u'_0 \in I(u_0,v_0) \cap N(u_0)$ and $\delta(C') =
  \delta(C)-1$. Observe also that if $k = 1$, then $C=(u_0,u_1)$ is
  homotopic to $C'=(u_0)$ and $\delta(C') < \delta(C)$.  Assume now
  that $k \geq 2$ and suppose without loss of generality that $u_1$
  maximizes the distance $d(v_0,u_i)$ among all $u_i\in C$.
  
  First we note that we may assume that for every $i$, we have
  $u_i\ne u_{i+1}$ and $u_{i-1}\ne u_{i+1}$, otherwise we are
  trivially done.  Let $r:=d(v_0, u_1)$.  First suppose that both
  $u_{0}$ and $u_2$ are at distance $r-1$ from $v_0$. By $\INC(v_0)$,
  we get $u_{0}\sim u_2$, so $C$ is homotopic in
  $X_{\Triangle, \pentagon}(G)$ to the circuit
  $C'=(u_0,u_2, u_3, \ldots, u_k)$. Moreover,
  $\delta(C')=\delta(C)-2 \cdot5^{2r-1} + 5^{2(r-1)} <\delta(C),$ and
  we are done.

 Otherwise, from the maximality choice of $u_1$, we can suppose
 without loss of generality that $d(v_0, u_2)=r$. Then
 $\TPC(v_0, u_1u_2)$ holds.  If $\mathrm{TC}(v_0, u_1u_2)$ applies,
 then there exists a vertex $z$ such that $z\sim u_1, u_2$ and
 $d(v_0, z)=r-1$. Then $C$ is homotopic in $X_{\Triangle, \pentagon}$
 to the circuit $C':=(u_0,u_1, z, u_2, u_3, \ldots, u_k)$ and we have
 $\delta(C')=\delta(C) - 5^{2r} + 2\cdot 5^{2r-1} < \delta(C).$
 Otherwise, $\mathrm{PC}(v_0, u_1u_2)$ applies and there exist
 $w, w', z$ such that $d(v_0, w)=d(v_0, w')= r-1$, $d(v_0, z)=r-2$ and
 $u_1wzw'u_2$ is a pentagon of $G$. Then $C$ is homotopic in
 $X_{\Triangle, \pentagon}(G)$ to the circuit
 $C':=(u_0,u_1, w, z, w', u_2, u_3, \ldots, u_k)$ and we have
 $\delta(C')= \delta(C) - 5^{2r} + 2\cdot5^{2r-1} + 2 \cdot 5^{2r-3}
 \leq \delta(C) - 5^{2r} + 4 \cdot 5^{2r-1} < \delta(C).$ In both
 cases we have the inequality $\delta(C')<\delta(C)$, so we are done.
\end{proof}

\subsection{The inductive construction of the universal cover} The proof of the implication (ii)$\Rightarrow$(i) of Theorem \ref{thm: triangle-pentagon} is technically involved and first we outline it. Notice, that this implication follows from the second assertion of the theorem. To establish it, let $G$ be a graph in which all balls of radius at most 3 are convex and set $X:= X_{\Triangle, \pentagon}(G)$.  We construct the universal cover $\tX$ of $X$ inductively, as the union $\bigcup_{i\geq 0}\tB_i$, where each $\tB_i$ is a ball of radius $i$ centered at an arbitrary (but fixed) basepoint $\tv_0$. We define the graph $\tG_i$ having $\tB_i$ as the vertex set and we let $\tX_i:= X_{\Triangle, \pentagon}(\tG_i)$ be the triangle-pentagon complex $\tG_i$. Our construction will ensure that the properties $\INC(\tv_0)$ and $\TPC(\tv_0)$ are satisfied for every $i$, i.e. when the distances involved are at most $i$. We will also construct inductively in a coherent way the maps $f_i : \tX_i\to X$ such that $f := \bigcup_{i\geq 0}f_i : \tX \to X$ is well defined and is a covering map. As the universal cover of $X$ is unique up to (cell-preserving) isomorphisms, this implies that the 1-skeleton $\tG$ of $\tX$ satisfies $\INC(\tv_0)$ and $\TPC(\tv_0)$ for any choice of the basepoint $\tv_0$. By Theorem  \ref{thm: TPC}, $\tG$ is a graph with convex balls, and by Lemma \ref{lem: triangle-pentagon-simply-connected} the complex $\widetilde X = X_{\Triangle, \pentagon}(\tG)$ is simply connected.
If the input complex $X$ is simply connected, then we must have $X=\tX$ and thus $G=\tG$, showing that $G$ is a CB-graph.

For any vertex $\tv$ of $\tG_i$, we denote by $v$ the image $f_i(\tv)$ of $\tv$ in $G$.
Let $v_0$ be an arbitrary but fixed vertex of $G$ and set $\tB_0:= \sg{\tv_0}$  and define $f_0(\tv_0):= v_0$.
We let $\tS_{i+1}:= \tB_{i+1}\setminus \tB_i$ for every $i \geq 0$.
We prove the following properties by induction:
\begin{itemize}
 \item[($P_i$)] $B_j(\tv_0,\tG_i)= \tB_j$ for every $0\leq j \leq i$.
 \item[($Q_i$)] $\tG_i$ satisfies $\INC(\tv_0)$ and $\TPC(\vz)$.
 \item[($R_i$)] for any $\tu \in \tB_{i-1}$, $f_i$ defines an isomorphism between the subgraphs $\tG_i[B_1(\tu,\tG_i)]$ and $G[B_1(f_i(\tu),G)]$.\item[($S_i$)] for any five vertices $\tu, \tx, \tw, \ty, \tv \in \tB_i$ such that $\tu \tx \tw \ty \tv$ is a path in $\tG_i$, if $uxwyv$ is a 5-cycle in $G$, then $\tu \sim \tv$ in $\tG_i$.
 \item[($T_i$)] for any $\tu \in \tS_i$, $f_i$ defines an isomorphism between the subgraphs $\tG_i\left[B_1(\tu,\tG_i) \right]$ and $G\left[f_i(B_1(\tu,\tG_i)) \right]$.
\end{itemize}

Assume that everything was constructed and proved for every $j=0, \ldots, i$.
We define the set of all the vertices ``which can be seen from $\tB_i$ but which were not reached yet'' by setting
\[Z :=\left\{(\tw, z): \tw \in \tS_i \text{ and } z\in B_1(w,\tG_i)\setminus f_i(B_1(\tw, \tG_i))\right\}.\]
On $Z$ we define the binary relation $\equiv $ by setting $(\tw, z)\equiv (\tw', z')$ if and only if $z= z'$ and either
$\tw=\tw'$ or $\tw \sim \tw'$. We will later show that $\equiv$ is an equivalence relation.
First we establish two useful lemmas.

\begin{lemma}\label{lem-voisinZ}
  For any $(\tw,z) \in Z$ and any $\tx \in \Bi$ such that $\tx \sim \tw$ and $f_i(\tx) = x \sim z$, $\tx \in \tS_i$ and $(\tx,z) \in Z$.
\end{lemma}

\begin{proof}
  If $\tx \in \tB_{i-1}$, then by \Ri, there exists $\tz \sim \tx$
  such that $f_i(\tz) = z$. Similarly, by the definition of $Z$, if
  $\tx \in \tB_i$ and $(\tx,z) \notin Z$, then there exists
  $\tz \sim \tx$ such that $f_i(\tz) = z$. In both cases, by \Ri or
  \Ti, we have $\tw \sim \tz$ and thus $(\tw,z) \notin Z$, a
  contradiction.
\end{proof}

The following lemma implies that $\INC(\vz)$ inductively holds:

\begin{lemma}
\label{lem: INCcongr}
 For any $(\tw, z), (\tw', z') \in Z$ such that $z=z'$ and $\tw \sim \tw'$, there exists a vertex $\tu \in \tB_{i-1}$ such that $\tu \sim \tw, \tw'$. Consequently, $(\tw, z)\equiv (\tw', z')$ if and only if $z= z'$ and $\tw = \tw'$ or $\tw \sim \tw'$ and there exists some $\tu \in \tB_{i-1}$ such that $\tu \sim \tw, \tw'$.
\end{lemma}

\begin{proof}
Assume that there is no such vertex $\tu$. Then by ($Q_i$) we can apply $\TPC(\tv_0, \tw\tw')$. Since  the triangle condition does not applies,  there exist vertices $\tu, \tu' \in \tB_{i-1}, \ty \in \tB_{i-2}$ such that $\tw \tu \ty \tu' \tw'$ is a pentagon in $\tG_i$. By ($R_i$) and ($T_i$) this means that $wuyu'w'$ is a pentagon in $G$. As $z\sim w, w'$, this  implies that $z\notin \sg{w, w', u, u', y}$. By the convexity of $B_2(z,G)$, we conclude that $1\le d(z,y)\leq 2$.

By Lemma~\ref{lem-voisinZ}, $z \nsim u, u'$.
Thus $d(z,y)=2$ and there exists $t \sim z, y$ different from $w,w',u,u',y$. By \Ri there exists $\ttt \in \tB_{i-1}$ such that $f_i(\ttt)=t$ and $\ttt \sim \ty$. By \Ri again there exists $\tz \in \Bi$ such that $f_i(\tz)=z$ and $\tz \sim \ttt$. We assert that $\tz \sim \tw$. Indeed,  either $t \sim u$ and by the convexity of $B_1(t,G)$, this is equivalent to $t\sim w$ and we can apply \Ri, otherwise, $t\nsim u, w$ and $ztyuw$ is a pentagon of $G$, so we can apply \Si. In both cases, $\tz \sim \tw$ contradicts the fact that $(\tw, z)\in Z$ and we are done.
\end{proof}

\begin{lemma}
\label{lem: equivcongr}
The relation $\equiv$ is an equivalence relation on $Z$.
\end{lemma}

\begin{proof}
  Reflexivity and symmetry of $\equiv$ immediately follow from its
  definition. Hence, we only have to show transitivity.  Let
  $(\tw, z), (\tw', z'), (\tw'', z'') \in Z$ such that
  $(\tw, z)\equiv (\tw', z')$ and $(\tw', z') \equiv (\tw'', z'')$.
  Then $z=z'=z''$. We assume that $\tw, \tw', \tw''$ are pairwise
  distinct, otherwise we are done. Hence $\tw' \sim \tw, \tw''$.  We
  also assume that $\tw \nsim \tw''$, otherwise we are done. By \Ti
  applied to $\tw'$, this implies that $w \nsim w''$. By Lemma
  \ref{lem: INCcongr} there exist $\tu, \tu' \in \tB_{i-1}$ such that
  $\tu\sim \tw, \tw'$ and $\tu'\sim \tw', \tw''$.  If $\tu=\tu'$, then
  by the convexity of $B_1(u,G)$ we have $u \sim z$, which is
  impossible by Lemma~\ref{lem-voisinZ}. Hence $\tu\neq \tu'$ and
  $z \nsim u,u'$. By ($Q_i$) we can use $\INC(\vz)$ to get
  $\tu\sim\tu'$ and a vertex $\ts \in \tB_{i-2}$ such that
  $\ts \sim \tu, \tu'$. From \Ri and the definition of $Z$ we conclude
  that the vertices $z,w,w',w'',u,u'$ are pairwise distinct. Moreover,
  $s\notin \sg{u, u', w, w', w'', z}$ by \Ri, so by the convexity of
  $B_2(s,G)$, we get $d(s,z)\leq 2$. If $d(s,z)=1$, then by the
  convexity of $B_1(s,G)$ and because $s\nsim w$, we have $u \sim z$,
  which gives a contradiction.  Hence $d(s,z)=2$ and let $t\sim
  z,s$. Then $t\neq u$ and by \Ri there exists $\ttt \in \tB_{i-1}$
  and $\tz \in \Bi$ such that $f_i(\tz)=z$ and such that
  $\ttt \sim \tz, \ts$. Depending whether $t \sim u, w$ or not, by \Ri
  and \Si we get $\tz \sim \tw$, which contradicts that
  $(\tw, z)\in Z$.
\end{proof}

Let $\tS_{i+1}$ denote the set of equivalence classes of $\equiv$, i.e., $\tS_{i+1}=Z/_{\equiv}$. For a pair $(\tw,z)\in Z$,
we denote by $[\tw,z]$ the equivalence class of $\equiv$ containing  $(\tw,z).$
We let $\tB_{i+1}:= \Bi \sqcup \tS_{i+1}$ and define the edges of $\tG_{i+1}$ in the following way: (a) the adjacencies between the
vertices of $\Bi$ do not change, (b) we let $\tw \sim \cro{\tw, z}$, and (c) if $\ta = \cro{\tw, z}, \tb = \cro{\tw', z'}\in \tS_{i+1}$,
then we let $\ta \sim \tb$ in $\tG_{i+1}$ if $z\sim z'$ in $G$ and one of the following two conditions is satisfied:
\begin{itemize}
\item[(1)] there exists  $\tw\in \Bi$ such that $\ta = \cro{\tw, z}$ and $\tb = \cro{\tw, z'}$,
\item[(2)] there exist $\tw, \tw' \in \Bi$ and $\ty \in \tB_{i-1}$ such that $\ta = \cro{\tw, z}$, $\tb = \cro{\tw', z'}$ and  $\ty \sim \tw, \tw'$.
\end{itemize}

From \Ri and the fact that balls of $G$ of radius $1$ are convex, we can equivalently require in (2) that $awyw'b$ is a pentagon of $G$, otherwise (1) holds.

Finally, we define $f_{i+1}$ on $\tB_{i+1}$ by setting $f_{i+1}(\tu):=f_{i}(\tu)$ when $\tu \in \Bi$ and $f_{i+1}(\cro{\tw, z}):=z$ for any $\cro{\tw, z} \in \tS_{i+1}$.

We now prove that if an edge $\ta\tb$ between vertices of $\tS_{i+1}$
is defined only by condition (2) above, then $\PCd(\tv_0,\ta\tb)$ holds.

\begin{lemma}\label{lem-ind-PCdist2}
  Consider an edge $\ta\tb$ of $\tG_{i+1}$ with
  $\ta,\tb \in \tS_{i+1}$. If there is no $\tw \in \tB_i$ such that
  $\ta = [\tw,a]$ and $\tb = [\tw,b]$, then for any
  $\ttt \in \tB_{i-1}$ such that $d(\ta,\ttt) = 2$, we have
  $d(\tb,\ttt) = 2$.
\end{lemma}

\begin{proof}
  By the definition of the edges of $\tG_{i+1}$ and since $\ta,\tb$
  have no common neighbor in $\tB_i$, there exist $\tw, \tw' \in \Bi$
  and $\ty \in \tB_{i-1}$ such that $\ta = \cro{\tw, a}$,
  $\tb = \cro{\tw', b}$ and $\ty \sim \tw, \tw'$. We distinguish two
  cases, depending whether $\ttt \sim \tw$ or not.

  \medskip\noindent
  {\bf Case 1.} $\ttt$ is a neighbor of $\tw$.

\begin{proof}  If $\ttt$ is a neighbor of $\tw'$, we are done. Thus, we can assume
  that $\ttt \nsim \tw'$ and $\ttt \neq \ty$. By $\INC(\tv_0)$, we
  have $\ttt \sim \ty$ and there exists $\ts \in \tB_{i-2}$ such that
  $\ts \sim \ttt,\ty$. In $G$, $abw'yw$ is a pentagon and since
  $t \sim w,y$, we have $t \notin \sg{a,b,w',y,w}$. By \Ri, we have
  $t \nsim w'$.  By Lemma~\ref{lem-voisinZ}, we have $t \nsim a$. If
  $t \sim b$, then $wabt$ is an induced 4-cycle in $G$, which is
  impossible. Since $B_2(t,G)$ is convex, we have $d(t,b) = 2$. Let
  $x \sim b,t$ and note that $x \notin \sg{a,b,w',y,w,t}$. By \Ri,
  there exists $\tx \in \tB_i$ such that $\tx \sim \ttt$ and
  $f_i(\tx) = x$. We claim that $\tx \in \tS_i$. Indeed, if
  $\tx \notin \tS_i$, there exists $\tb' \in \tB_i$ such that
  $f_i(\tb') = b$ and $\tb' \sim \tx$. By \Si applied to the vertices
  $\tb',\tx,\ttt,\ty,\tw'$, we have $\tb'\sim \tw'$, contradicting the
  fact that $(\tw',b) \in Z$. Since $G$ does not contain induced
  4-cycles, $x \sim w'$ if and only if $x \sim y$.

  If $x \sim w',y$, by \Ri, $\tx \sim \ty,\tw'$ and by
  Lemma~\ref{lem-voisinZ}, we have $(\tx,b) \in Z$. By the definition of $\equiv$, we have
  $\tb = [\tw',b] = [\tx,b]$. Consequently, in this case, we have
  $d(\tb,\ttt) = 2$. Let now $x \nsim w',y$. Let $s = f_i(\ts)$ and observe that
  $s \notin \sg{a,b,x}$ since $s \sim y$. Moreover, by \Ri, we have
  $s \notin \sg{t,w,w',y}$ and $s \nsim w,x,w'$. Consequently,
  $d(s,x) = d(s,w') = 2$ and by the convexity of $B_2(s,G)$, we get
  $d(s,b) \leq 2$. If $s \sim b$, then $syw'b$ is a forbidden
  induced $4$-cycle, thus $d(s,b) = 2$. Let
  $r \sim b,s$. By \Ri, there exists $\tr \in \tB_{i-1}$ and
  $\tb' \in \tB_i$ such that $f_i(\tr) = r$, $f_i(\tb') = b$, and
  $\tr \sim \ts,\tb'$. By \Si applied to the vertices
  $\tb',\tr,\ts,\ty,\tw'$, we have $\tb'\sim \tw'$, contradicting the
  fact that $(\tw',b) \in Z$.
\end{proof}

  \medskip\noindent
  {\bf Case 2.} $\ttt$ is not a neighbor of $\tw$.

  \begin{proof} Since $d(\ttt,\ta) = 2$, there exists
    $\ts \sim \ta,\ttt$. Since $\ttt \in \tB_{i-1}$ and
    $\ta \in \tS_{i+1}$, we have $\ts \in \tS_i$ and
    $[\ts,a] = \ta = [\tw,a]$. By the definition of $\equiv$, this
    implies that $\tw \sim \ts$ and by Lemma~\ref{lem: INCcongr},
    there exists $\tr \in \tB_{i-1}$ such that $\tr \sim \tw,\ts$. By
    applying Case~1 to $\tr$, we have that $d(\tr,\tb) = 2$ and thus,
    there exists $\tx \sim \tr,\tb$. Consequently, we can replace
    $\ty$, $\tw$ and $\tw'$ by $\tr$, $\ts$ and $\tx$ and we can apply
    Case~1 to $\ttt$.
  \end{proof}

Cases 1 and 2 establish the lemma.
\end{proof}

\subsection{Properties ($P_{i+1}$), ($Q_{i+1}$), ($R_{i+1}$), and ($T_{i+1}$)}

We start with the proof of ($P_{i+1}$) and ($Q_{i+1}$):

\begin{lemma}
\label{lem: Qi}
$\tG_{i+1}$ satisfies ($P_{i+1}$).
\end{lemma}

\begin{proof}
  We show that for every $0\le j\le i+1$ we have
  $\tB_j= B_j(\vz,\tG_{i+1})$. This is true when $j \leq i$ by
  induction hypothesis.  By construction, every vertex in $\tS_{i+1}$
  is adjacent to at least one vertex of $\Bi$ and is only adjacent to
  vertices of $\tS_i$, hence it is at distance $i+1$ from $\vz$ and we
  are done.
\end{proof}

\begin{lemma}
 $\tG_{i+1}$ satisfies ($Q_{i+1}$).
\end{lemma}
\begin{proof}
$\INC(\tv_0)$ is obtained by induction hypothesis and by Lemma \ref{lem: INCcongr}.
Now we prove that $\TPC(\tv_0)$ holds. Let $\tw \sim \tw' \in \tB_{i+1}$, both at distance $k$ from $\vz$. If $k\leq i$, then we conclude by induction hypothesis.
Otherwise, note that the adjacency $\tw \sim \tw'$ either comes from (1) or (2), which correspond to the triangle or the pentagon condition, respectively.
\end{proof}

The following two lemmas are trivial.

\begin{lemma}
\label{lem: l313}
If $\ta \sim \tb$ in $\tG_{i+1}$, then $a\sim b$ in $G$.
\end{lemma}

\begin{lemma}
\label{lem: l314}
If $a \sim b$ in $G$ and $\ta \in \Bi$, then there exists $\tb \in \tB_{i+1}$ such that $\ta \sim \tb$ and $f_{i+1}(\tb)=b$.
\end{lemma}

Now we show that $f_{i+1}$ is ``locally injective''.
\begin{lemma}
\label{lem: injloc}
Let $\ta,\tb,\tc \in \tB_{i+1}$, such that  $\tb \neq \tc$ and $\ta \sim \tb, \tc$, then $b\neq c$.
\end{lemma}

\begin{proof}
First,  if $\tb \sim \tc$, then $b \sim c$ by Lemma \ref{lem: l313} and we are done.
Hence, let $\tb \nsim \tc$. Moreover, if $\ta, \tb, \tc \in \Bi$, then we are done by \Ri. From the definition of
$\tS_{i+1}$ we are also done when $\ta \in \Bi$ and $\tb, \tc \in \tS_{i+1}$. Suppose by way of contradiction  that $b=c$.
If $\ta, \tb \in \Bi$ and $\tc \in \tS_{i+1}$, then we
have $\tc = \cro{\ta, b}$, and we get a contraction with the fact that $(\ta, b)\in Z$. Hence, further we
may assume that $\ta \in \tS_{i+1}$. If $\tb, \tc \in \Bi$, then by definition of $\tS_{i+1}$ we have $[\tb,a] = \ta = [\tc,a]$, and by definition of $\equiv$, we have $\tb \sim \tc$, a contradiction. Now,  let $\tb \in \tS_{i+1}$ and $\tc \in \Bi$.
Then $\ta = \cro{\tc, a}$.
If condition (1) applies to the edge $\ta\tb$, then there exists $\tw \in \Bi$
such that $\tw \sim \tb, \ta$. Necessarily, $\tw \neq \tc$ because
$\tb \nsim \tc$. By $\INC(\tv_0)$, we get $\tw \sim \tc$.
Replacing $\ta$ by $\tw$, we are in the previous case where  $\tw \sim \tb, \tc$,
which was shown to be impossible. Therefore assume that there is no such $\tw\in \Bi$ and thus only condition (2) applies. Consider a neighbor $\tw \in \tB_{i-1}$ of $\tc$ and observe that by Lemma~\ref{lem-ind-PCdist2}, there exists $\tu \sim \tb,\tw$. By \Ri, we get that $\tu \sim \tc$, and replacing $\ta$ by $\tu$
we are again in a previous case where $\tu \sim \tb, \tc$.

Finally, let $\ta, \tb, \tc \in \tS_{i+1}$. First assume that
condition (1) applies to both $\ta\tb$ and $\ta\tc$. Then there exist
$\tu, \tu'\in \Bi$ such that $\tu \sim \ta, \tb$ and
$\tu' \sim \ta, \tc$. If $\tu = \tu'$, replacing $\ta$ by $\tu$, we
are in a previous case where $\tu \sim \tb,\tc$. So, assume that
$\tu \neq \tu'$. This implies that $\tu \sim \tu'$ since
$\ta = [\tu,a] = [\tu',a]$. By definition of $\equiv$, we thus have
$\tb = [\tu, b] = [\tu', b] = \tc$, a contradiction.

Assume now that condition (1) applies to $\ta\tb$, while $\ta\tc$ is
defined only by condition (2). Then, there exists $\tu \in \Bi$ such
that $\tu\sim \ta,\tb$. Pick any neighbor $\ty$ of $\tu$ in
$\tB_{i-1}$. By Lemma~\ref{lem-ind-PCdist2}, there exists
$\tu' \in \tB_i$ such that $\tu'\sim \ty,\tc$. Since condition (1)
does not apply to the edge $\ta\tc$, $acu'yu$ should be a pentagon of
$G$, contradicting the fact that $c=b \sim u$.

Consequently, we may assume that both $\ta\tb$ and $\ta\tc$ are
defined only by condition (2). Pick any vertex $\ty \in \tB_{i-1}$ such
that $d(\ty,\ta) = 2$. By Lemma~\ref{lem-ind-PCdist2}, there exists
$\tw \sim \ty,\tb$ and $\tw' \sim \ty,\tc$. Consequently, $bwyw'$ is a
4-cycle of $G$ that cannot be induced. By Lemma~\ref{lem-voisinZ},
$b \nsim y$ and thus $w\sim w'$. By \Ri, we have $\tw \sim \tw'$ and
consequently, $\tb = [\tw,b] = [\tw',b] = [\tw',c] = \tc$, a contradiction.
\end{proof}

Now we show that $f_{i+1}$ is ``locally an isomorphism''.
\begin{lemma}
\label{lem: isoloc}
 Let $\ta, \tb, \tc \in \tB_{i+1}$ be such that $\ta \sim \tb, \tc$. Then $\tb \sim \tc$ if and only if $b \sim c$.
\end{lemma}

\begin{proof}
  If $\tb \sim \tc$, then $b\sim c$ by Lemma \ref{lem:
    l313}. Conversely, suppose that $b\sim c$ in $G$. Suppose first
  that $\ta \in \tB_i$.  If $\tb, \tc \in \tB_i$, then we get the
  result by \Ri or by \Ti. If $\tb \in \tB_i$ and $\tc \in \tS_{i+1}$,
  then by Lemma~\ref{lem-voisinZ}, $(\tb,c) \in Z$ and
  $\tb \sim [\tb,c] = [\ta,c] = \tc$. If $\tb,\tc \in \tS_{i+1}$,
  $\tb \sim \tc$ by the definition of  $E(\tG_{i+1})$.

  Suppose now that $\ta \in \tS_{i+1}$. If $\tb, \tc \in \tB_i$, then
  $\ta = [\tb,a] = [\tc,a]$ and the result follows from the definition
  of $\equiv$.  Assume that $\tb \in \tS_{i+1}$ and $\tc \in
  \tB_i$. If $(\tc,b) \in Z$, then $\ta = [\tc,a] \sim [\tc,b]$ and,
  by Lemma~\ref{lem: injloc}, we have $[\tc,b] = \tb$, implying that
  $\tb \sim \tc$.  If $(\tc,b) \notin Z$, then there exists
  $\tb' \in \tB_i$ such that $\tb' \sim \tc$ and $f_i(\tb') = b$.  By
  Lemma~\ref{lem-voisinZ}, $\tb' \in \tS_i$, $(\tb',a) \in Z$, and
  $[\tb',a] = [\tc,a] = \ta$. Consequently, $\ta \sim \tb, \tb'$,
  contradicting Lemma~\ref{lem: injloc}.

  Now, let $\ta, \tb, \tc \in \tS_{i+1}$. Observe that if $\tb, \tc$
  have a common neighbor $\tu \in \tB_i$, we are done by the previous
  cases. Assume that such a common neighbor does not exist.
First suppose that both $\ta\tb$ and $\ta\tc$ are defined by
  condition (1). Thus there exists $\tu, \tu' \in \tB_i$ such that
  $\tu \sim \ta,\tb$ and $\tu' \sim \ta,\tc$. Then,
  $\ta = [\tu,a] = [\tu',a]$ and $\tu \sim \tu'$. Since
  $\tu \nsim \tc$ and $\tu' \nsim \tb$, by the previous cases, we have
  $u \nsim c$ and $u'\nsim b$, and thus $bcu'u$ is an induced 4-cycle
  of $G$, a contradiction.
Assume now that condition (1) applies to $\ta\tb$, while $\ta\tc$ is
  defined only by condition (2). Then, there exists $\tu \in \Bi$ such
  that $\tu\sim \ta,\tb$. Pick any neighbor $\ty$ of $\tu$ in
  $\tB_{i-1}$. By Lemma~\ref{lem-ind-PCdist2}, there exists
  $\tu' \in \tB_i$ such that $\tu'\sim \ty,\tc$. But then,
  $\tb \sim \tc$ by condition (2).
Finally, assume that both $\ta\tb$ and $\ta\tc$ are defined only by
  condition (2). Pick any vertex $\ty \in \tB_{i-1}$ such that
  $d(\ty,\ta) = 2$. By Lemma~\ref{lem-ind-PCdist2}, there exists
  $\tw \sim \ty,\tb$ and $\tw' \sim \ty,\tc$. Again, $\tb \sim \tc$ by
  condition (2).
\end{proof}

Now we have everything at hand to prove \Rip and \Tip.

\begin{lemma}
 \label{lem: RipTip}
 The properties \Rip and \Tip hold.
\end{lemma}
\begin{proof}
  Let $\tw\in \Bip$. By Lemmas~\ref{lem: injloc} and~\ref{lem:
    isoloc}, $f_{i+1}$ induces an isomorphism between $G[B_1(\tw,\tG_{i+1})]$ and
  $G[f_{i+1}(B_1(\tw,\tG_{i+1}))]$ and thus \Tip holds. Moreover, by
  Lemma~\ref{lem: l314}, for any $\tw \in \tB_i$, $f_{i+1}$ induces a
  surjection between $B_1(\tw,\tG_i)$ and $B_1(f_{i+1}(\tw),G)$ and consequently,
  an isomorphism between $\tG_{i+1}[B_1(\tw,\tG_{i+1})]$ and $G[B_1(f_{i+1}(\tw),G)]$.  Thus
  \Rip holds.
\end{proof}

From  properties \Rip and \Tip, we immediately obtain the following lemma:

\begin{lemma}
 \label{lem: morphism}
 For any triangle $\tu\tv\tw$ of $\tG_{i+1}$, its image $uvw$ under $f_{i+1}$ is a triangle of $G$.
 Similarly, the image of every pentagon $\tu \tx \tw \ty \tv$ of $\tG_{i+1}$ is a pentagon of $G$.
\end{lemma}

\begin{lemma}\label{lemma: virtC4} Let $\ta\tb\tc\td$ be a path of $\tG_{i+1}$. Then the following holds:
\begin{itemize}
\item if $a=d$ in $G$, then $\ta = \td$;
\item if $a\sim d$ in $G$, then $\ta \sim \td$.
\end{itemize}
\end{lemma}

\begin{proof}
If $a=d$, then  by applying \Rip or \Tip to $\tb$, we conclude that $\ta \sim \tc$. By applying \Rip or \Tip to $\tc$, we must have $\ta=\td$.
If $a\sim d$, then by \Rip or \Tip, $abcd$ is a $4$-cycle of $G$. Since this cycle  cannot be induced, we can assume without loss of generality that  $a\sim c$.
By applying \Rip or \Tip to $\tb$, we get $\ta \sim \tc$. Then applying  \Rip or \Tip to $\tc$, we get $\ta \sim \td$.
\end{proof}

\subsection{Property ($S_{i+1}$)}
In this subsection, we establish the property ($S_{i+1}$), whose proof is the most involved.  We call a path $\tpi=\tu\tx\tw\ty\tv$ of $\tG_{i+1}$ a \emph{virtual 5-cycle} if it image $\pi=uxwyv$ in $G$ is a 5-cycle.

\begin{proposition} \label{Si+1} The property ($S_{i+1}$) holds, i.e., if $\tpi=\tu\tx\tw\ty\tv$ of $\tG_{i+1}$ is a virtual 5-cycle, then $\tu\sim\tv$.
\end{proposition}

\begin{proof} We say that a vertex $\ta$ of $\tB_{i+1}$ has \emph{height}  $j$ (notation $h(\ta)=j$) if $\ta\in \tS_j$. We call an edge $\ta\tb$
\emph{horizontal of height $j$} if $h(\ta)=h(\tb)=j$ and \emph{vertical} if $h(\ta)\ne h(\tb)$.
From the definition of edges of the graph $\tG_{i+1}$,  $|h(\ta)-h(\tb)|\le 1$ for any edge $\ta\tb$.
The \emph{height} $h(\tpi)$ of a virtual 5-cycle $\tpi$ is the sum of heights of its vertices. We proceed by induction
on the height $h(\tpi)$ of $\tpi$. If all vertices of $\tpi$ belong to $\tB_i$, then applying \Si we get $\tu\sim\tv$.
Thus, we will assume that $\tpi$
contains a vertex of $\tS_{i+1}$.

\medskip\noindent
{\bf Case 1.} The cycle $\pi=uxwyv$ has a chord.

\begin{proof} First, let $u \sim w$ or  $v\sim w$, say the first. Then applying Lemma \ref{lemma: virtC4} to the path $\tu \tw \ty \tv$ we get $\tu \sim \tv$. Now, let
$w\nsim u,v$. Since  $G$ does not contain induced $4$-cycles, we conclude that $y\sim x$. By \Rip or \Tip, we get $\tx \sim \ty$. Applying Lemma \ref{lemma: virtC4}
to the path $\tu\tx\ty\tv$, we  get $\tu \sim \tv$.
\end{proof}

In the remaining part of the proof we assume that $\pi=uxwyv$ is a pentagon of $G$. This implies that $\tpi$ does not contain \emph{peaks}, i.e., vertices whose heights are larger than the heights of
their neighbors.   Indeed, if say $h(\tx)<h(\tw)>h(\ty)$, then by $\INC(\tv_0)$  we have $\tx\sim \ty$ and thus $x\sim y$, a contradiction.

\medskip\noindent
{\bf Case 2.} $|h(\tu)-h(\tv)|\ge 2$.

\begin{proof}
Let $h(\tu)\ge h(\tv)+2$. By \Ri there exists $\tu'\sim \tv$. Since $h(\tu')\le h(\tv)+1<h(\tu)$, the height of the virtual 5-cycle $\tx\tw\ty\tv\tu'$ is less than $h(\tpi)$. By induction hypothesis, $\tx\sim \tu'$, contradicting Lemma \ref{lem: injloc}.
\end{proof}

Therefore, further we can suppose that $|h(\tu)-h(\tv)|\le 1$.

\medskip\noindent
{\bf Case 3.} All edges of $\tpi$ are vertical.

\begin{proof}
Since $\tpi$ does not contain peaks and horizontal edges and  $|h(\tu)-h(\tv)|\le 1$, necessarily $h(\tu)>h(\tx)>h(\tw)<h(\ty)<h(\tv)$.
Applying rule (2) of the definition of edges of $\tG_{i+1}$, we get $\tu\sim \tv$.
\end{proof}

Thus, we can assume that at least one edge of $\tpi$ is horizontal. We pick as $\te$  the second
or the third edge of $\tpi$ if such an edge is horizontal and the first or the fourth edge if this edge is horizontal and $\tx\tw$
and $\tw\ty$ are vertical. Then $\TPC(\tv_0)$ applies to $\te$.  Denote by $j$ the height of the edge $\te$. Observe that if $\te \notin \{\tx\tw,\tw\ty\}$, then $\tx\tw, \tw\ty$ are vertical and since $\tpi$ does not contain peaks and $|h(\tu)-h(\tv)|\leq 1$, we necessarily have $h(\tw) = j-1$.

\medskip\noindent
{\bf Case 4.} $\mathrm{TC}$  applies to the edge $\te$.

\begin{proof} First, let $\te=\tx\tw$. By $\mathrm{TC}$  there exists
$\tz \in \tB_{j-1}$ such that $\tz \sim \tx, \tw$. Since $\pi$ is a pentagon, $z$ does not belong to $\pi$.
If $z$ is adjacent to $y$ or to $u$, by applying ($T_{i+1}$) to $\tw$
or to $\tx$ we conclude that $\tz\sim \ty$ or $\tz\sim
\tu$. Consequently, we obtain a virtual 5-cycle $\tu\tx\tz\ty\tv$ or
$\tu\tz\tw\ty\tv$ with height $h(\tpi)-1$, and by induction hypothesis
we conclude that $\tu\sim\tv$. Suppose now that $z \nsim u,y$. Observe
that if $h(\tu) = j-1$ (respectively, $h(\ty) = j-1$), then by $\INC(\tv_0)$ we have $\tz\sim \tu$ (respectively, $\tz \sim \ty$)
and by \Rip we have $z \sim u$ (respectively, $z \sim y$). Consequently, we have $h(\tu) \geq j$ and $h(\ty) \geq j$.
Consider the ball $B_2(z,G)$. Since $u,y\in B_2(z,G)$, $u\nsim y$,
and $v\sim u,y,$ by the convexity of $B_2(z,G)$, we have $d(z,v)\le 2$.
Since  $z\nsim u,y$ and since $G$ does not contain squares, we
deduce that $d(z,v)=2$. Let $t$ be a common neighbor of $z$ and $v$.
By \Ri, there exists a corresponding $\ttt \sim \tz$ in
$\tB_{j}$. Then $\pi'=\ttt\tz\tw\ty\tv$ is a virtual 5-cycle. Since
$h(\tu)\ge h(\tx) = j$, $h(\ttt)\le j$, and $h(\tz)=j-1$, we have
$h(\tpi')\le h(\tpi)-1$. By induction hypothesis, $\ttt\sim
\tv$. Consequently, $\pi''=\tu\tx\tz\ttt\tv$ is a virtual 5-cycle and since
$h(\ty)\ge h(\tw) = j$, $h(\ttt)\le j$, and $h(\tz)=j-1$, we have
$h(\tpi'')\le h(\tpi)-1$. By induction hypothesis, $\tu\sim \tv$.

Now, let $\te=\tu\tx$.
In this case, $h(\tu) = h(\tx) = j$ and $h(\tw) = j-1$. Since
$\mathrm{TC}$ applies to $\tu\tx$, there exists $\tz \in \tB_{j-1}$
such that $\tz \sim \tu, \tx$.  By $\INC(\tv_0)$, $\tz\sim \tw$ and
thus $\tu\tz\tw\ty\tv$ is a virtual 5-cycle with height
$h(\tpi)-1$. By induction hypothesis, $\tu\sim \tv$, concluding the
analysis of Case 4.
\end{proof}

Now suppose that only $\mathrm{PC}$ applies to the edge $\te$. If
$\te=\tx\tw$, then by $\mathrm{PC}$ there exist
$\tw_1, \tw_2 \in \tB_{j-1}$, and $\tz\in \tB_{j-2}$ such that
$\tpi_1=\tx\tw_1\tz\tw_2\tw$ is a pentagon of $\tG_{i+1}$ (and thus
$\pi_1=xw_1zw_2w$ is a pentagon of $G$).  If $\te=\tu\tx$, then
$h(\tw) = j-1$ and by Lemma~\ref{lem-ind-PCdist2}, there exists
$\tw_1 \in \tB_{j-1}$ and $\tz\in \tB_{j-2}$ such that
$\tpi_2=\tx\tw_1\tz\tw\tu$ is a pentagon of $\tG_{i+1}$ (and thus
$\pi_2=xw_1zwy$ is a pentagon of $G$). Note that in both cases, the
image $e$ of $\te$ is a common edge of $\pi$ and of the new pentagon
$\pi_1$ or $\pi_2$. Denote the union $\pi\cup \pi_1$ or
$\pi\cup \pi_2$ by $U$.  Since the balls of radius 3 of $G$ are
convex, $U$ has diameter 2 or 3. Moreover, if $\te = \tu\tx$, then
$\pi$ and $\pi_2$ share three vertices and in this case $U$ has
diameter $2$.

\begin{claim} \label{claim:better5virtcycle} Any virtual  5-cycle of the form $\tpi'=\tp\ta\tz\tb\tq$ satisfies the induction hypothesis, thus $\tp\sim \tq$. \end{claim}

\begin{proof}
  First notice that
  $h(\tpi')=h(\tp)+h(\ta)+h(\tz)+h(\tb)+h(\tq)\le
  j+j-1+j-2+j-1+j=5j-4$ and the equality holds only if
  $h(\tp)=h(\tq)=j, h(\ta)=h(\tb)=j-1$, and $h(\tz)=j-2$.  Therefore
  $h(\tpi')=5j-4$ only if all edges of $\tpi'$ are vertical. In this
  case, the result holds by Case 3.  Now we will show that
  $h(\tpi)\ge 5j-4$.
First suppose that $\te=\tx\tw$.  Then $h(\tu)\ge h(\tx)-1=j-1$,
  $h(\ty)\ge h(\tw)-1=j-1$, and $h(\tv)\ge h(\ty)-1\ge j-2$, yielding
  $h(\tpi)\ge j-1+j+j+j-1+j-2=5j-4$.
Now suppose that $\te=\tu\tx$. Then $h(\tu)=h(\tx)=j$,
  $h(\tv)\ge h(\tu)-1$, $h(\tw)=j-1$, and
  $h(\ty)\ge h(\tw)-1=j-2$, we deduce that
  $h(\tpi)\ge j+j+j-1+j-2+j-1=5j-4$.
\end{proof}

\medskip\noindent
{\bf Case 5.} $\diam(U)=3$.

\begin{proof}
  In this case, necessarily $\te = \tx\tw$ and the vertices of
  $\pi \setminus e$ and $\pi_1\setminus e$ are pairwise distinct. By
  Lemma \ref{lem: deuxpent}, $\pi$ or $\pi_1$ has a universal vertex
  $t$.  First suppose that $t$ is a universal for $\pi_1$. By \Ri
  there exists $\ttt \sim \tz$ in $\tB_{j-1}$.  By ($R_{i+1}$), $\ttt$
  is adjacent to all vertices of $\tpi_1=\tx\tw_1\tz\tw_2\tw$.  This
  contradicts the assumption that only $\mathrm{PC}$ applies to $\te$.

  Now suppose that $t$ is a universal vertex for $\pi=uvywx$. Since
  $B_2(t,G)$ is convex in $G$ and $\diam(U)=3$, we have $d(t,z)=2$. Let
  $s$ be a common neighbor of $t$ and $z$. By \Ri there exists
  $\ts \in\tB_{j-1}$ adjacent to $\tz$.  Applying \Ri the second time
  we conclude that there exists $\ttt\in \tB_j$ such that
  $\ttt\sim \ts$. By Claim~\ref{claim:better5virtcycle} applied to the
  virtual 5-cycle $\tx\tw_1\tz\ts\ttt$, we have $\ttt \sim \tx$.
  Applying several times ($T_{i+1}$), we conclude that
  $\ttt\sim \tu,\tw,\ty,\tv$. Applying ($T_{i+1}$) once again to
  $\ttt\sim \tu,\tv$, we deduce that $\tu\sim \tv$.
\end{proof}

\medskip\noindent
{\bf Case 6.} $\diam(U)=2$.

\begin{proof}
We start with a special subcase, which is used to prove a useful claim.

\medskip\noindent
{\bf Subcase 6.1.} $h(\tw)=i-1$.

\medskip
Let $h(\tv)\ge h(\tu)$. Since $\max\{ h(\tx),h(\ty)\}\le h(\tw)+1=i$ and $\tpi$ contains a vertex of height $i+1$, necessarily $h(\tv)=i+1$.
Since $h(\tu)\ge h(\tv)-1$ and $\tpi$ contains horizontal edges (by Cases 1-2), either $h(\tu)=i=h(\ty), h(\tx)=h(\tw)=i-1, h(\tv)=i+1$ or
$h(\tu)=h(\tx)=i=h(\ty),h(\tw)=i-1,h(\tv)=i+1$. Thus the edge $\te=\tx\tw$ or $\te=\tu\tx$ is the unique horizontal edge of $\tpi$. 

First, let $\te=\tx\tw$. Since $\diam(U)=2$, $d(u,z)\le 2$. If $u=z$ or $u\sim z$, then  Lemma \ref{lemma: virtC4} applied to the path $\tz\tw_1\tx\tu$ implies
$\tz=\tu$ or $\tz\sim \tu$, which is impossible because $h(\tu)=i$ and $h(\tz)=i-3$. Thus $d(u,z)=2$ and let $s$ be a common neighbor of $u$ and $z$.
By \Ri there exists $\ts\sim \tz$ with $h(\ts)\le i-2$. The virtual 5-cycle $\tu\tx\tw_1\tz\ts$ is included $\Bi$, thus $\tu\sim \ts$ by ($S_i$), contrary
to $h(\tu)=i$ and $h(\ts)\le i-2$. Now, let $\te=\tu\tx$. Then $h(\tz)=i-2$. As $h(\tu)=i$, there exists $\tv'\in \Bip$ in the preimage of $v$ such that $\tv'\sim \tu$. If $h(\tv')\leq i$,
then setting $\pi'=\tv'\tu\tx\tw\ty$, we have $h(\tpi')<h(\tpi)$, and by ($S_i$), we get $\tv'\sim \ty$. Since $h(\tv') \leq i < h(\tv)$, $\ty$ has two distinct neighbors $\tv, \tv'$ that are mapped to $v$, contradicting \Rip. 
Assume now that $h(\tv')=i+1$. If $d(v,z)<2$, we immediately get a contradiction with Lemma \ref{lemma: virtC4} applied to $\tz\tw_1\tu\tv'$. Thus $d(v,z)=2$
and there exists $t \sim v,z$. By \Ri there is an associated $\ttt$ in the preimage of $t$ such that $\ttt \sim \tz$. Applying Case 2 to $\tv'\tu\tw_1\tz\ttt$, we get a
contradiction because $h(\tv')-h(\ttt)\geq 2$.

\begin{claim}\label{claim:virtC4} Let $\ta\tb\tc\td$ be a path of $\tG_{i+1}$ with $h(\td)\le i-1$. Then $d_{\tG_{i+1}}(\ta,\td)=d_G(a,d)$.
Moreover, if $d(a,d)=2$, for any $t\sim a,d$ in $G$, then there exists  $\ttt\sim\ta,\td$ in $\tG_{i+1}$ such that $f_{i+1}(\ttt) = t$.
\end{claim}

\begin{proof}
  By Lemma~\ref{lem: l313},
  $d_G(a,d) \leq d_{\tG_{i+1}}(\ta,\td) \leq 3$.  The cases
  $d_G(a,d)\in\{ 0,1\}$ follow from Lemma \ref{lemma: virtC4}.  Now,
  let $d_G(a,d)=2$ and let $t \sim a,d$. Since $\td\in \Biu$, by \Ri
  there exists $\ttt \in \Bi$ in the preimage of $t$ such that
  $\ttt \sim \td$. By \Rip there exists $\ta' \in \Bip$ in the
  preimage of $a$ such that $\ta' \sim \ttt$.  Consider the virtual
  5-cycle $\tpi'= \tb \tc \td \ttt \ta'$. Since $h(\td)\leq i-1$, we
  can apply Subcase 6.1 to $\tpi'$, whence $\ta'\sim \tb$.  By \Rip we
  must have $\ta=\ta'$, hence $\ttt \sim \ta, \td$.
\end{proof}

Since $\diam(U)=2$, in $G$ the vertex $z$ has distance at most 2 from all vertices of $\pi$.
Now we prove that in $\tG_{i+1}$  the vertex $\tz$ also has distance at most 2 from all vertices of $\tpi$.
Suppose first that $\te = \tx\tw$. Clearly,
$d_{\tG_{i+1}}(\tz,\tx)=d_{\tG_{i+1}}(\tz,\tw)=2$. Applying Claim
\ref{claim:virtC4} to the paths $\tu\tx\tw_1\tz$ and $\ty\tw\tw_2\tz$,
we conclude that $d_{\tG_{i+1}}(\tz,\tu)=d_G(z,u)\le 2$ and
$d_{\tG_{i+1}}(\tz,\ty)=d_G(z,y)\le 2$. Finally, we show that
$d_{\tG_{i+1}}(\tz,\tv)\le 2.$ If this is not the case, then
necessarily $d_{\tG_{i+1}}(\tz,\ty)=d_G(z,y)=2$ and by Claim
\ref{claim:virtC4} there exists $\ttt$ such that $\ttt \sim
\tz,\ty$. By Claim \ref{claim:virtC4} applied to the path
$\tv\ty\ttt\tz$ we conclude that
$d_{\tG_{i+1}}(\tz,\tv)=d_G(z,v)\le 2$.
Assume now that $\te = \tu\tx$. Since in this case, $\tz \sim \tw$,
we have $d(\tz,\ty) \leq 2$ and Applying Claim \ref{claim:virtC4} to
the path $\tv\ty\tw \tz$, we have
$d_{\tG_{i+1}}(\tz,\tv)=d_G(z,v)\le 2$.

Consequently,  in both cases $\tz$ is at distance at most $2$ from $\tu$ and $\tv$.
If $d_{\tG_{i+1}}(\tz,\tu) \leq 1$ or $d_{\tG_{i+1}}(\tz, \tv)\leq 1$, then  $\tu$ and $\tv$ can be connected by a path of length at most $3$ passing via $\tz$.
Applying Lemma \ref{lemma: virtC4} to this path, we conclude that $\tu\sim\tv$. Assume now that $d_{\tG_{i+1}}(\tz,\tu)=d_{\tG_{i+1}}(\tz,\tv)=2$.
Then there exist $\ttt \sim \tz,\tu$ and $\ts \sim \tz, \tv$. If $\ttt=\ts$, then we get $\tu \sim \tv$ by \Rip applied to $\ttt=\ts$.
Otherwise, consider the virtual 5-cycle $\tu\ttt\tz\ts\tv$. By Claim \ref{claim:better5virtcycle}, we conclude that $\tu\sim \tv$. This finishes the proof of Case 6.
\end{proof}

In all cases we proved that $\tu\sim \tv$, establishing ($S_{i+1}$) and concluding the proof of the proposition.
\end{proof}

\subsection{Universal Cover}
As for every $i\geq 0$, $\tG_i$ satisfies ($Q_i$), the graph $\tG$ satisfies $\INC(\tv_0)$ and $\TPC(\tv_0)$, so in particular by Lemma \ref{lem: triangle-pentagon-simply-connected}, the complex $\tX= X_{\Triangle, \pentagon}(\tG)$ is simply connected.

We recall that to end the proof of Theorem \ref{thm: triangle-pentagon}, we need to show that $\tX$ is a cover of $X$. Indeed, as there is a unique (up to isomorphism) simply connected cover of $X= X_{\Triangle, \pentagon}(G)$, this will imply that $\tX=X$. As $\tX$ can be constructed from any arbitrary basepoint $v_0\in V$, in particular $\tG$ satisfies $\INC(\tv)$ and $\TPC(\tv)$ for any vertex $\tv$ of $\tG$, so it has convex balls.

To complete the proof that $f := \bigcup_{i\geq 0}f_i : \tX \to X$ is
a covering map, we show that $f$ induces an isomorphism between the
(closed) star of a vertex of $\tX$ and the (closed) star of its image in
$X$.

\begin{lemma}
 \label{lem: etoiles}
 For every $\widetilde w \in \tX$, the map $f$ induces an isomorphism between the closed stars $\cSt(\tw,\widetilde{X})$ and $\cSt(w,X)$, where $w := f(\tw)$. Thus $f$ is a covering map.
\end{lemma}
\begin{proof} That $f$ defines a morphism from $\tX$ to $X$ follows from Lemmas \ref{lem: RipTip} and \ref{lem: morphism}. It remains to show that $f$ is bijective on  stars. To do so, we establish the stronger result that $f$ is bijective on closed stars. Let $i:=d(\vz, \tv)$. Then the restriction of $f$ on the set $\cSt(\tw,\widetilde{X})$ equals to the restriction of $f_{i+2}$ on this set. By \Rip, $f$ induces an isomorphism between $\tG\left[B_1(\tw,\tG)\right]$ and $G\left[B_1(w,G)\right]$.

Let $u \in \cSt(w,X)\setminus B_1(w,G)$. Then there exists a pentagon $uxwyv$ in $G$. By iterated applications of ($R_{i+2}$), we find an associated virtual 5-cycle  $\tu\tx\tw\ty\tv$ in $\tG$. By ($S_{i+2}$), $\tu\tx\tw\ty\tv$ is a pentagon of $\tG$, so  it is a cell of $\cSt(\tw,\widetilde{X})$. Thus we proved that $f|_{\cSt(\tw,\widetilde{X})}$ is onto $\cSt(w,X)$.

It remains to show that $f|_{\cSt(\tw,\widetilde{X})}$ is injective. Assume that there exist $\tu, \tup \in \cSt(\tw,\widetilde{X})$ such that $f(\tu) = f(\tup) = u$. By the previous remark, we cannot have $\tu, \tup \in B_1(\tw,\tG)$, so we can assume that $\tu\notin B_1(\tw,\tG)$. Let $\tx \sim \tu, \tw$. Then by ($R_{i+2}$) we must have $u \neq w$ and $u \nsim w$, hence $\tup \notin B_1(\tw,\tG)$ and there exists $\tz \sim \tup, \tw$. Set $x:= f(\tx)$ and $z:= f(\tz)$. By ($R_{i+2}$) we get $x\sim u, w$ and $z \sim u', w$.
 If $x = z$, then by ($R_{i+2}$) applied two times we must have $\tx = \tz$ and $\tu = \tup$. If $x\neq z$, then the 4-cycle $uxwz$ cannot be induced, so we must have $x\sim z$. Then by iterated applications of ($R_{i+2}$) we get $\tx \sim \tz$ and $\tup \sim \tx$ and eventually $\tu = \tup$. Hence we proved the injectivity of $f|_{\cSt(\tw,\widetilde{X})}$.
\end{proof}

 \section{Contractibility of Rips complexes}
\label{sec: dism}

In this section, we prove that the square $G^2$ of any CB-graph $G$ is
dismantlable and that the dismantling order can be obtained by a BFS
ordering of the vertices of $G$.  This implies that if $G$ is
locally-finite, then the clique complex $X(G^2)$ of $G^2$ is
contractible. Consequently, this shows that all Rips complexes
$X_k(G)$, $k \ge 2$ of a locally-finite CB-graph $G$ are
contractible.  

\subsection{Dismantlability of squares}
For several subclasses of weakly modular graphs, BFS
(Breadth-First-Search) and its refinements turn out to provide
orderings with interesting and strong properties, which can be used, for instance,
to prove contractibility of associated clique complexes. First, it was shown
in~\cite{Ch_bridged} that for locally finite  bridged graphs, any BFS ordering is a
dismantling order, showing in particular that the clique complexes of bridged
graphs are contractible.  Polat \cite{Po_bridged1} proved that arbitrary
connected graphs (even if they are not locally finite) admit a BFS ordering and,
extending the result of~\cite{Ch_bridged}, he showed that  BFS provides a dismantling
order for non-locally-finite bridged
graphs.  For weakly bridged graphs the same kind of results has been obtained
for specific BFS orderings. Namely, any LexBFS ordering of a locally
finite weakly bridged graph provides a dismantling
order~\cite{ChepoiOsajda}. In the case of non-locally-finite graphs, it is
not always possible to define a LexBFS ordering. However, for graphs
without infinite cliques, it was shown in~\cite{Bresaretal2013} that it is always
possible to define an ordering, intermediate between BFS and LexBFS, and called SimpLexBFS,
and it was shown that for weakly bridged graphs, any SimpLexBFS ordering is a
dismantling order. Notice also that the contractibility of Kakimizu complexes was
established by defining a BFS-like orderings of their vertices; for details,
see Section 5 of \cite{PrzSch}.

A vertex $x$ of a graph $G$ is \emph{dominated} by another vertex $y$ if the unit ball $B_1(y)$
includes $B_1(x).$  A graph $G$ is \emph{dismantlable} if
its vertices can be well-ordered $\prec$ so that, for each $v$  there is a neighbor $w$ of $v$ with
$w\prec v$ which dominates $v$ in the subgraph of $G$ induced by the vertices $u\preceq v$. The order $\prec$ is then called a \emph{dismantling order} of $G$.
Following Polat \cite{Po_bridged1}, a well-order $\preceq$ on the vertex set $V(G)$ of a graph $G$ is called a \emph{BFS order} if there exists a family
$\{ A_x: x\in V(G)\}$ of subsets of $V(G)$ such that, for
every vertex $x\in V(G)$,
\begin{enumerate}[{(S}1)]
\item $x\in A_x$;
\item if $x\preceq y$, then $(A_x,\preceq)$ is an initial segment of $(A_y,\preceq)$;
\item $A_x=A_{(x)}\cup N(x)$, where $A_{(x)}:=\{ x\}$ if $x$ is the least element of $(V(G),\preceq)$ and $A_{(x)}:=\bigcup_{y\prec x} A_y$ otherwise.
\end{enumerate}

\begin{lemma}[\cite{Po_bridged1}*{Lemma 3.6}]\label{l:Polat_BFS}  There exists a BFS order on the vertex set of any connected graph.
\end{lemma}

The vertex $x$ will be called the \emph{parent} of each vertex of $A_x\setminus A_{(x)}$. We will denote by $f$ the map from $V(G)$ to $V(G)$
such that $f(v)$ is the parent of $v$, for every $v\in V(G)$. The least element of $(V(G),\preceq)$ will be called the \emph{base-point}
and will be denoted by $v_0$ (by convention, we set $f(v_0)=v_0)$). Notice that like in the case of finite graphs, for every vertices $x$ and $y$ of $G$,
$x\preceq y$ implies $d(v_0,x)\le d(v_0,y)$, and $d(v_0,x)<d(x_0,y)$ implies $x\prec y$. In particular, $d(v_0,x)=d(v_0,f(x))+1$.
For two distinct vertices $x$ and $y$ of $G$, we set $\max\{ x,y\}=x$ if $y\prec x$ and $\max\{ x,y\}=y$ if $x\prec y$.

The main goal is to prove the following result:

\begin{theorem} \label{dismantl} Any BFS order of the vertices of a CB-graph $G$ is a dismantling order of its square $G^2$.
\end{theorem}

\begin{proof} We start with several properties of BFS on all graphs.
  Let $b$ be a basepoint of $G$ and let $\preceq$ be the basepoint
  (partial) order of the vertices of $G$: for two vertices $u,v$ of
  $G$ we set $u\preceq v$ if and only if $u \in I(b,v)$. For a vertex
  $v$, let $F(v)=\{u \in V: d(b,u)=d(b,v)+d(v,u)\}$ and call $F(v)$
  the \emph{filter} of $v$ with respect to this basepoint order
  $\preceq$. Let $\prec$ be a BFS (total) order of the vertices of $G$
  with basepoint $b$. Clearly, $\prec$ is a linear extension of
  $\preceq$. For a vertex $v$ of $G$, let $f(v)$ be the \emph{parent}
  of $v$ defined by BFS, $f^2(v)=f(f(v))$ be the \emph{grandparent} of
  $v$, and so on, let $f^i(v)=f(f^{i-1}(v))$; if $u = f^i(v)$ for some
  $i$, we say that $u$ is an \emph{ascendant} of $v$ and that $v$ is a
  \emph{descendant} of $u$.

\begin{lemma}If $x\prec y\prec z,$ $d(b,x)=d(b,z)=:k$, and $f^2(x)=f^2(z)=:s,$
  then $f^2(y)=s$.
\end{lemma}

\begin{proof}
  Suppose $f^2(y)=s'\ne s$. Since $x\prec y\prec z$ and
  $d(b,x)=d(b,z)=k,$ necessarily $d(b,y)=k$ and
  $f(x)\prec f(y)\prec f(z).$ Since $f^2(x)=f^2(z)=s$, all vertices
  $t$ with $f(x)\prec t\prec f(z)$ will necessarily have $s$ as their
  parent. In particular, $f(f(y))=s$.
\end{proof}

\begin{lemma}\label{BFS0}
  Let $x,y,v$ be vertices of $G$ with $d(b,x)=d(b,y)$, $x=f^i(v)$, and
  $v\in F(y)$. Then $x\prec y$.
\end{lemma}

\begin{proof}  Suppose $y\prec x$.  Since $v\in F(y)$ and $d(b,x)=d(b,y)$, we conclude that $d(y,v)=d(x,v)=i$.
Let $P=(y=z_{i},z_{i-1},\ldots,z_1,z_0=v)$ be a shortest path from $y$ to $v$ and $Q=(x=f^{i}(v),f^{i-1}(v),\ldots,f^1(v),v)$ be the shortest path
from $x$ to $v$ consisting of the ascendants of $v$. We can suppose that $P$ and $Q$ intersects only in  $v$, otherwise
we can replace $v$ by a closest to $x$ and $y$ vertex from the intersection. Since $y\prec x$ and $x=f(f^{i-1}(v))$, we deduce that
$z_{i-1}\prec f^{i-1}(v)$. Continuing this way, we get  $z_{i-j}\prec f^{i-j}(v)$ for any $j<i$. In particular,
$z_1\prec f^1(v)=f(v)$, contradicting that $f(v)$ is the parent of $v$.
\end{proof}

\begin{lemma} \label{BFS}  Let $u,v$ be vertices of $G$ with $k:=d(b,v)-d(b,u)\ge 0$, $u\prec f^k(v)$, and $v\in F(f^i(u))$. Then $f^{k+i}(v)=f^i(u)$.
\end{lemma}

\begin{proof} Let $x=f^i(u)$ and $y=f^{i+k}(v)=f^i(f^k(v))$. Suppose $y\ne x$. Since $u\prec f^k(v)$, necessarily $x\prec y$.
Since $v\in F(x)$ and $d(v,y)=d(v,x)$, necessarily $d(b,x)=d(b,y)$. But this contradicts Lemma \ref{BFS0}.
\end{proof}

In the remaining results, we suppose that $G$ is a CB-graph.

\begin{lemma} \label{2-fellow-traveller} Let $u\prec v$.
If $d(b,u)=d(b,v)=k$ and $u\sim v$, then either $f(u)=f(v),$ or $f(v)\sim u,f(u),$ or $f(u)\nsim f(v)$ and $f^2(u)=f^2(v).$
\end{lemma}

\begin{proof}
  Observe that if $f(u) \neq f(v)$, we have $f(u) \sim f(v)$ if and
  only if $u \sim f(v)$. Indeed, if $u \sim f(v)$, then
  $f(v),f(u)\in B_{k-1}(b)$ and $u\notin B_{k-1}(b),$ thus the
  convexity of $B_{k-1}(b)$ implies that $f(v)\sim f(u)$. Conversely,
  if $f(v)\sim f(u),$ since $G$ does not contain induced 4-cycles and
  since $v\nsim f(u)$, as $f(v) \neq f(u)$ and $u \prec v$, we
  conclude that $f(v)\sim u$.

  We prove the lemma by induction on $k$, the cases $k=1,2$ being
  trivial. Let $k\ge 3$.  Suppose by way of contradiction that
  $f(u)\ne f(v)$, that $f(v) \nsim u,f(u)$ and that
  $f^2(u)\ne f^2(v)$.
By Lemma \ref{BFS}, since $f^2(v)\ne f^2(u)$ and $u\prec v$,
  necessarily $d(v,f^2(u))=3$. Since $f(u),f(v)\in B_{k-1}(b),$
  $B_{k-1}(b)$ is convex, and $f(u)\nsim f(v)$, necessarily
  $d(f(u),f(v))=2$, and consequently, $d(f^2(u),f(v)) \leq 3$.  If
  $d(f^2(u),f(v)) \leq 2$, then since $u,f(v)\in B_2(f^2(u))$ and
  $v\notin B_2(f^2(u))$, the convexity of $B_2(f^2(u))$ implies that
  $f(v)\sim u$, which contradicts our assumption. Thus, further we
  will suppose that $d(f^2(u),f(v))=d(f^2(u),v)=3$.  Let $w$ be a
  common neighbor of $f(u),f(v)$ in $B_{k-1}(b)$. If $d(b,w)=k-2,$
  then since $f^2(u),w\in B_{k-2}(b)$ and $f(u)\notin B_{k-2}(b),$ the
  convexity of $B_{k-2}(b)$ implies that $d(w,f^2(u)) \leq 1$,
  contrary to the assumption that $d(f^2(u),f(v))=3$.  Thus further we
  can assume that $d(b,w)=k-1$ and that $f^2(u)\nsim w$.

  Now, we apply the induction assumption to $f(u)\sim w$ and
  $w\sim f(v)$. We distinguish two cases.

  \medskip\noindent
  {\bf Case 1.} $f(w)\sim f(u),f^2(u)$.

  \begin{proof}
    Since $d(f^2(u),f(v))=3$, necessarily $f(w)\nsim f(v).$ Moreover,
    as $f(v)\nsim u$, we can assume that $d(f(w),v)=2$, otherwise we
    get a contradiction with the convexity of $B_2(f(w))$.  First
    suppose that $f^2(v)\sim w,f(w).$ Then $f(w)\prec f^2(v)$, so
    $w\prec f(v)$. Since $d(f(w),v)=2,$ then $v\in F(f(w))$ and, Lemma
    \ref{BFS} yields $f^2(v)=f(w)$, a contradiction.  Thus
    $f^2(v)\nsim w,f(w)$ and by induction hypothesis
    $f^3(v)=f^2(w)$. Since $f^2(v),f(w)\in B_2(v)$ and
    $f^2(w)=f^3(v)\notin B_2(v),$ the convexity of $B_2(v)$ implies
    that $f^2(v)\sim f(w),$ a contradiction.
\end{proof}

\medskip\noindent
{\bf Case 2.} $x:=f^2(w)=f^3(u)$ and $f(w)\nsim f^2(u)$.

\begin{proof}
  Observe that $d(x,u) = 3$ and that $d(x,f(v)) \leq 3$. Since
  $u \nsim f(v)$, the convexity of $B_3(x)$ implies that
  $d(x,v) \leq 3$. Lemma \ref{BFS} implies that $x=f^3(v)$ and
  consequently, $d(x,f(v) = 2$.

  Note also that $d(f(v),f^2(u)) \geq 3$. Indeed, if
  $d(f(v),f^2(u)) \leq 2$, then since $u,f(v) \in B_2(f^2(u))$ and
  $u \nsim f(v)$, the convexity of $B_2(f^2(u))$ implies that
  $v \in B_2(f^2(u))$, and by Lemma \ref{BFS}, we get that
  $f^2(v) = f^2(u)$.

  Therefore, we have $x,f(u) \in B_2(f(v))$ and
  $f^2(u) \notin B_2(f(v))$ and thus $x \sim f(u)$, a
  contradiction.
\end{proof}

This concludes the proof of the lemma.
\end{proof}

Now, we are ready to complete the proof of the theorem. By induction
on the labels of vertices of $G$, we will prove that the BFS ordering
of $G$ is a domination order of $G^2$. Let $u\prec v$ and
$d(u,v)\le 2$. We have to prove that $f(v)$ is adjacent to $u$ in
$G^2,$ i.e., that $d(f(v),u)\le 2$ in $G$.  This is obviously true if
$u$ is adjacent to $v$ in $G$. Thus, further let $d(v,u)=2$ and let
$w$ be a common neighbor of $u$ and $v$.  Since $u\prec v,$
$d(b,u)\le d(b,v)=k$. If $w\in B_{k+1}(b)$, then the convexity of
$B_k(b)$ implies that $u \sim v$, a contradiction.  If
$w\in B_{k-1}(b)$, then since $v\sim f(v),w$ and $v\notin B_{k-1}(b),$
the convexity of $B_{k-1}(b)$ implies that $f(v)=w$ or $f(v)\sim
w$. Then obviously $d(f(v),u)\le 2$ and we are done. So, further let
$d(b,w)=k$. If $u\in B_{k-1}(b),$ again the convexity of $B_{k-1}(b)$
implies that $d(f(v),u)\le 2.$ So, we can suppose that $d(b,u)=k$. If
$f(v)=f(w),$ then again $d(f(v),u)\le 2.$ So, $f(v)\ne f(w)$ and
$f(v)\nsim w$.

We distinguish the following cases:

\medskip\noindent
{\bf Case 1:} $f(v)\sim f(w).$

\begin{proof}
Since $f(v)\nsim w,$ we conclude that $f(w)\sim v$ and thus $v\prec w$, $f(v)\prec f(w)$. Since $u\prec v\prec w,$ we obtain that $f(w)\ne f(u)$ and
$f(u)\prec f(w)$. If $f(u)\sim f(w),$ then $f(w)\sim u$ and thus $d(f(v),u)\le 2.$ So, assume that $f(u)\nsim f(w)$. Them Lemma \ref{2-fellow-traveller}
implies that $f^2(u)=f^2(w)=:x.$ Since $u\prec v\prec w$ and $f^2(u)=f^2(w),$ by BFS $f^2(v)=x$. Now, if $d(f(v),u)=3,$ then $v\sim x=f^2(v)$ by the convexity of the ball $B_2(u)$. Since this is impossible, $d(f(v),u)\le 2$.
\end{proof}

\medskip\noindent
{\bf Case 2:} $f(v)\nsim f(w).$

\begin{proof}
By Lemma \ref{2-fellow-traveller}, $f^2(v)=f^2(w)=:x$. First, let $f(w)\sim u$. If $d(f(v),u)=3,$ since $v,x\in B_2(u)$ and $f(v)\notin B_2(u),$
by convexity of $B_2(u)$ we conclude that $v\sim x=f^2(v),$ which is impossible. So $f(w)\nsim u$. In particular we have $f(w)\ne f(u).$ First suppose that $f(w)\sim f(u).$
Then $f(u)\sim w$ and thus $w\prec u$. Hence $w\prec u\prec v$. Since $f^2(v)=f^2(w),$ BFS implies that $f^2(u)=x$, i.e., $x\sim f(u)$. Again, if $d(f(v),u)=3,$ since
$v,x\in B_2(u)$ and $f(v)\notin B_2(u),$ we conclude that $v\sim x=f^2(v),$ a contradiction.
Finally, let $f(w)\nsim f(u)$. By  Lemma \ref{2-fellow-traveller}, $f^2(u)=f^2(w)=x$. Again, if $d(f(v),u)=3,$ then the fact that $v\nsim x$ and $w\nsim f(u)$
leads to a contradiction with the convexity of the balls $B_2(f(v))$ and $B_2(u)$.
\end{proof}

This concludes the proof of the theorem.
\end{proof}

\begin{remark} Most of known classes of dismantlable graphs (systolic, weakly systolic, Helly) are weakly modular. Thickenings  of median and swm-graphs are Helly graphs. Therefore, one may ask if the 
powers (Rips complexes) of CB-graphs belong to one of these classes of graphs, in particular, whether they are weakly modular. 
The graphs $G_2, G_3, G_4, \ldots$ from Figure~\ref{fig: G2 nonwm} show that this is not the case: for any $k$, the $k$th power $G_k^k$ of $G_k$ does not satisfy the 
quadrangle condition for the vertices $u,v,w,y$ as indicated in the figure. 
\end{remark}
\tikzexternaldisable
  \begin{figure}[h]
    \centering
    \begin{tikzpicture}[scale=0.65]

    \begin{scope}[xshift=-6cm]
     \def\py{0.866}
     \def\k{2}
     \tikzstyle{every node}=[draw, circle,fill=black,minimum size=4pt,
                            inner sep=0pt]

    \foreach \y in {1,..., \k}{
        \foreach \x in {1,...,\y}{
            \pgfmathsetmacro{\cx}{\x-(\y +1)*0.5}
            \pgfmathsetmacro{\cy}{(\y-1)*\py}
            \ifnum \y=1
               \node [fill=red, label=below: $u$] (x-\x-\y) at (\cx,\cy) {};
            \else
               \node (x-\x-\y) at (\cx,\cy) {};
            \fi
            \ifnum \x>1
                \pgfmathtruncatemacro{\px}{\x -1}
                \pgfmathtruncatemacro{\py}{\y -1}
                \draw[thick] (x-\x-\y) -- (x-\px-\y);
                \draw[thick] (x-\x-\y) -- (x-\px-\py);
            \fi
            \ifnum \x<\y
                \pgfmathtruncatemacro{\py}{\y -1}
                \draw[thick] (x-\x-\y) -- (x-\x-\py);
            \fi
        }
    }
    \pgfmathtruncatemacro{\nk}{\k+1}
    \foreach \y in {0,...,\nk}{
        \pgfmathtruncatemacro{\hy}{(\k +1)-\y}
        \foreach \x in {0,...,\hy}{
             \pgfmathsetmacro{\cy}{(\y+\k-1)*\py}
             \pgfmathsetmacro{\cx}{\x-(\hy)*0.5}
             \ifnum \y=0
                \ifnum \x = 0
                   \node [fill=red, label=left: $x$] (xx-\x-\y) at (\cx,\cy)
                   {};
                \else
                   \ifnum \x = \hy
                      \node [fill=red, label=right: $y$] (xx-\x-\y) at (\cx,\cy)
                      {};
                   \else
                      \node (xx-\x-\y) at (\cx,\cy) {};
                      \fi
                \fi
             \else
                \node (xx-\x-\y) at (\cx,\cy) {};
             \fi
\ifnum \x>0
            \pgfmathtruncatemacro{\px}{\x -1}
            \draw[thick] (xx-\x-\y) -- (xx-\px-\y);
        \fi
        \ifnum \y>0
            \pgfmathtruncatemacro{\nx}{\x +1}
            \pgfmathtruncatemacro{\py}{\y -1}
            \draw[thick] (xx-\x-\y) -- (xx-\x-\py);
            \draw[thick] (xx-\x-\y) -- (xx-\nx-\py);
        \fi
        }
    }
    \pgfmathtruncatemacro{\pk}{\k-1}
    \foreach \y in {0,...,\pk}{
      \pgfmathsetmacro{\cy}{(\y+2*\k)*\py}
      \ifnum \y=\pk
      \node [fill=red, label=above: $v$] (v-\y) at (0,\cy) {};
      \else
      \node (v-\y) at (0,\cy) {};
      \fi
\ifnum \y>0
            \pgfmathtruncatemacro{\py}{\y-1}
            \draw[thick] (v-\py) -- (v-\y);
        \fi
    }

    \node (lab) at (0, -1.5)[draw=white, fill=white, opacity =0, text opacity=1] {$G_2$};

  \end{scope}

    \begin{scope}
     \def\py{0.866}
     \def\k{3}
     \tikzstyle{every node}=[draw, circle,fill=black,minimum size=4pt,
                            inner sep=0pt]

    \foreach \y in {1,..., \k}{
        \foreach \x in {1,...,\y}{
            \pgfmathsetmacro{\cx}{\x-(\y +1)*0.5}
            \pgfmathsetmacro{\cy}{(\y-1)*\py}
            \ifnum \y=1
               \node [fill=red, label=below: $u$] (x-\x-\y) at (\cx,\cy) {};
            \else
               \node (x-\x-\y) at (\cx,\cy) {};
            \fi
            \ifnum \x>1
                \pgfmathtruncatemacro{\px}{\x -1}
                \pgfmathtruncatemacro{\py}{\y -1}
                \draw[thick] (x-\x-\y) -- (x-\px-\y);
                \draw[thick] (x-\x-\y) -- (x-\px-\py);
            \fi
            \ifnum \x<\y
                \pgfmathtruncatemacro{\py}{\y -1}
                \draw[thick] (x-\x-\y) -- (x-\x-\py);
            \fi
        }
    }
    \pgfmathtruncatemacro{\nk}{\k+1}
    \foreach \y in {0,...,\nk}{
        \pgfmathtruncatemacro{\hy}{(\k +1)-\y}
        \foreach \x in {0,...,\hy}{
             \pgfmathsetmacro{\cy}{(\y+\k-1)*\py}
             \pgfmathsetmacro{\cx}{\x-(\hy)*0.5}
             \ifnum \y=0
                \ifnum \x = 0
                   \node [fill=red, label=left: $x$] (xx-\x-\y) at (\cx,\cy)
                   {};
                \else
                   \ifnum \x = \hy
                      \node [fill=red, label=right: $y$] (xx-\x-\y) at (\cx,\cy)
                      {};
                   \else
                      \node (xx-\x-\y) at (\cx,\cy) {};
                      \fi
                \fi
             \else
                \node (xx-\x-\y) at (\cx,\cy) {};
             \fi
\ifnum \x>0
            \pgfmathtruncatemacro{\px}{\x -1}
            \draw[thick] (xx-\x-\y) -- (xx-\px-\y);
        \fi
        \ifnum \y>0
            \pgfmathtruncatemacro{\nx}{\x +1}
            \pgfmathtruncatemacro{\py}{\y -1}
            \draw[thick] (xx-\x-\y) -- (xx-\x-\py);
            \draw[thick] (xx-\x-\y) -- (xx-\nx-\py);
        \fi
        }
    }
    \pgfmathtruncatemacro{\pk}{\k-1}
    \foreach \y in {0,...,\pk}{
        \pgfmathsetmacro{\cy}{(\y+2*\k)*\py}
      \ifnum \y=\pk
      \node [fill=red, label=above: $v$] (v-\y) at (0,\cy) {};
      \else
      \node (v-\y) at (0,\cy) {};
      \fi
\ifnum \y>0
            \pgfmathtruncatemacro{\py}{\y-1}
            \draw[thick] (v-\py) -- (v-\y);
        \fi
    }
    \node (lab) at (0, -1.5)[draw=white, fill=white, opacity =0, text opacity=1] {$G_3$};
    \end{scope}

    \begin{scope}[xshift=7cm]
     \def\py{0.866}
     \def\k{4}
     \tikzstyle{every node}=[draw, circle,fill=black,minimum size=4pt,
                            inner sep=0pt]

    \foreach \y in {1,..., \k}{
        \foreach \x in {1,...,\y}{
            \pgfmathsetmacro{\cx}{\x-(\y +1)*0.5}
            \pgfmathsetmacro{\cy}{(\y-1)*\py}
            \ifnum \y=1
               \node [fill=red, label=below: $u$] (x-\x-\y) at (\cx,\cy) {};
            \else
               \node (x-\x-\y) at (\cx,\cy) {};
            \fi
            \ifnum \x>1
                \pgfmathtruncatemacro{\px}{\x -1}
                \pgfmathtruncatemacro{\py}{\y -1}
                \draw[thick] (x-\x-\y) -- (x-\px-\y);
                \draw[thick] (x-\x-\y) -- (x-\px-\py);
            \fi
            \ifnum \x<\y
                \pgfmathtruncatemacro{\py}{\y -1}
                \draw[thick] (x-\x-\y) -- (x-\x-\py);
            \fi
        }
    }
    \pgfmathtruncatemacro{\nk}{\k+1}
    \foreach \y in {0,...,\nk}{
        \pgfmathtruncatemacro{\hy}{(\k +1)-\y}
        \foreach \x in {0,...,\hy}{
             \pgfmathsetmacro{\cy}{(\y+\k-1)*\py}
             \pgfmathsetmacro{\cx}{\x-(\hy)*0.5}
             \ifnum \y=0
                \ifnum \x = 0
                   \node [fill=red, label=left: $x$] (xx-\x-\y) at (\cx,\cy)
                   {};
                \else
                   \ifnum \x = \hy
                      \node [fill=red, label=right: $y$] (xx-\x-\y) at (\cx,\cy)
                      {};
                   \else
                      \node (xx-\x-\y) at (\cx,\cy) {};
                      \fi
                \fi
             \else
                \node (xx-\x-\y) at (\cx,\cy) {};
             \fi
\ifnum \x>0
            \pgfmathtruncatemacro{\px}{\x -1}
            \draw[thick] (xx-\x-\y) -- (xx-\px-\y);
        \fi
        \ifnum \y>0
            \pgfmathtruncatemacro{\nx}{\x +1}
            \pgfmathtruncatemacro{\py}{\y -1}
            \draw[thick] (xx-\x-\y) -- (xx-\x-\py);
            \draw[thick] (xx-\x-\y) -- (xx-\nx-\py);
        \fi
        }
    }
    \pgfmathtruncatemacro{\pk}{\k-1}
    \foreach \y in {0,...,\pk}{
        \pgfmathsetmacro{\cy}{(\y+2*\k)*\py}
      \ifnum \y=\pk
      \node [fill=red, label=above: $v$] (v-\y) at (0,\cy) {};
      \else
      \node (v-\y) at (0,\cy) {};
      \fi
\ifnum \y>0
            \pgfmathtruncatemacro{\py}{\y-1}
            \draw[thick] (v-\py) -- (v-\y);
        \fi
    }

    \node (lab) at (0, -1.5)[draw=white, fill=white, opacity =0, text opacity=1] {$G_4$};
  \end{scope}

    \end{tikzpicture}
    \caption{The graphs $G_2, G_3, G_4$ are CB-graphs whose powers are
      not weakly modular. For any $k \geq 2$, in $G_k^k$, $u$ is
      adjacent to $x, y$ and at distance $3$ from $v$, and the
      vertices $v, x$, and $y$ are pairwise at distance $2$ but they
      do not have a common neighbor.}\label{fig: G2 nonwm}
\end{figure}
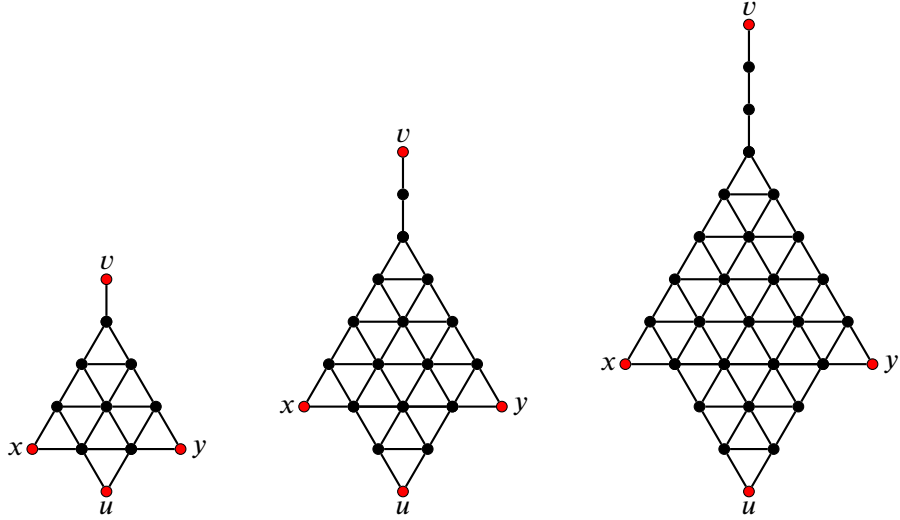

\begin{remark}
General CB-graphs are not
dismantlable (even if the main result of this section shows that their squares are dismantlable). A finite graph $G$ is \emph{dismantlable to a subgraph} $H$, if $G$ and $H$ satisfy the following conditions:
(1) there exists a total  ordering $v_k,\ldots,v_1$ of the vertices of $V(G)\setminus V(H)$ such that each vertex $v_i$ is dominated in the subgraph $G_i$ induced by $\{ v_i,v_{i-1},\ldots,v_1\}\cup V(H)$
either by a vertex $v_j$ with $j<i$ or by a vertex of $H$ and (2) $H$ does not contain any dominated vertex. It is well-known that all subgraphs $H$ of a graph $G$ to which $G$ is dismantlable are isomorphic \cite{Hell_Nesetril}
and any such subgraph is called the \emph{core} of $G$.  One can easily show that if $G$ is a CB-graph, then all intermediate subgraphs $G_i$ are also CB-graphs, thus the core $H$ of $G$ is a CB-graph.
Since weakly systolic graphs are dismantlable \cite{ChepoiOsajda}, their cores are trivial (a single vertex).
On the other hand, the core of any triangle-free CB-graph $G$ in which each 2-connected component contains a cycle is the graph $G$ itself. By Proposition \ref{thm: Moore}, each 2-connected component of such a graph
is a Moore graph and thus has diameter 2. One can ask if this property extends to the cores of all finite CB-graphs:
\end{remark}

\begin{question} Is it true that the core of any finite CB-graph is a CB-graph in which all 2-connected components have diameter 1 or 2?
\end{question}

\subsection{Rips complexes and stabilized sets}
We continue with several consequences of Theorem \ref{dismantl}.

\begin{corollary} \label{Rips}  For any $k\ge 2$, the Rips complex $X_{k}(G)$ of a locally-finite CB-graph is contractible. The maximal
simplices of $X_{k}(G)$ define convex sets of diameter $k$ of $G$.
\end{corollary}

\begin{proof} Note that $X_{k}(G)=X(G^k)$. Since $G^2$ is dismantlable, $G^k$ is also dismantlable for every $k \geq 2$.  Indeed, pick any two vertices $u,v\in V$ such that $v$ dominates $u$ in $G^2$. Pick any vertex $x$ such that $d(u,x)\le k$ and let $y$ be a vertex at distance 2 from $u$ on a shortest $(u,x)$-path of $G$. Since $v$ dominates $u$ in $G^2$, $d(v,y)\le 2$, thus by triangle inequality we obtain that $d(v,x)\le k$. Consequently, $v$ dominates $u$ also in $G^k$. Since $G$ is locally-finite, all cliques of $G^k$ are finite. Thus $X_{k}(G)$ is contractible, as the flag complex of a dismantlable graph without infinite cliques is contractible.

Now, let $A$ be a maximal by inclusion simplex of $X_{k}(G)$. Then $\diam(A)\le k$. Since in CB-graphs the equality  $\diam(\conv(A))=\diam(A)$ holds \cite{SoCh}, the maximality of $A$ implies that $\conv(A)=A$.
\end{proof}

\begin{remark} The second assertion of Corollary \ref{Rips} implies that the Rips complex $X_k(G)$ of a CB-graph can be viewed as the thickening of
$G$ with respect to all convex sets of diameter at most $k$. Other similar thickening operations have been used in the context of Helly graphs in \cite{ChChGeHiOs}.
\end{remark}

It is known (see for example \cite{Hell_Nesetril}) that for every dismantlable finite graph $G$ and every graph homomorphism
$f:G\to G$, there exists a (nonempty) clique $C$ in $G$ stabilized by $f$, i.e. such that $f(C)=C$. Hence we immediately get the following corollary:

\begin{corollary}
 \label{cor: fixpoint}
 Let $G$ be a finite graph with convex balls and $f: G \to G$ be a graph homomorphism. Then there exists a convex set of diameter at most $2$ in $G$ which is stabilized by $f$.
\end{corollary}

Corollary \ref{cor: fixpoint} is tight in the sense that one cannot
hope to find a stabilized clique for any homomorphism.  Indeed a
simple counter-example is the $5$-cycle, which is not dismantlable and
which have the cyclic permutation of order $5$ in its automorphism
group stabilizing no clique. The following example shows that the
5-cycle is not the only obstruction to get such a property.

\begin{example}\label{ex-circulant}
  We describe now a graph $H$ with convex balls such that there exists
  an automorphism $f:H\to H$ that does not stabilize a clique or a
  5-cycle. See Figure \ref{fig: fixpointctre-ex} for an
  illustration. Let
  $V(H):=\sg{a_1, a_2, a_3, b_1, b_2, b_3, c_1, c_2, c_3}$ and define
  the edge set of $H$ as follows\footnote{One may observe that $H$ is
    isomorphic to the circulant graph
    $\mathrm{Cay}(\mathbb{Z}_9, \sg{1,2})$.}:

\begin{gather*}
E(H):=\sg{a_ib_i, i \in \sg{1,2,3}}\cup \sg{a_{i+1}b_i, i \in \sg{1,2,3}}\cup \sg{a_ic_i, i \in \sg{1,2,3}}\\
\cup \sg{a_{i+1}c_i, i \in \sg{1,2,3}}\cup\sg{b_ic_i, i \in \sg{1,2,3}}\cup\sg{b_{i+1}c_i, i \in \sg{1,2,3}},
\end{gather*}
where the operations on the indices are done modulo $3$.  Observe that
$H$ has diameter $2$ and that $H$ has no induced 4-cycle. This implies
that $H$ has convex balls.  Now consider the automorphism $f$ defined
by $f(a_i):=a_{i+1}$, $f(b_i):= b_{i+1}$ and $f(c_i):= c_{i+1}$ for
every $i\in \sg{1,2,3}$ (again, every addition is done modulo $3$).
Note that $f$ cannot stabilize any clique as every vertex is nonadjacent to
its image. Moreover every 5-cycle of $H$ must have at least one of the
$a_i$'s, one of the $b_i$'s and one of the $c_i$'s in its vertex
set. By this observation, it follows that $f$ cannot stabilize a 5-cycle of
$H$.
\end{example}

  \tikzexternaldisable
  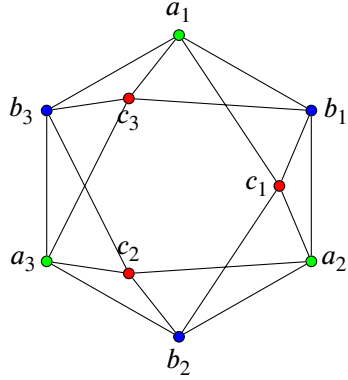
\begin{figure}[h]
    \centering
    \begin{tikzpicture}[scale=2]
    \tikzstyle{every node}=[draw,circle,fill=green,minimum size=4pt,
                            inner sep=0pt]

    \node (1) at (0,1) [label=above: $a_1$] {};
    \node (2) at (0.87,-0.5) [label=right: $a_2$] {};
    \node (3) at (-0.87,-0.5) [label=left: $a_3$] {};
    \tikzstyle{every node}=[draw,circle,fill=blue,minimum size=4pt,
                            inner sep=0pt]
    \node (4) at (0.87,0.5) [label=right: $b_1$] {};
    \node (5) at (0,-1) [label=below: $b_2$] {};
    \node (6) at (-0.87,0.5) [label=left: $b_3$] {};
    \tikzstyle{every node}=[draw,circle,fill=red,minimum size=4pt,
                            inner sep=0pt]
    \node (7) [label=left: $c_1$] at (0.66,0) {};
    \node (8) at (-0.33,-0.58) [label=above: $c_2$] {};
    \node (9) at (-0.33,0.58) [label=below: $c_3$] {};

    \draw (1) -- (4) -- (2) -- (5) -- (3) -- (6) -- (1) -- (7) -- (5) -- (8) -- (6) -- (9) -- (4) -- (7) -- (2) -- (8) -- (3) -- (9) -- (1);

    \end{tikzpicture}
    \caption{The graph $H$ of Example~\ref{ex-circulant}
}  \label{fig:
      fixpointctre-ex}
\end{figure}

 \section{Biautomaticity}\label{sec: NCP}
In this section, we prove that CB-groups (i.e., groups acting geometrically on CB-graphs) are biautomatic.
For this,  we  construct in a canonical way clique-paths between all pairs of  vertices of a graph with $2$-convex balls. In CB-graphs, we characterize those normal clique-paths locally. Then we prove that in CB-graphs, the normal clique-paths satisfy the 2-sided fellow traveller property. Using this property and the local definition of normal clique-paths, in the last subsection we prove biautomaticity of CB-groups. 

\subsection{Clique-paths: definition and existence}

Following \cite{ChChGeHiOs}, for a set $S$ of vertices of a graph
$G = (V,E)$ and an integer $k \ge 0$, let
$B^*_k(S)=\bigcap_{s\in S} B_k(s)$.  If $S$ is a clique, then
$B^*_1(S)$ is the union of $S$ and the set of vertices adjacent to all
vertices in $S$. Notice also that if $S \subseteq S'$, then
$B^*_k(S')\subseteq B^*_k(S)$.

For two vertices $u,v$ of a graph $G$, a \emph{clique-path} is a
sequence of cliques
$(\sg{u} = C^{(0)}, C^{(1)}, \ldots, C^{(k)}=\sg{v})$ such that any
two consecutive cliques $C^{(i)}$ and $C^{(i+1)}$ are disjoint and
their union induces a clique of $G$; $u$ and $v$ are respectively
called the \emph{source} and the \emph{sink} of this clique-path.

Now we define special clique-paths between all pairs of vertices of $G$.
Let $u,v$ be two arbitrary vertices of a graph $G$ and let $d(u,v)=k$.
For each $i$ running from $k$ to 0 we inductively define the sets
$C_{(u,v)}^{(i)}\subseteq S_i(u)\cap I(u,v)$ by setting $C_{(u,v)}^{(k)}:=\sg{v}$ and
\[C_{(u,v)}^{(i-1)}:=B^*_1\left(C_{(u,v)}^{(i)}\right)\cap
  B_{i-1}(u)\] for any $i<k$. We also set $C_{(u,v)}^{(i)}:= \sg{v}$
for any $i \geq k+1$. In the next lemma, we establish that each set
$C_{(u,v)}^{(i)}$ is a clique of $G$ and that  $\gamma_{(u,v)}:= (\sg{u} = C_{(u,v)}^{(0)}, C_{(u,v)}^{(1)}, \ldots,
C_{(u,v)}^{(k)}=\sg{v})$ is a clique-path. Note that these clique-paths are directed, since the cliques of
$\gamma_{(u,v)}$ and $\gamma_{(v,u)}$ are not the same in general.

\begin{lemma}\label{lem: INCclique}
Let $G$ be a graph with $2$-convex balls and let $u,v\in V$ with $d(u,v)=k$. For any $i\in \sg{0,\ldots,k}$, the set $C_{(u,v)}^{(i)}$ is nonempty and defines a clique of $G$. Furthermore,
$C_{(u,v)}^{(i)}\cap C_{(u,v)}^{(i+1)}=\varnothing$ and $C_{(u,v)}^{(i)}\cup C_{(u,v)}^{(i+1)}$ is a clique for any $i=0,\ldots,k-1$.
\end{lemma}

\begin{proof}
We proceed by induction on decreasing values of $i\in \sg{0, \ldots, k}$. The cases $i=k,k-1$ are immediate since
$G$ is a graph with 2-convex balls, so by Theorem \ref{thm: INC} $\INC$ holds. Now, let $i < k-1$. By induction hypothesis,
$C_{(u,v)}^{(i+2)}$ and $C_{(u,v)}^{(i+1)}$ are disjoint nonempty cliques and their union is a clique.
Let $x\in C_{(u,v)}^{(i+2)}$.
Then $x$ is adjacent to all vertices of $C_{(u,v)}^{(i+1)}$ and $C_{(u,v)}^{(i+1)}\subset I(x,u)$.  Applying
$\INCp(x,u)$, there exists a vertex  $z \in B_{i}(u)$ adjacent to every vertex of $S_1(x) \cap B_{i+1}(u)$.
In particular, $z$ is adjacent to all vertices of $C_{(u,v)}^{(i+1)}$, therefore $z$ is a vertex of
$C_{(u,v)}^{(i)}$. This shows that $C_{(u,v)}^{(i)}$ is nonempty. Since $C_{(u,v)}^{(i)}$ belongs to $S_1(y)\cap I(y,u)$
for any $y\in C_{(u,v)}^{(i+1)}$, by $\INCz$ the set $C_{(u,v)}^{(i)}$ defines a clique.
From the definition of $C_{(u,v)}^{(i)}$, it immediately follows that $C_{(u,v)}^{(i)}\cap C_{(u,v)}^{(i+1)}=\varnothing$ and
$C_{(u,v)}^{(i)}\cup C_{(u,v)}^{(i+1)}$  is a clique as well.
\end{proof}

For any graph $G$, any vertex $u$ of $G$, and  any subset $K\subseteq S_{k}(u)$, one can define the sets $C_{(u,K)}^{(i)}$ in a similar way, by setting
$C_{(u,K)}^{(k)}:=K$ and
$C_{(u,K)}^{(i)}:=B_1^*\left(C_{(u,K)}^{(i)}\right)\cap B_{i}(u),$
for any $i\in \sg{0,\ldots, k-1}$. We also let $C_{(u,K)}^{(i)}:=K$ for any $i \geq k+1$. Again, we let $\gamma_{(u,C)}:=(\sg{u} = C_{(u,C)}^{(0)}, C_{(u,C)}^{(1)},\ldots, C_{(u,C)}^{(k-1)},C_{(u,C)}^{(k)}=C)$ denote such a clique-path.
However, Lemma \ref{lem: INCclique} does not hold anymore and  the resulting sets  may be empty. Nevertheless, the following simple observation
holds:

\begin{lemma}
\label{lem: inclalt}
 Let $G$ be any graph, $u\in V$ and $k\geq 1$ be an integer. Then for any $L \subseteq K\subseteq S_k(u)$, the inclusion $C_{(u,K)}^{(k-1)} \subseteq C_{(u,L)}^{(k-1)}$ holds.
\end{lemma}

\begin{proof}
Pick  $x\in C_{(u,K)}^{(k-1)}$. Then  $x\in S_{k-1}(u)$ and clearly $x$ is adjacent to every vertex of $K$, so to every vertex of $L$ as well, which implies that $x \in C_{(u,L)}^{(k-1)}$.
\end{proof}

  \tikzexternaldisable
  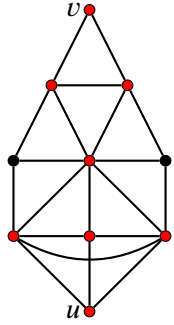
\begin{figure}[h!]
    \centering
    \begin{tikzpicture}[scale=1]
    \tikzstyle{every node}=[draw,circle,fill=red,minimum size=4pt,
                            inner sep=0pt]

    \node (u) at (0,0) [label=left: $u$] {};
    \node (11) at (-1,1) [] {};
    \node (12) at (0,1) [] {};
    \node (13) at (1,1) [] {};
    \node (21) at (-1,2) [color = black] {};
    \node (22) at (0,2) [] {};
    \node (23) at (1,2) [color= black] {};
    \node (31) at (-0.5,3) [] {};
    \node (32) at (0.5,3) [] {};
    \node (v) at (0,4) [label=left: $v$] [] {};
    \draw[thick] (u) -- (11) -- (12) -- (13) -- (u) -- (12) -- (22) -- (21) -- (11);
    \draw[thick] (11) to[bend right] (13);
    \draw[thick] (11) -- (22) -- (13) -- (23) -- (22) -- (31) -- (v) -- (32) -- (23);
    \draw[thick] (21) -- (31) -- (32) -- (22);
    \end{tikzpicture}
    \caption{In red, the vertices of the cliques $C_{(u,v)}^{(i)}$ of the clique-path $\gamma_{(u,v)}$ in a CB-graph.}
\end{figure}

\subsection{Local characterization}
In this subsection,  we provide a local characterization of the clique-paths  $\gamma_{(u,v)}, (u,v)\in V\times V$ defined  above. This will show that they
are canonically defined and are unique.

Recall that two cliques $\sigma, \tau$ of a graph $G$ are at
{uniform distance} $k$ (notation \du{\sigma}{\tau}{k}) if
$d(s,t)=k$ for any $s\in \sigma$ and any $t\in \tau$.
Let $u$ be a vertex and $C$ be a nonempty clique of $G$. A sequence of nonempty cliques $\eta_{(u,C)}=(\sg{u} = C_0, C_1, \ldots, C_k = C)$ is
called a \emph{normal clique-path} from $u$ to $C$ if the following local conditions hold:
\begin{enumerate}[($i$)]
 \item for any $i=0,\ldots,k-1$, $C_i \cap C_{i+1} = \varnothing$ and $C_i \cup C_{i+1}$ form a clique;
 \item for any $i=1,\ldots, k-1$, $C_{i-1}\cap C_{i+1}= \varnothing$ and there is no edge between $C_{i+1}$ and $C_{i-1}$;
 \item for any $i=3, \ldots, k$, $C_i$ is at uniform distance $3$ from $C_{i-3}$.
 \item for any $i=1,\ldots,k-1$,  $C_i=B_1^*(C_{i+1})\cap B_1(C_{i-1}).$
 \end{enumerate}
 Note that ($i$) is the definition of a clique-path.  Note also that when
 $k \geq 3$, condition ($iii$) implies condition ($ii$).
Our next goal is to prove that the clique-paths $\gamma_{(u,v)}$ are
the unique normal clique-paths. We prove this result by induction and
to do so, we need to consider a stronger induction hypothesis dealing
with the clique-paths $\gamma_{(u,C)}$ between vertices $u$ and
cliques $C$ at uniform distance from $u$.

\begin{lemma}
\label{lem: cliquepathsarenormal}
Let $G$ be a graph with $2$-convex balls, $u$ be a vertex, and $C$ be a nonempty  clique of $G$ at uniform distance $k$ from $u$.
Additionally assume that  all sets $C_{(u,C)}^{(i)}$ are nonempty. Then the sequence $\gamma_{(u,C)}=(\sg{u} = C_{(u,C)}^{(0)}, C_{(u,C)}^{(1)},\ldots, C_{(u,C)}^{(k-1)},C_{(u,C)}^{(k)}=C)$ is a normal clique-path. Consequently, for any pair of vertices $u,v$  of $G$, the clique-path
 $\gamma_{(u,v)}$ is a normal clique-path.
\end{lemma}

\begin{proof}
To simplify the notation, we set $C_i := C_{(u,C)}^{(i)}$ for any $i \in \sg{0, \ldots, k}$. Since $G$ is a graph with 2-convex balls, each $C_i$ is a clique, thus
$\gamma_{(u,C)}=(\sg{u} = C_0,C_1,\ldots, C_{k-1},C_{k}=C)$ is a path of cliques. From the definition of $\gamma_{(u,C)}$ it immediately follows that $\gamma_{(u,C)}$
satisfies ($i$), ($ii$) and ($iii$). It remains to establish ($iv$), i.e. that for any $i \in \sg{1, \ldots, k-1}$ we have
$C_i :=B_1^*(C_{i+1}) \cap B_i(u) =B_1^*(C_{i+1}) \cap B_1(C_{i-1})$. We proceed by induction on $k=d(u,C)$. This is direct if $k=2$. Hence assume that $k \geq 3$ and
that the required equality holds for every $i \in \sg{1, \ldots, k-2}$.  Since $C_i = C_{(u,C_{k-1})}^{(i)}$ for every $i \in \sg{0, \ldots, k-1}$, it is enough
to prove  the equality for $i = k-1$, i.e., that $B_1^*(C_{k}) \cap B_{k-1}(u) =B_1^*(C_{k}) \cap B_1(C_{k-2})$.

By ($i$), $C_{k-1} \cup C_{k-2}$ is a clique and thus
$C_{k-1} \subseteq B_1(C_{k-2})$, implying that
$C_{k-1} = B_1^*(C_{k}) \cap B_{k-1}(u) \subseteq B_1^*(C_{k}) \cap
B_1(C_{k-2})$. Conversely, $B_1(C_{k-2}) \subseteq B_{k-1}(u)$ since
all vertices of $C_{k-2}$ are at distance $k-2$ from $u$. This implies
that
$B_1^*(C_{k}) \cap B_1(C_{k-2}) \subseteq B_1^*(C_{k}) \cap
B_{k-1}(u)$ and thus ($iv$) holds.
\end{proof}

Now we prove the converse direction for CB-graphs.

\begin{lemma}
 \label{lem: normalcliquepathunicity} Let $G$ be a CB-graph, $u$ be a vertex, and $C$ be a nonempty  clique of $G$ at uniform distance $k$ from $u$.
 Let $\eta_{(u,C)}=(\sg{u} = C_0, C_1, \ldots, C_{\ell}= C)$ be a normal  clique-path from $u$ to $C$. Then  $\eta_{(u,C)}$ coincides with the clique-path $\gamma_{(u,C)}$ defined in Lemma  \ref{lem: cliquepathsarenormal}. Consequently, for any pair of vertices $u,v$  of $G$ there exists a unique normal clique-path from $u$ to $v$ and this path is
 $\gamma_{(u,v)}$.
\end{lemma}

\begin{proof}
We proceed by induction on the length $\ell$ of the normal clique-path
$\eta_{(u,C)}$.  If $\ell\le 2$ there is nothing to prove. Thus, we assume that $\ell\ge 3$ and that the assertion of the lemma holds for any normal clique-path of length smaller than $\ell$. To prove the assertion for $\eta_{(u,C)}$, we will show that for any $i \in \sg{0,\ldots,\ell-1}$, (a) $C_i$ is at uniform distance $i$ from $u$ and (b) the equality $C_i = B^*_1(C_{i+1})\cap B_{i}(u)$ holds. If for some $j<\ell$, the clique $C_j$ is at uniform distance $j$ from $u$, then the subpath $\eta_{(u,C_j)}:=(\sg{u} = C_0, C_1, \ldots, C_{j})$ of $\eta_{(u,C)}$ is a normal clique-path from $u$ to $C_j$. Therefore, by induction assumption $\eta_{(u,C_j)}$ coincides with the clique-path $\gamma_{(u,C_j)}$, consequently, $\eta_{(u,C_j)}$ satisfies the properties (a) and (b).

By condition ($i$) of normal clique-paths, the union of two consecutive cliques of $\eta_{(u,C)}$ is a clique, therefore for any $i\le \ell-1$,  all vertices of $C_{i}$ have distance at most $i$ from $u$. Now we prove that each  $C_i$ is at uniform distance $i$ from the origin $u$, thus establishing (a). Suppose this is not true
and let $i$ be the smallest index  such that $C_i$ is not at uniform distance $i$ from $u$.
By condition ($iii$) of normal clique-paths, we must have $i > 3$.
The minimality choice of $i$
implies that all cliques $C_j$ with $j<i\le \ell$ are at uniform distance $j$ from  $u$. Consequently, for  all such $C_j$, $\eta_{(u,C_j)}$ coincides with $\gamma_{(u,C_j)}$.  From the choice of $i$ it also follows that there exists a vertex $x\in C_i$ with $d(u,x)<i$.

If $d(u,x)=i-1$,  then we can apply  $\TPCp(u,xy)$, where $y\in C_{i-1}$. If $\TC$ applies, then there exists a vertex $z\sim x,y$ at distance $i-2$ from $u$. By $\INCp(u)$, $z$ must be adjacent to every vertex of $C_{i-2}$ and there exists a vertex $z'$ at distance $i-3$ from $u$ which is adjacent to all vertices of the clique $C_{i-2}\cup\sg{z}$.
Since $\eta_{(u,C_{i-1})}=\gamma_{(u,C_{i-1})}$, we have $C_{i-3}=B^*_1(C_{i-2})\cap B_{i-3}(u)$, yielding $z' \in C_{i-3}$. This implies that $x$ is at distance at most $2$ from $z' \in C_{i-3}$, contradicting condition ($iii$). Otherwise, if $\PCp$ applies, then
there exists a vertex $z$ at distance $i-3$ from $u$ and at distance 2 from $x,y$ and which is adjacent to every $w \in C_{i-2}$.
Again, since $C_{i-3}=B^*_1(C_{i-2})\cap B_{i-3}(u)$, we conclude that $z \in C_{i-3}$. This implies that $x$ is at distance $2$ from a vertex of $C_{i-3}$,
contradicting again ($iii$).

Now suppose that $d(u,x)\le i-2$ and that $C_i$ does not contain a vertex at distance $i-1$ from $u$. Since all vertices of $C_{i-1}$ are at distance $i-1$ from $u$ and $C_i\cup C_{i-1}$ is a clique, this implies that $C_i$ has uniform distance $i-2$ from $u$. Since $C_{i-2}$ is at uniform distance $i-2$ from $u$, all vertices of $C_i\cup C_{i-2}$ are adjacent to any vertex $w\in C_{i-1}$ and $C_i\cup C_{i-2}\subset I(w,u)$.  Then $\INCp(u)$ immediately implies that every vertex of $C_i$ is adjacent to every vertex of $C_{i-2}$, which contradicts condition ($ii$). This establishes that each clique $C_i$ of $\eta_{(u,C)}$ is at uniform distance $i$ from $u$, establishing (a).  

By induction hypothesis,
$\eta_{(u,C_{\ell-1})}=\gamma_{(u,C_{\ell-1})}$, and thus to establish
(b), it remains to prove that
$C_{\ell-1}=B^*_1(C_{\ell}) \cap B_{\ell-1}(u)$. We use the following
claim.
\begin{claim}
 \label{clm: lemk-1}
 Let $k\geq 2$, and $C,C',C''$ be three cliques at uniform distance respectively $k, k-1,$ and $k-2$ from $u$. Assume that $C'' = B^*_1(C')\cap B_{k-2}(u)$ and $C' = B^*_1(C)\cap B_1(C'')$. Then $C'=B^*_1(C)\cap B_{k-1}(u).$
\end{claim}
\begin{proof}
As $C'$ has uniform distance $k-1$ to $u$, the direct inclusion is trivial.
Since $C\neq \varnothing$,  $B^*_1(C)\cap B_{k-1}(u)$ is a clique by $\INCz(u)$. Hence by $\INCp(u)$, there exists
$z \in B_{k-2}(u)$ such that $z$ is adjacent to any vertex of $B^*_1(C) \cap B_{k-1}(u)$. Necessarily, $z \in C''$ and we get the reverse inclusion $B^*_1(C) \cap B_{k-1}(u) \subseteq C'$.
\end{proof}

The desired equality $C_{\ell-1}=B^*_1(C_{\ell}) \cap B_{\ell-1}(u)$ follows by applying Claim \ref{clm: lemk-1} with $C:=C_{\ell}$, $C':= C_{\ell-1}$, and $C'':=C_{\ell-2}$.
\end{proof}

\subsection{Fellow traveler property}
Let $G=(V,E)$ be a CB-graph. Denote by $\Upsilon$ the set of normal clique-paths between all pairs of vertices of $G$: $\Upsilon=\{ \gamma_{(u,v)}: (u,v)\in V\times V\}$. A path $(u=u_0,u_1, \ldots, u_k = v)$ between two vertices $u,v$ of a graph $G$ is called a \emph{normal path from $u$ to $v$} if for every $i \in \sg{0, \ldots, k}$ we have $u_i \in C_{i},$ where $\gamma_{(u,v)}=(\sg{u} = C_0, C_1, \ldots, C_k = \sg{v})$ is the normal clique-path from $u$ to $v$. From the definition it follows that every normal path from $u$ to $v$ is a shortest $(u,v)$-path. Denote by $\Upsilon^+$ the set of all normal paths  between all pairs of vertices of $G$.

Let ${\mathcal P}(G)$ be the set of all paths of $G$ and ${\mathcal CP}(G)$ be the set of all clique-paths of $G$.  A \emph{(clique-)path system} $\mathcal P$ \cite{Swiat} is any subset of ${\mathcal P}(G)$ (respectively, of ${\mathcal CP}(G)$). A (clique-)path system $\mathcal P$ is \emph{complete} if any two vertices are endpoints of some (clique-)path in $\mathcal P$. Let $[0,k]^*$ denote the set of integer points
from the segment $[0,k].$ Given a (clique-)path $\gamma$ of length $k=|\gamma|$ in $G$, we can parametrize it and denote by $\gamma:[0,k]^*\rightarrow V(G)$. It will be convenient to extend $\gamma$ over $[0,\infty]$ by setting $\gamma(i)=\gamma(k)$ for any $i>k$. A (clique-)path system $\mathcal P$ of a graph $G$ satisfies the
 \emph{2-sided fellow traveler property} if there are constants $C>0$ and $D\ge 0$ such that for any two paths $\gamma_1,\gamma_2\in \mathcal P$, the following inequality holds  for all $i$:
\[d_G(\gamma_1(i),\gamma_2(i))\le C\cdot \max\{ d_G(\gamma_1(0),\gamma_2(0)),d_G(\gamma_1(\infty),\gamma_2(\infty))\}+D.\]
A \emph{bicombing} of a graph $G$ is a complete (clique-)path system $\mathcal P$ satisfying the $2$--sided fellow traveler property.   If all (clique-)paths in
${\mathcal P}$ are shortest paths of $G$, then
${\mathcal P}$ is called a \emph{geodesic bicombing}.

Our goal in this subsection is to show that in CB-graphs, the systems of normal clique-paths and of normal paths enjoy the $2$-sided fellow traveler property and thus define a geodesic bicombing of $G$:

\begin{theorem}\label{thm: fellowtravdist1}
Let $G$ be a graph with convex balls. For arbitrary vertices $u,u',v,v'$ of $G$,  let $\gamma_{(u,v)}=(\sg{u}=C_0,C_1,\ldots,C_k=\sg{v})$ and $\gamma_{(u',v')}=(\sg{u'}=C'_0,C'_1,\ldots,C'_{k}=\sg{v'})$ be the normal clique-paths from $u$ to $v$ and from $u'$ to $v'$, respectively.
Then for any $i \geq 0$ and for any pair $(x,y) \in C_{i}\times C'_{i}$, we have $d(x,y) \leq 7 \max(d(u,u'), d(v,v')).$
Consequently, the sets of normal clique-paths $\Upsilon$ and of normal  paths $\Upsilon^+$ of $G$ enjoy the $2$-sided fellow
traveler property and thus define a geodesic bicombing of $G$.
\end{theorem}

\begin{proof}
We let $\ell:=d(u,u')$ and $m:=d(v,v')$. First we show by induction over $m\geq 0$ that if $\ell=0$, i.e $u=u'$, then for any $i \geq 0$ and $(x,y) \in C_{i}\times C_{i'}$ we have $d(x,y)\leq 3m$. Indeed, if $m=1$, this is a direct consequence of Lemma \ref{lem: cliquepath}. Assume now that $m\geq 2$ and let $(v=v_0, v_1, \ldots, v_m=v')$ be a shortest $(v,v')$-path, and set $v'':=v_{m-1}$.
By applying the induction hypothesis to $\gamma_{(u,v)}$ and $\gamma_{(u,v'')}$ and Lemma \ref{lem: cliquepath} to $\gamma_{(u,v'')}$ and $\gamma_{(u,v')}$, by triangle inequality we deduce that  $d(x,y)\leq 3m$ for every $(x,y)\in C_i \times C_{i'}$ and $i \geq 0$. Similarly we can prove by Lemma \ref{lem: invcliquepath} and induction on $\ell \geq 0$ that if $m=0$, then for any $i \geq 0$ and $(x,y) \in C_{i}\times C_{i'}$ we have $d(x,y)\leq 4\ell$.

Now we prove Theorem \ref{thm: fellowtravdist1} by induction on $(\ell,m)$ together with the product order. The cases $\ell=0$ or $ m=0$ are covered by what we did. If
$\ell=1$ and $m=1$, we apply these two basis cases to the two pairs $\{ \gamma_{(u,v)},\gamma_{(u,v')}\}$,$\{ \gamma_{(u,v')},\gamma_{(u',v')}\}$ and, by
the triangle inequality we deduce that $d(x,y)\le 7$. Hence we can assume that $\ell\geq 1$, $m\geq 1$, and $\ell+m \geq 3$. Let $i \geq 0$ and $(x,y) \in C_{i}\times C'_{i}$.  Let $u'':=u_{\ell-1}$ and $v'':= v_{m-1}$ and let $z \in C''_i$, where $C''_i$ is the $i$-th clique of the normal clique-path $\gamma_{(u'',v'')}$.  By induction hypothesis,  $d(x, z) \leq 7\max\{ d(u,u''), d(v,v'')\}=7\max\{ \ell-1, m-1\}.$ On the other hand, since $u''\sim u'$ and $v''\sim v$, by the case $(\ell, m)=(1,1)$ we get $d(z,y)\le 7$. Consequently, $d(x,y)\le d(x,z)+d(z,y)\le  7\max\{ \ell-1, m-1\}+7=7\max \{ \ell, m\}$, concluding the proof of the theorem.
\end{proof}

In our next results, $G$ is a CB-graph and we follow the notations of Theorem \ref{thm: fellowtravdist1}.
We start with a structural property of normal clique-paths $\gamma_{(u,v)}=(\sg{u}=C_0,C_1,\ldots,C_{k-1},C_k=\sg{v})$ and $\gamma_{(u,v')}=(\sg{u}=C'_0,C'_1,\ldots,C'_{k-1},C'_{k}=\sg{v'})$
in the particular case when $v\sim v'$ and $d(u,v)=d(u,v')=k$.

\begin{lemma}\label{lem: clique2nvx} Let $u,v,v'$ be three vertices of $G$ such that $v\sim v'$ and $d(u,v)=d(u,v')=k$.
Then for any $1\le i\le k$ there exists $j\in \{ i-1,i\}$ for which either $C_j\cap C'_j \neq \varnothing$ or $C_{j}\cup C'_{j}$ is a clique. \end{lemma}

\begin{proof} We start with the particular case $i=k-1$.

\begin{claim} \label{claim:clique2nvx}
$C_{k-1}\cap C'_{k-1}\neq \varnothing$ or $C_{k-2}\cup C'_{k-2}$ is a clique.
\end{claim}

\begin{proof}
Notice that $C_{k-1}$ and $C'_{k-1}$ are cliques of $G$ coinciding with the intersections $B_{k-1}(u)\cap B_1(v)$ and $B_{k-1}(v)\cap B_1(v')$. Hence, $C_{k-1}\cap C'_{k-1}\neq \varnothing$  exactly when $\TC(u,vv')$ applies. Now assume that $C_{k-1}\cap C'_{k-1}=\varnothing$. This means that $\PCd(u,vv')$ applies so $B_2(v)\cap B_{k-2}(u) = B_2(v')\cap B_{k-2}(u)$ and in particular $C_{k-2}\subseteq B_2(v')$ and $C'_{k-2}\subseteq B_2(v)$. Let $z\in C_{k-2}$ and $z'\in C'_{k-2}$. We assert that $z\sim z'$. Indeed there exists $x\sim v,z'$, and in particular we must have $x\in C_{k-1}$. Thus by $\INCz(x,u)$ we get $z\sim z'$.
\end{proof}

Now, we prove the lemma  by decreasing induction on $i$. The case  $i=k$ is immediate since $C_k=\sg{v}, C'_k=\sg{v'}$, and $v\sim v'$. The case $i=k-1$ is covered by Claim \ref{claim:clique2nvx}. Hence, let $0<i\le k-2$. First suppose that there exists $x\in C_{i+1}\cap C'_{i+1}$. Then $C_i\cup C'_i\subset I(x,u)$. By $\INCz(x,u)$, $C_i \cup C'_i$ forms  a clique, thus we are in the second situation with $j=i$. Therefore, let $C_{i+1}\cap C'_{i+1} = \varnothing$. We can assume that $C_{i}\cap C'_{i}=\varnothing$ and
$C_{i-1}\cap C'_{i-1}=\varnothing$, otherwise we are immediately done.

By induction hypothesis, $C_{i+1}\cup C'_{i+1}$ forms a clique. We distinguish two cases.

\begin{claim}
 \label{claim: cliqueTC}
 If there exists $(x,x')\in C_{i+1}\times C'_{i+1}$ such that $\TC(u,xx')$ applies, then  $C_{i}\cup C'_{i}$ is a clique.
\end{claim}

\begin{proof}
  Pick $(x,x')\in C_{i+1}\times C'_{i+1}$ such that $\TC(u,xx')$
  applies.  Then there exists $s\in B_{i}(u) \cap B_1(x)\cap
  B_1(x')$. Since we assumed that $C_i \cap C'_i = \varnothing$, we can
  assume that $s\notin C'_i$.  So there exists $b\in C'_{i+1}$ such
  that $s\nsim b$. Let $y'\in C'_i$. By $\INCz(x',u)$, we have
  $s\sim y'$.  The $4$-cycle $sy'bx$ cannot be induced so we must have
  $y'\sim x$. Thus $y'$ is adjacent to any other vertex of $C_{i}$ by
  $\INCz(u)$ and consequently $C_{i}\cup C'_{i}$ is a clique.
\end{proof}

\begin{claim}
 \label{claim: cliquePC2}
 If for any $(x,x')\in C_{i+1}\times C'_{i+1}$, $\TC(u,xx')$ does not apply, then  $C_{i-1}\cup C'_{i-1}$ is a clique.
\end{claim}
\begin{proof}
 Let $z\in C_{i-1}$ and $x\in C_{i+1}$. We will show that $z$ is adjacent to any vertex of $C'_{i-1}$. Note that $\TPCd(u,xx')$ applies for any $x'\in C'_{i+1}$. In particular, we get that $z\in B_2(x)\cap B_{i-1}(u) = B_2(x')\cap B_{i-1}(u)$, so $z$ is at uniform distance $2$ from any vertex $x'$ of $C'_{i+1}$. As $i\leq k-2$, we can consider a vertex $t'\in C'_{i+2}$. By $\INCp(t', z)$, there exists a vertex $s$ adjacent to $z$ and to every vertex of $C'_{i+1}$. In particular, we must have $d(u,s)=i$ so $s \in C'_i$. Thus for every $z'\in C_{i-1}$, we get $z\sim z'$ by $\INCz(s,u)$ and we are done.
\end{proof}

We can now conclude applying Claim \ref{claim: cliqueTC} or Claim \ref{claim: cliquePC2}.
\end{proof}

\begin{lemma}\label{lem: cliquepath}
Let $u,v,v'$ be three vertices of $G$ such that $v\sim v'$ and $k:=d(u,v)\ge k':= d(u,v')$. Let $\gamma_{(u,v)}=(\sg{u}=C_0,C_1,\ldots,C_{k-1},C_k=\sg{v})$ and $\gamma_{(u,v')}=(\sg{u}=C'_0,C'_1,\ldots,C'_{k'-1},C'_{k'}=\sg{v'})$. Then for any $i \geq 0$ and any $(x,y) \in C_{i}\times C'_{i},$  we have $d(x,y) \leq 3.$
\end{lemma}

\begin{proof}
If $d(u,v')=k$, then the desired result is a consequence of Lemma \ref{lem: clique2nvx}. So, suppose that $d(u,v')=d(u,v)-1=k-1$. Then  $v' \in C_{k-1}$, so by Lemma \ref{lem: inclalt} for every $i \in \sg{0, \ldots, k-1}$, either $C_{i}\subseteq C'_{i}$ or $C'_{i} \subseteq C_{i}$ according to the parity of $i$. Hence we are also done in this case.
\end{proof}

The next result has to be viewed as a symmetric version of the previous lemma, when we exchange the sources and sinks (i.e., $u\sim u'$ and $v=v'$).

\begin{lemma}
\label{lem: invcliquepath}
Let $u,u',v$ be three vertices of $G$ such that $u\sim u'$ and $k:=d(u,v)\ge k':=d(u',v)$. Let $\gamma_{(u,v)}=(\sg{u}=C_0,C_1,\ldots,C_{k-1},C_k=\sg{v})$ and $\gamma_{(u',v)}=(\sg{u'}=C'_0,C'_1,\ldots,C'_{k-1},C'_{k}=\sg{v})$ with the convention that $C'_i=\{ v\}$ if $i>k'$. Then for any
$0\le i\le k$ and any $(x,y) \in C_i \times C'_i$ we have:
\begin{itemize}
 \item $d(x,y)=1$ if $d(u,v)>d(u',v)$;
 \item $d(x,y) \leq 4$ if $d(u,v)=d(u',v)$.
\end{itemize}
\end{lemma}

\begin{proof}
Assume first that $k=d(u,v)>d(u',v)$. Then $d(u',v)=k-1$. Note that in this case for any $1\le i\le k$ the inclusion $C'_{i-1}\subseteq B_i(u)$ holds.

\begin{claim}
 \label{clm: altinclusions}
 For any $2\le i\le k$, the inclusion $C_{i}\subseteq C'_{i-1}$
 implies the inclusion $C'_{i-2}\subseteq C_{i-1}.$
\end{claim}

\begin{proof}
  Note that
  $C'_{i-2}\subseteq B_{i-2}(u')\cap I(u',v) \subseteq B_{i-1}(u)\cap
  I(u,v)$. If $C_{i}\subseteq C'_{i-1}$ holds, then any vertex of
  $C'_{i-2}$ is adjacent to any vertex of $C'_{i-1}$, and therefore to
  any vertex of $C_{i-1}$, so we have $C'_{i-2}\subseteq C_{i-1}$.
\end{proof}

A particularly nice case occurs when the inclusions between the
$C_i$'s and the $C'_{i-1}$'s alternate at each level, since in this
case we are immediately done. However such an inclusion alternation
does only hold in the sense of Claim \ref{clm: altinclusions},
i.e. the fact that $C'_{i-1}\subseteq C_i$ does not imply that
$C_{i-1}\subseteq C'_{i-2}$. Hence we consider the maximal index
$2\le i_0\le k-1$ (if it exists) such that
$C'_{i_0-1}\subseteq C_{i_0}$ and $C_{i_0-1}\not\subseteq
C'_{i_0-2}$. If $i_0$ does not exist, then we are done because by
Claim \ref{clm: altinclusions} the inclusion alternation
$C'_{k-2} \subseteq C_{k-1}$, $C_{k-2}\subseteq C'_{k-3}$,
$C'_{k-4}\subseteq C_{k-3}$ hold and so on, until we reach the vertex
$u'$. Otherwise, there is such an inclusion alternation until we reach
$C_{i_0}$. This means that for any $i \geq i_0-1$ and any
$(x,y) \in C_{i}\times C'_i$ we have $d(x,y) \leq 1$. Now we show that
this inequality still holds for all $(x,y) \in C_{i}\times C'_i$ with
$i<i_0$. 

\begin{claim}
\label{clm: triv}
$C'_{i_0 - 2}\cup C_{i_0 - 1}$ is a clique.
\end{claim}
\begin{proof}
Pick any  $a\in  C'_{i_0 -1}\subseteq C_{i_0}$. By $\INCz(a,u)$, any two vertices of $C'_{i_0 - 2}\cup C_{i_0 - 1}$ are adjacent. \end{proof}

Since $C_{i_0}\neq \varnothing$ and $C_{i_0-1}\not\subseteq C'_{i_0-2}$, there exists $p \in  C_{i_0-1}\setminus C'_{i_0-2}$. Then $C'_{i_0-1}\subseteq C_{i_0}$ implies that
$p$ is adjacent to all vertices of $C'_{i_0-1}$, hence the only reason why $p$ does not belong to $C'_{i_0 - 2}$ is that $d(u',p)=i_0-1$.

\begin{claim}\label{clm: disjointcliques}
For any $1 \leq i < i_0-2$, there is no edge between $C'_{i-1}$ and $C_{i+1}$. In particular, $C'_{i-1} \cap C_{i} = \varnothing$.
\end{claim}

\begin{proof}
  Suppose not, and let $ab$ be an edge with $a \in C'_{i-1}$ and
  $b \in C_{i+1}$. Since $i+1 < i_0-1$, we have $b \in I(u,p)$.  Then
  we can find a path of length $i_0-2$ from $u'$ to $p$ going via
  $ab$, with exactly one vertex in each clique $C'_j$ for
  $0\leq j\leq i-1$ and exactly one vertex in each clique $C_j$ for
  $ i+1 \leq j \leq i_0-1$, leading to a contradiction.
\end{proof}

\begin{claim}
\label{clm: zigzag}
For any $0\leq i\leq i_0-2$, the unions $C'_i\cup C_{i+1}$ and
$C'_i\cup C_i$ are cliques.
\end{claim}

\begin{proof}
  For every $0 \leq i \leq i_0-2$, pick $u_i \in C_i$ and
  $u'_i \in C_i'$ and consider the paths
  $u = u_0, u_1, \ldots, u_{i_0-1} = p$ and
  $u' = u'_0, u'_1, \ldots, u'_{i_0-2}$. By the definition of $p$,
  $u'_{i_0-3} \neq u_{i_0-2}$, and by Claim~\ref{clm:
    disjointcliques}, $u'_{i-1} \neq u_i$ for $1 \leq i \leq
  i_0-3$. Consequently, these two paths are vertex-disjoint. By
  Claim~\ref{clm: triv}, $p \sim u'_{i_0-2}$.  By $\INCz(v)$,
  $u' \sim u_1$. Since $d(u', p)=i_0-1$, we have
  $d(u_i,u') = d(u'_i,u') = i$ for every $1 \leq i \leq
  i_0-2$. Observe that if $u_{i+1} \sim u_i'$ for
  $1 \leq i \leq i_0-2$, then by $\INCz(u')$, we have that
  $u_i \sim u_i'$ and by $\INCz(u)$, we then have that
  $u_i \sim u'_{i-1}$.  Since $p = u_{i_0-1}$ is adjacent to
  $u'_{i_0-2}$, we deduce that $C'_{i} \cup C_i $ is a clique for
  $1 \leq i \leq i_0-2$ and that $C'_{i-1} \cup C_i $ is a clique for
  $1 \leq i \leq i_0-2$. The claim then follows since
  $C'_{i_0-1} \cup C_{i_0-2} $ is a clique by Claim~\ref{clm: triv}
  and since $C'_0 \cup C_0 =\{u',u\}$ is a clique.
\end{proof}

By Claim \ref{clm: zigzag}, for any $i < i_0$ and any $(x,y) \in C'_{i}\times C_i$ we have $d(x,y) \leq 1$. This concludes the proof in the case $d(u,v) > d(u',v)$. Now, assume that $k = d(u,v) = d(u',v)$ and apply $\TPCz(v,uu')$. If the triangle condition holds, then there exists $z \in B_{k-1}(v)$ with $z \sim u,u'$. Applying the
previous case to $v$ and the pairs $uz$ and $u'z$, by the triangle inequality we conclude that for any $(x,y)\in C_i\times C'_i$ we have $d(x,y) \leq 2.$
If the pentagon condition holds, then there exists $w_1, w_2, z$ with $d(w_1, v)=d(w_2, v) = k-1$, $d(z,v) = k-2$, $w_1 \sim u$, $w_2\sim u'$ and $z \sim w_1, w_2$. By the previous case applied to $v$ and the pairs $u w_1$, $w_1z$, $zw_2$, and $w_2u'$, the triangle inequality implies
that for any $(x,y)\in C_i\times C'_i$ we have $d(x,y) \leq 4$.
\end{proof}

\subsection{Biautomaticity}
\label{sec: biauto}
In this subsection, we apply the previous results to show that  CB-groups are biautomatic.
We continue by recalling the definition of biautomatic group \cite{ECHLPT,BrHa}.  Let $G=(V,E)$ be a graph and
suppose that $\Gamma$ is a group acting geometrically by automorphisms
on $G$.  These assumptions imply that the graph $G$ is locally finite
and that the degrees of the vertices of $G$ are uniformly bounded.
 The action of $\Gamma$ on $G$ induces the action
of $\Gamma$ on the set ${\mathcal P}(G)$ of all paths of $G$.  A path
system ${\mathcal P}\subseteq {\mathcal P}(G)$ is called
$\Gamma$--\emph{invariant} if $g\cdot \gamma \in \mathcal P$, for all
$g\in \Gamma$ and $\gamma \in \mathcal P$.

Let $\Gamma$ be a group generated by a finite set $S$. A \emph{language} over $S$ is some set of words in $S\cup S^{-1}$ (in the
free monoid $(S\cup S^{-1})^{\ast}$).
A language over $S$ defines a $\Gamma$--invariant path system in the Cayley graph Cay$(\Gamma,S)$.
A language is \emph{regular} if it is accepted by some finite state automaton.
A \emph{biautomatic structure} is a pair $(S,\mathcal L)$, where $S$ is as above, $\mathcal L$ is a regular language over $S$, and
the associated path system in Cay$(\Gamma,S)$ is a bicombing. A group is
\emph{biautomatic} if it admits a biautomatic structure.
We use specific conditions implying biautomaticity for groups acting geometrically on graphs. The method, relying on the notion of locally recognized path system, was developed by {\'S}wi{\c{a}}tkowski \cite{Swiat}.

Let $G$ be a graph and let $\Gamma$ be a group acting geometrically on $G$.
Two paths $\gamma_1$ and $\gamma_2$ of $G$ are $\Gamma$-\emph{congruent} if there is $g\in \Gamma$ such that $g\cdot\gamma_1=\gamma_2$. Denote by ${\mathcal S}_k$ the set of $\Gamma$-congruence classes
of paths of length $k$ of $G$.  Since $\Gamma$ acts cocompactly on $G$, the sets ${\mathcal S}_k$ are finite for any natural $k$. For any path $\gamma$ of $G$, denote by $[\gamma]$ its $\Gamma$-congruent class. For a subset $R\subset {\mathcal S}_k$, let ${\mathcal P}_R$ be the path system in $G$ consisting of all paths $\gamma$ satisfying the following two conditions:
\begin{itemize}
\item[(1)] if $|\gamma|\ge k$, then $[\eta]\in R$ for any subpath $\eta$ of length $k$ of $\gamma$;
\item[(2)] if $|\gamma|<k$, then $\gamma$ is a prefix of some path $\eta$ such that $[\eta]\in R$.
\end{itemize}

A path system $\mathcal P$ in $G$ is $k$--\emph{locally recognized} if for some $R\subset {\mathcal S}_k$, we have ${\mathcal P}={\mathcal P}_R$, and $\mathcal P$ is \emph{locally recognized}
if it is $k$--locally recognized for some $k$. The following result of  {\'S}wi{\c{a}}tkowski  \cite{Swiat} provides sufficient conditions of biautomaticity in terms of local recognition and bicombing.

\begin{theorem}[\cite{Swiat}*{Corollary 7.2}]\label{swiat}\label{th:swiat} Let $\Gamma$ be a group acting geometrically on a graph $G$ and let $\mathcal P$ be a path system in
$G$ satisfying the following conditions:
\begin{itemize}
\item[(1)] $\mathcal P$ is locally recognized;
\item[(2)] there exists $v_0\in V(G)$ such that any two vertices from the orbit $\Gamma \cdot v_0$ are connected by a path from $\mathcal P$;
\item[(3)] $\mathcal P$ satisfies the $2$--sided fellow traveler property.
\end{itemize}
Then $\Gamma$ is biautomatic.
\end{theorem}

Based on Theorems \ref{thm: fellowtravdist1} and \ref{th:swiat}, we obtain the following result:

\begin{theorem}\label{thm: biauto}
Let  $G$ be a CB-graph and $\Gamma$ be a group acting geometrically on $G$. Then $\Gamma$ is biautomatic.
\end{theorem}

\begin{proof}
As mentioned above, since $\Gamma$ acts geometrically on $G$, the graph $G$ has bounded degrees, so all cliques of $G$ are finite and we can use all previous results.
Recall that every group action of $\Gamma$ on $G$ induces a group action of $\Gamma$ on the barycentric subdivision $\beta(X(G))$
of the clique complex $X(G)$, and hence on the graph $\beta(G):= \beta(X(G))^{(1)}$. Moreover, if $\Gamma$ acts geometrically
on $G$, then $\Gamma$ also acts geometrically on $\beta(G)$. Now observe that every normal clique-path $\gamma_{(u,v)}(\sg{u}=C_0,C_1,\ldots,C_k=\sg{v})$
of length $k$ between two vertices $u,v$ of $G$ correspond to the path of $\gamma^*_{(u,v)}=(C_0, C_0\cup C_1, C_1, C_1\cup C_2, C_2, \ldots, C_{k-1}, C_{k-1}\cup C_k, C_k)$
of length $2k+1$ of the graph $\beta(G)$. Denote by $\Upsilon^*$ the set of all paths $\gamma^*_{(u,v)}$ of $\beta(G)$ derived from the set
$\Upsilon$ of normal clique-paths $\gamma^*_{(u,v)}$ of $G$. We assert that if $\Gamma$ acts geometrically on $G$, then $\Upsilon^*$
satisfies the conditions of Theorem \ref{th:swiat}.

First observe that if $u$ and $v$ are at distance $k>0$ in $G$, then for any two cliques $C,C'$ with $u\in C$ and $v\in C'$, the vertices of $\beta(G)$ corresponding to  $C$ and $C'$ are at distance at most $4k$ in $\beta(G)$. Thus by Theorem \ref{thm: fellowtravdist1} the system of paths $\Upsilon^*$ of $\beta(G)$ also satisfies the $2$-sided fellow traveler property, with associated constants $C$ and $D$ four times higher that those for $\Upsilon$. A simple argument similar to the one used in \cite{JS,Swiat} implies that the system of paths $\Upsilon^*$ is $3$-locally recognizable, thanks to the local characterization of normal clique-paths provided by Lemma \ref{lem: normalcliquepathunicity}. Finally, observe that the action of $\Gamma$ on $\beta(G)$ preserve the size of the cliques, so vertices are sent to vertices, thus condition $(2)$ of Theorem  \ref{th:swiat} also holds, by Lemma \ref{lem: normalcliquepathunicity}.
\end{proof}

\subsection{Falsification by fellow traveler property}
The \emph{falsification by fellow traveler property (FFTP)} was initially introduced for groups by Neumann and Shapiro \cite{NeSh}  and further investigated by Elder \cite{Elder_thesis,Elder02}. It was shown in \cite{NeSh} that in the groups satisfying FFTP the language of geodesics is regular. It was shown in \cite{Elder02} that
the groups satisfying  FFTP are almost convex (in the sense of Cannon \cite{Cannon}) and satisfy a quadratic isoperimetric inequality.
The definition of FFTP can be formulated in the language of graphs as follows.

We say that two paths $\gamma, \gamma'$ in $G$ \emph{asynchronously $K$-fellow travel} for some constant $K>0$ if there exists a proper nondecreasing function $f:\mathbb [0, \infty) \to \mathbb [0, \infty)$ such that for every $i\geq 0$,
\[d(\gamma(i), \gamma'(f(i)))\leq Kd(\gamma(0), \gamma'(0)).\]

A graph $G$ is said to have the \emph{falsification by fellow traveler property} (or FFTP for short) if there exists a constant $K>0$ such that for every path $\gamma$ of $G$ which is not a geodesic, there exists a path $\eta$ of $G$ such that $\eta(0)=\gamma(0)$, $\eta(\infty)=\gamma(\infty)$, $|\eta|<|\gamma|$ and $\gamma$ and $\eta$ asynchronously $K$-fellow travel.

\begin{proposition}
 \label{lem: FFTP}
 Let $G$ be a CB-graph or a weakly modular graph. Then $G$ enjoys FFTP. Consequently, any group $\Gamma$ whose Cayley graph is weakly modular or a CB-graph enjoys FFTP.
\end{proposition}

\begin{proof} We first prove the result for CB-graphs and start with a general claim about them.
 \begin{claim}
 \label{clm: path-homotopy}
  Let $G$ be a CB-graph (respectively a weakly modular graph) and $u,v,w,w'$ be four vertices such that $k:=d(u,v)=d(u,w)+1=d(u,w')+1$ and $v\sim w,w'$. Then for every geodesic $\gamma$ from $u$ to $w$ there exists a geodesic $\gamma'$ from $u$ to $w'$ which $1$-fellow travels (respectively $2$-fellow travels)
  with $\gamma$.
 \end{claim}
 \begin{proof}
 Assume that $G$ is a CB-graph.
 By $\INCz(v,u)$ we have $w\sim w'$. Now we proceed by induction on $k\geq 2$. If $k =1$, then we are immediately done. If $k\geq 2$, then
 by $\INC(v,u)$ there exists $z\sim w, w'$ such that $d(u,z)=k-2$. Now we can apply the induction hypothesis to the vertices $u,w,\gamma(k-2),z$ and find a geodesic from $u$ to $z$ which $1$-fellow travels with $(\gamma(0), \ldots, \gamma(k-2))$. Completing this path with $w'$ we are done. If $G$ is weakly modular, then we are  done  by using $\mathrm{QC}$ instead of $\INC$.
 \end{proof}

 We now show that FFTP holds with constant $K=2$ for CB-graphs and $K=3$ for weakly modular graphs. Let $G$ be a CB-graph or a weakly modular graph.
 Let $\gamma$ be a non-geodesic path in $G$ and let $u:=\gamma(0)$ and $k:=|\gamma|$. Then $k\geq 2$.
 As $\gamma$ is not a geodesic, we can define $i_0:=\min\sg{i\geq 1, d(u,\gamma(i))\geq d(u,\gamma(i+1))}= \min\sg{i\geq 1, d(u, \gamma(i+1))\leq i}$.
 If $i_0=1$, then either $\gamma(2)=u$ or $\gamma(2)\sim u$, so we are done by removing $\gamma(0)$ and $\gamma(1)$ or only $\gamma(1)$ from $\gamma$.

 Now, let $i_0\geq 2$.  Let $w:=\gamma(i_0-1)$, $v:=\gamma(i_0)$, $v':=\gamma(i_0+1)$. If $d(u,v')=i_0-1$, then by Claim \ref{clm: path-homotopy} there exists a $(u,v')$-geodesic $\gamma'$ which either $1$-fellow or $2$-fellow travels with $(\gamma(0), \ldots, \gamma(i_0-1)=w)$. So we can choose $\eta := (\gamma'(0), \ldots, \gamma'(i_0-2), v',\gamma(i_0+2), \ldots, \gamma(k))$ and we are done.

 Otherwise, $d(u,v')=i_0$ and we apply $\TPCu(u, vv')$ or $\TC(u, vv')$ according to whether $G$ is a CB-graph or weakly modular. If $\TC(u,vv')$ applies, then there exists a vertex $w'\sim v,v'$ at distance $i_0-1$ from $u$. By Claim \ref{clm: path-homotopy}, we can find a geodesic $\gamma'$ from $u$ to $w'$ which $1$-fellow or $2$-fellow travels with $(\gamma(0), \ldots, \gamma(i_0-1)=w)$. We thus let $\eta:=(u=\gamma'(0), \ldots, \gamma'(i_0-2),w', v', \gamma(i_0+2), \ldots, \gamma(k))$.
 If $\PCu(u,vv')$ holds with respect to the neighbor $w$ of $v$, then there exists two vertices $w',z$ with $d(u,z)=i_0-2$, $d(u,w')=i_0-1$, $w'\sim v'$ and $z\sim w', w$. By Claim \ref{clm: path-homotopy} there exists a geodesic $\gamma'$ from $u$ to $z$ which $1$-fellow travels with $(u=\gamma(0), \ldots, \gamma(i_0-2))$. Thus we let
 $\eta:=(u=\gamma'(0), \ldots, \gamma'(i_0-3), z, w', v', \gamma(i_0+2), \ldots, \gamma(k))$.
\end{proof}

In \cite{Cannon}, Cannon introduced the notion of \emph{almost convexity}. As for FFTP, the initial definition was for groups, and we extend it there in the language of graphs.
A graph $G$ is said to be \emph{$k$-almost convex} ($\mathrm{AC}(k)$ for short) for some $k\geq 2$ if there exists a constant $K_k>0$ such that for every $v\in V, n \geq 0$ and every pair of vertices $x, y\in S_n(v)$ with $d(x,y)\leq k$ there exists a path from $x$ to $y$ of length at most $K_k$ which is entirely included in $B_n(v)$. A graph $G$ is said to be \emph{almost convex} if there exists some $k\geq 3$ such that $G$ is $k$-almost convex.

Observe that the class almost convex graphs trivially generalizes the
class of CB-graphs. It is easy to show (see for example \cite{Cannon})
that for every $k\geq 3$, the property $\mathrm{AC}(k)$ is equivalent
to $\mathrm{AC}(3)$. The fact that CB-graphs are graphs with
$3$-convex balls (item (ii) of Theorem \ref{thm: well-bridged}) holds
for exactly the same kind of reason. It was shown in
\cite{Elder02}*{Proposition 1} that graphs satisfying FFTP are almost
convex.

\begin{remark} In general, the Cayley graph of a weakly modular or a CB-group is not weakly modular or a CB-graph. Recently, the paper \cite{Soe} provides
sufficient conditions on a presentation of a group, which imply that its Cayley graph is systolic.
\end{remark}

In view of this remark, one can ask  the following questions:

\begin{question} Is it true that CB-groups and weakly modular groups satisfy FFTP? Are CB-groups and weakly modular groups almost convex?
\end{question}

\begin{remark} Example \ref{CB-pentagons-infinite} shows that there exist CB-graphs which are not weakly systolic and which have  an arbitrary diameter and  contain an arbitrary number of pentagons not included in 5-wheels.
Nevertheless, such graphs are quasi-isometric to weakly systolic graphs. In view of this, one can ask the following question: \end{remark}

\begin{question} Is it true that CB-groups are weakly systolic?
\end{question}  \section{Metric triangles}
\label{sec: metriangle}

Recall that three vertices $u,v,w$ of a graph $G$ form a \emph{metric
  triangle} $uvw$ if $I(u,v)\cap I(u,w) = \sg{u}$,
$I(u,v)\cap I(v,w) = \sg{v}$, and $I(u,w)\cap I(u,w) = \sg{w}$
\cite{Ch_metric}. The pairs $uv$, $vw$, and $wu$ are called the
\emph{sides of $uvw$}. The metric triangle $uvw$ has \emph{type}
$(k_1,k_2,k_3)$ if its sides have lengths $k_1, k_2,k_3$ and
$k_1\ge k_2\ge k_3$. If $k_1=k_2=k_3=k$, then $uvw$ is called
\emph{equilateral of size $k$}. If $k_3\ge k_1-1$, then we call $uvw$
\emph{almost equilateral}. A metric triangle $uvw$ is called
\emph{strongly equilateral} if all vertices of $I(v,w)$ (respectively,
$I(u,w)$, $I(u,v)$) have the same distance to the opposite vertex $u$
(respectively, $v$, $w$). Clearly, strongly equilateral metric
triangles are equilateral. A metric triangle $uvw$ is a
\emph{quasi-median} of vertices $x,y,z$ if the following equalities
hold:
\begin{align*}
  d(x,y)&=d(x,u)+d(u,v)+d(v,y),\\
  d(y,z)&=d(y,v)+d(v,w)+d(w,z),\\
  d(z,x)&=d(z,w)+d(w,u)+d(u,x).
\end{align*}
If $uvw$ has size 0, then it is called the \emph{median} of $x,y,z$. Any triplet of vertices $x,y,z$ admits at least one quasi-median: it suffices to chose $u$ in $I(x,y)\cap I(x,z)$ as far as possible from $x$, $v$ in $I(y,u)\cap I(y,z)$ as far as possible from $y$, and $w$ in $I(z,u)\cap I(z,v)$ as far as possible from $z$; if $x,y,z$ form a metric triangle, then $xyz$ is the unique quasi-median of $x,y,z$.

Metric triangles and quasi-medians an important role and have interesting properties in classes of graphs defined by metric conditions \cite{BaCh_survey}. For example, it was shown in \cite{Ch_metric} that weakly modular graphs are exactly the graphs in which all metric triangles are strongly equilateral. Modular graphs are exactly the graphs in which all metric triangles have size $0$, i.e., all triplets of vertices have medians.  Median graphs are exactly the graphs in which all triplets of vertices have unique medians.   

In this section, we show that in CB-graphs metric triangles behave quite similarly to metric triangles in weakly modular graphs.  Namely, our main result there is the following classification theorem:
\begin{theorem}
 \label{thm: metrictriangle}
 Let $uvw$ be a metric triangle of a CB-graph $G$ with
 $d(u,v) \geq d(u,w) \geq d(v,w)$.  Then $uvw$ is either strongly
 equilateral or has type $(2,2,1)$ and is included in a pentagon of
 the form
 $uvxwy$. \end{theorem}

\begin{proof}
  The proof is based on the following lemmas, that we will
  prove later. 

  \begin{lemma}\label{lem: almosteq}
    Let $uvw$ be a metric triangle of a CB-graph $G$. If
    $d(u,v)\geq d(u,w)$, then for any $x\in I(v,w)$ we have
    $d(u,v)-1 \leq d(u,w) \leq d(u,x) \leq d(u,v)$.  Consequently,
    $uvw$ is almost equilateral and if $uvw$ is equilateral, then it
    is strongly equilateral.
\end{lemma}

  \begin{lemma}\label{lem: no2k-1}
    A CB-graph $G$ does not contain metric triangles of the type
    $(k,k-1,k-1)$ for any $k\ge 2$.
  \end{lemma}

  \begin{lemma}\label{lem: no2k}
    A CB-graph $G$ does not contain metric triangles of type
    $(k,k,k-1)$ for any $k\ge 3$.
\end{lemma}
  
  By Lemma \ref{lem: almosteq}, $uvw$
   is almost equilateral. If $uvw$ is not equilateral, then Lemma
   \ref{lem: no2k-1} implies that $uvw$ has type $(k,k,k-1)$ for some
   $k \geq 2$. By Lemma \ref{lem: no2k}, we must have $k=2$, so $uvw$
   lies on a pentagon as asserted. If $uvw$ is equilateral, then $uvw$
   is strongly equilateral as a direct consequence of Lemma \ref{lem:
   almosteq}. This concludes the proof of  Theorem \ref{thm: metrictriangle}.  
\end{proof} 

  We now prove Lemmas \ref{lem: almosteq}, \ref{lem: no2k-1}, and
  \ref{lem: no2k}.  We start with an auxiliary lemma:

  \begin{lemma}\label{closestk-1}
    Let $uvw$ be a metric triangle in a CB-graph $G$ such that
    $k=d(u,v)\ge d(u,w)$ and $v'$ be a neighbor of $v$ in
    $I(v,u)$. Suppose that $I(v,w)$ contains a vertex $t$ such that
    $d(u,t)\le k-1$. Then $I(v,t)\subset I(v,w)$ contains a vertex $x$
    such that $d(u,x)=k-1$ and $d(v,x)=d(v',x)=2$. Moreover, there
    exists a vertex $v''\in I(v,u)$ at distance 2 from $v$ and
    adjacent to $x$.
  \end{lemma}

  \begin{proof}
    Let $x$ be a closest to $v$ vertex of $I(v,t)$ such that
    $d(u,x)\le k-1$. First we prove that $d(v',x)=d(v,x)$. Indeed,
    since $v'\in I(v,u)$ and $uvw$ is a metric triangle, necessarily
    $v'\notin I(v,w)$, yielding $v'\notin I(v,x)$ since
    $I(v,x) \subset I(v,w)$. Thus $d(v',x)\ge d(v,x)$. If
    $d(v',x)>d(v,x)$, then $v\in I(v',x)$ and since
    $v',x\in B_{k-1}(u)$ and $v\notin B_{k-1}(u)$, we obtain a
    contradiction with the convexity of $B_{k-1}(u)$. This shows that
    $d(v',x)=d(v,x)$.

    Now we show that $d(v,x)=2$.  Since $d(v',x)=d(v,x)$, we can apply
    $\TPC^1(x,vv')$. If $\TC(x,vv')$ applies, then there exists some
    $s \sim v,v'$ one step closer to $x$. In particular,
    $s \in I(v',x)$, so by convexity of $B_{k-1}(u)$ we have
    $d(u,s)\le k-1$. This implies that $s \in I(v,u)\cap I(v,w)$,
    which contradicts that $uvw$ is a metric triangle. Hence, only
    $\PCu(x,vv')$ can be applied. By this property there exists a
    vertex $s\in I(v,x)\cap I(v',x)$ at distance 2 from $v$ and
    $v'$. Since $v',x\in B_{k-1}(u)$, the convexity of $B_{k-1}(u)$
    implies that $d(u,s)\le k-1$. The minimality choice of $x$ implies
    that $x=s$, i.e., $d(v,x)=2$.

    Let $z$ be a common neighbor of $v$ and $x$. Then $d(u,z)=k$ and
    we can apply $\TPC^1(u,vz)$. If $\TC(u,vz)$ applies, then there
    exists a vertex $t\sim v,z$ at distance $k-1$ from $u$. Since
    $x,t\in I(z,u)$ are neighbors of $z$, by $\INCz(u)$ we conclude
    that $x\sim t$. This implies that $t \in I(v,w) \cap I(v,u)$ which
    is impossible since $uvw$ is a metric triangle.  Hence
    $\PCu(u,vz)$ holds for the neighbor $x \in I(z,u)$ of $z$.  By
    this condition, there exists a vertex $v''\in I(v,u)\cap I(z,u)$
    having distance 2 to $v$ and $z$ and adjacent to $x$.
 \end{proof}

 We now prove Lemma~\ref{lem: almosteq}

 \begin{proof}[Proof of Lemma~\ref{lem: almosteq}]
   The first assertion immediately implies the remaining
   assertions. Indeed, the fact that any two sides of $uvw$ have
   length which differ by at most 1 implies that all metric triangles
   of $G$ are almost equilateral. If $uvw$ is equilateral of size $k$,
   then the fact that $d(u,w) \leq d(u,x) \leq d(u,v)$ implies that
   $d(u,x)=k$ for any $x\in I(v,w)$.

   So, it remains to show that for any $x \in I(v,w)$,
   $d(u,v)-1 \leq d(u,w) \leq d(u,x) \leq d(u,v) $. Let
   $k:= d(u,v)\geq d(u,w)$.  By convexity of the ball $B_k(u)$,
   $d(u,x)\le k$ for any $x\in I(v,w)$.  If the first assertion does
   not hold, then there exists $x\in I(v,w)$ such that $d(u,x)\le k-2$
   if $d(u,v)>d(u,w)$ and such that $d(u,x)\le k-1$ if
   $d(u,v)=d(u,w)$.
We distinguish two cases.

 \medskip\noindent
 {\bf Case 1:} $k=d(u,v)>d(u,w)$.

 \begin{proof}
   Let $x$ be a closest to $v$ vertex of $I(v,w)$ such that
   $d(u,x)\le k-2$. From the choice of $x$ we conclude that
   $d(u,x)=k-2$ and $k-1\le d(u,x')\le k$ for any vertex
   $x'\in I(v,x)\setminus \{ x\}$. Let $z$ be a closest to $v$ vertex
   of $I(v,x)$ such that $d(u,z)\le k-1$. From Lemma \ref{closestk-1}
   we conclude that $d(v,z)=2$ and that there exists $v''\in I(v,u)$
   having distance 2 to $v$ and adjacent to $z$. Since
   $d(u,v'')=k-2=d(u,x)$ and $d(u,z)=k-1$, from the convexity of the
   ball $B_{k-2}(u)$ we infer that $z\notin I(v'',x)$. This implies
   that $d(v'',x)\le d(z,x)$. Since $d(v,v'')=d(v,z)=2$, we conclude
   that $v''\in I(v,x)\cap I(v,u)\subset I(v,w)\cap I(v,u)$, contrary
   to the assumption that $I(v,w)\cap I(v,u)=\{ v\}$. This
   contradiction finishes the proof of Case 1.
\end{proof}

\medskip\noindent
{\bf Case 2:} $k=d(u,v)=d(u,w)$.

\begin{proof}
  Let $x$ be a vertex of $I(v,w)$ with $d(u,x)\le k-1$.  Let $x'$ be a
  closest to $v$ vertex of $I(v,x)$ at distance $\le k-1$ from $u$ and
  let $y'$ be a closest to $w$ vertex of $I(x,w)$ at distance
  $\le k-1$ from $u$. By Lemma \ref{closestk-1}, we get that
  $d(u,x')=d(u,y')=k-1$, that $d(v,x')=d(w,y')=2$ and that for any
  neighbor $v'$ of $v$ in $I(v,u)$ and any neighbor $w'$ of $w$ in
  $I(w,u)$, we have $d(v',x')=d(w',y')=2$. Moreover, there exists a
  vertex $v''\in I(v,u)$ at distance 2 from $v$ and adjacent to $x'$
  and a vertex $w''\in I(w,u)$ at distance 2 from $w$ and adjacent to
  $y'$.  Let also $z$ be a common neighbor of $v$ and $x'$ and $v'$ be
  a common neighbor of $v$ and $v''$. If $z\sim v'$, to avoid an
  induced 4-cycle defined by the vertices $z,v',v'',$ and $x'$, then
  $z\sim v''$. This implies that $z\in I(v,u)\cap I(v,w)$, contrary to
  the assumption that $I(v,u)\cap I(v,w)=\{ v\}$. Consequently,
  $v,v',v'',z,x$ induce a pentagon.

Since $v'',w''\in B_{k-2}(u)$ and $x',y'\notin B_{k-2}(u)$, the path consisting of the edges $v''x', y'w''$ and a shortest $(x',y')$-path cannot be a shortest $(v'',w'')$-path. If we suppose that $d(v,w)=m$, since $d(v,x')=d(y',w)=2$, this implies that $d(v'',w'')<m-4+2=m-2$. Hence $d(v'',w'')\le m-3$. On the other hand, if $d(v'',w'')\le m-4$, then $v''$ and $w''$ belong to a shortest $(v,w)$-path, contradicting the assumption that $I(v,u)\cap I(v,w)=\{ v\}$. This shows that $d(v'',w'')=m-3$. This implies that $d(w'',v')\le m-2$. Since $w''\sim y'$ and $d(y',z)=d(y',x')+1=m-3$, we also conclude that $d(w'',z)\le m-2$. Since $z\nsim v'$, we obtain that $v\in I(v',z)$. Since $v',z\in B_{m-2}(w'')$ and $B_{m-2}(w'')$ is convex, we conclude that $d(w'',v)\le m-2$. This implies that $w''\in I(w,v)$, contrary to the assumption that $I(w,v)\cap I(w,u)=\{ w\},$ concluding the analysis of Case 2.
\end{proof}
This finishes the proof of the lemma. \end{proof}

Using Lemma \ref{lem: almosteq} we can now easily prove
Lemma~\ref{lem: no2k-1}.

\begin{proof}[Proof of Lemma~\ref{lem: no2k-1}]
 Assume by way of contradiction that $G$ contains a metric triangle $uvw$ with $d(u,w)=d(v,w)=k-1$ and $d(u,v)=k$. Then obviously $k\geq 3$. By Lemma \ref{lem: almosteq}, for any vertex $y\in I(u,v)$ we have $k-1=d(w,v)\le d(w,y)\le d(w,u)=k-1$, i.e., $d(w,y)=k-1$. Let $x$ be a closest to $v$ vertex of $I(v,w)$ such that $d(u,x)\le k-1$ (such a vertex exists because $d(u,w)=k-1$). By Lemma \ref{closestk-1}, $d(v,x)=2$ and there exists a vertex $v''\in I(v,u)$ at distance 2 from $v$ and adjacent to $x$. Since $d(x,w)=k-3$, we conclude that $d(w,v'')\le k-2$, contrary to the fact that $d(w,y)=k-1$ for any $y\in I(u,v)$.
\end{proof}

It remains to establish Lemma~\ref{lem: no2k}.

\begin{proof}[Proof of Lemma~\ref{lem: no2k}] We proceed by induction
  on $k$.  Suppose by way of contradiction that $G$ contains a metric
  triangle $uvw$ with $d(u,v)=d(u,w)=k \geq 3$ and $d(v,w)=k-1$. Let
  $x$ be a closest to $u$ vertex of $I(v,u)$ at distance $k-1$ from
  $w$ (such a vertex exists because $d(w,v)=k-1$).  By Lemma
  \ref{closestk-1}, $d(x,u)=2$ and there exists a vertex $y\in I(u,w)$
  at distance 2 from $u$ and adjacent to $x$. Consider a quasi-median
  $x'v'w'$ of the triplet $x,v,w$. Since
  $I(v,x)\cap I(v,w)\subset I(v,u)\cap I(v,w)=\{ v\}$, we obtain that
  $v'=v$.  If $x'\ne x$, then $x$ has a neighbor
  $s \in I(x,v) \subseteq I(u,v)$ such that $d(w,s) = k-2$,
  contradicting Lemma~\ref{lem: almosteq}.
Hence $x'=x$. Thus any quasi-median of $x,v,$ and $w$ has the form
  $xvw'$.

  First, let $w'\ne w$. Since $d(x,w)=d(v,w)=k-1$ and $d(x,v)=k-2$,
  Lemma \ref{lem: no2k-1} implies that the metric triangle $xvw'$ has
  type $(k-2,k-2,k-2)$ and thus $w'$ is adjacent to $w$.  Observe that
  if $y\notin I(x,w')$, then $d(y,w')=k-3$, yielding $d(u,w')\le
  k-1$. This implies $w'\in I(w,u)\cap I(w,v)$, contrary to the
  assumption that $I(w,u)\cap I(w,v)=\{w\}$. Thus $y\notin
  I(x,w')$. Pick any neighbor $y'$ of $x$ in $I(x,w')$.  Since
  $I(x,w')\subset I(x,w)$, we conclude that $y',y\in I(x,w)$. Hence by
  $\INCz(w)$, $y'\sim y$.  Let $s$ be a common neighbor of $x$ and $u$
  and $t$ be a common neighbor of $y$ and $u$. If $s\sim t$, then
  $x,s,t,y$ define a 4-cycle which cannot be induced, thus $x\sim t$
  or $y\sim s$. This implies that $t\in I(u,w)\cap I(u,v)$ or
  $s\in I(u,v)\cap I(u,w)$, which is impossible. Hence we have
  $s\nsim t$ and by convexity of $B_2(y')$ we get $d(y',u)\leq 2$. In
  particular this implies that $d(w', u)\leq k-1$, hence
  $w'\in I(u,w)\cap I(v,w)$, which is impossible.

Now, let $w'=w$. Then $xvw$ is a metric triangle of type $(k-1,k-1,k-2)$.  By induction hypothesis, this is possible only if $k-1=2$, i.e., $k=3$. This implies that $x\sim v$, $y\sim w$, and $d(v,w)=2$. Since $uvw$ is a metric triangle, $x\nsim w$ and $y\nsim v$. Let $s$ be any common neighbor of $x$ and $u$ and $t$ be any common neighbor of $y$ and $u$. Then $xsuty$ is an induced 5-cycle, otherwise one of the vertices $s$ or $t$ belongs to $I(u,v)\cap I(u,w)$, which is impossible.
Let $z$ be a common neighbor of $v$ and $w$. By Lemma \ref{lem: almosteq}, $d(u,z)=d(u,v)=d(u,w)=3$. If $z$ is adjacent to one of the vertices $x$ and $y$, then $z$ is adjacent to both $x$ and $y$. In this case, $d(z,s)=d(z,t)=2$ and $u\in I(s,t)\setminus B_2(z)$, contrary to the convexity of $B_2(z)$. Hence $z\nsim x,y$. In particular $xywzv$ is an induced 5-cycle, and we can apply Lemma \ref{lem: deuxpent}. Thus there exists a vertex $a$ universal for one of the two pentagons $xsuty$ and $xywzv$. If $a$ is adjacent to all the vertices of $xsuty$, then $a\in I(u,v)\cap I(u,w)$, which is a contradiction as $uvw$ is a metric triangle. Thus $a$ is adjacent to all the vertices of $xywzv$. In particular, by convexity of $B_2(a)$ we get $d(a,u)=2$. Hence $a \in I(u,v)\cap I(u,w)\cap I(v,w)$, contradicting again the fact that $uvw$ is a metric triangle.
\end{proof}

We continue with a useful property of (strongly) equilateral metric triangles.

\begin{lemma}
\label{lem: equilateral1}
Let $uvw$ be an equilateral metric triangle of size $k\ge 3$ of a
CB-graph $G$. If $u'\in I(v,u)$ and $w'\in I(v,w)$ are such that
$d(v,u')=d(v,w')=2$ and $d(u',w')\le 2$, then $d(u',w') = 2$. If
$a,b,c$ are three vertices such that $a\sim u',v$; $b\sim w',v$; and
$c\sim u',w'$, then $a,b,c$ are pairwise adjacent.
\end{lemma}

\begin{proof}
  By Theorem \ref{thm: metrictriangle}, $uvw$ is strongly equilateral,
  thus $d(u',w')=2$. The convexity of $B_2(v)$ (respectively,
  $B_2(u')$, $B_2(w')$) implies that $d(c,v)\leq 2$ (respectively,
  $d(u',b) \leq 2$, $d(w',a) \leq 2$).  As $uvw$ is a metric triangle,
  since $d(u,c) \leq k-1$ and $d(w,c) \leq k-1$, necessarily
  $d(c,v)=2$ and $d(u,c) = d(w,c) = k-1$. Since $uvw$ is a strongly
  equilateral triangle, $d(w,a) = d(u,b) = k$ and
  $d(w,w') = d(u,u') =k-2$, necessarily $d(w',a) = d(u',b) = 2$.
Let $s$ be a common neighbor of $c$ and $v$. If $s$ coincides with
  one of the vertices $a,b$ and is adjacent to the second, say $s=a$
  and $s\sim b$, then $a$ and $b$ are adjacent and $a$ is adjacent to
  $c$.  Then we get the 4-cycle $acw'b$, which cannot be
  induced. Since $a \nsim w'$, we have $b\sim c$ and then $a,b,c$ are
  pairwise adjacent and we are done.  Thus, further we can suppose
  that either $s$ is different from $a$ and $b$ or $s$ coincides with
  $a$ and $b$ but is not adjacent to the second vertex.

  Let $w''$ be a neighbor of $w'$ in $I(w',w)$. Since $d(w,w'')=k-3$
  and $d(w,u')=k$ because $uvw$ is strongly equilateral, we deduce
  that $d(u',w'')=3$. This implies that $w''\ne c$ and $w''\nsim c$.
  Since $uvw$ is strongly equilateral, we have
  $d(u,w')=d(u,w'')=k$. Hence we can apply $\TPC^2(u, w'w'')$ to the
  edge $w'w''$.  First suppose that $\PCd(u,w'w'')$ applies with
  $u' \in B_2(w') \cap B_{k-2}(u)$. Then $d(u',w'') \leq 2$, a
  contradiction.
This contradiction
  shows that only $\TC(u,w'w'')$ applies to edge
  $w'w''$. Consequently, we can find $z\ne c$ adjacent to $w',w''$ and
  having distance $k-1$ to $u$. By $\INC(u)$, we get $z\sim c$ and
  there exists $z'\sim c,z$ at distance $k-2$ from $u$.  Again, since
  $d(u',w'')=3$, $z'\neq u'$. Since $z',u'\in I(u,c)$, by $\INC(u)$ we
  get $z'\sim u'$ and there exists $u''\sim z',u'$ at distance $k-3$
  from $u$.  Since $d(z,u)=k-1$, we must have $d(v,z)=3$, otherwise
  $z\in I(v,w'')\subset I(v,w)$ and we obtain a contradiction with the
  fact that $uvw$ is strongly equilateral. Analogously, one can show
  that $d(v,z')=3$. Summarizing, we obtain that
  $d(z',a)=d(z',s)=d(z,b)=d(z,s)=2$ and $d(v,z)=d(v,z')=3$.  The
  convexity of the balls $B_2(z)$ and $B_2(z')$ implies that $s$
  either is different and adjacent to $a$ and $b$ or $s$ coincides
  with one of the vertices $a,b$ and is adjacent to the second. Since
  the second case is impossible, we deduce that $s\ne a,b$ and
  $s\sim a,b$.

  Consequently, we obtain two $4$-cycles $u'csa$ and $w'csb$, which
  cannot be induced. If $s \sim u', w'$ we get
  $s\in I(u,v)\cap I(v,w)$, contradicting that $uvw$ is a metric
  triangle. Thus we can assume without loss of generality that
  $c \sim b$. Therefore the vertices $a$ and $b$ belong to the ball
  $B_2(z')$. Since $d(z',v)=3$, the convexity of $B_2(z')$ implies
  that $a\sim b$. Then we also have $a\sim c$ as the unique diagonal
  of the cycle $abcu'$, finishing the
  proof.
\end{proof}

The following lemma will be useful in the next section. It shows that the conclusions of Lemma \ref{lem: equilateral1} hold under weaker
hypotheses. 

\begin{lemma}
\label{lem: equilateral2}
 Let $uvw$ be an equilateral  metric triangle of size $k\ge 3$ of a CB-graph $G$.  Pick three vertices $a,b,w'$ such that $a\in I(v,u)$, $w'\in I(v,w)$ and $a\sim v$,
$d(w',v)=2$, and $b\sim v,w'$. If $d(a,w')\le 2$, then $d(a,w') = 2$ and for any common neighbor $c$ of $a$ and $w'$, either $c=b$ or the vertices $a,b,c$ are pairwise adjacent.
\end{lemma}

\begin{proof}
 Since $uwv$ is strongly equilateral, we must have $d(w,a)=k$, whence $d(a,w')=2$. In particular, if $a \sim b$, then $b\sim c$ as the unique diagonal
 of the 4-cycle $cabw'$.  Hence, it suffices to show that $a\sim b$. Since  $d(w,c)\le k-1$ and $uvw$ is strongly equilateral, necessarily $c\notin I(u,v)$, hence $d(u, c) \in \sg{k-1, k}$.

\medskip\noindent
{\bf Case 1:} $d(u,c)=k$.

\begin{proof}
 Since $d(u,w')=k$, we can apply $\TPC^1(u,cw')$. If $\TC(u,cw')$ applies, then there exists a vertex  $z\sim c,w'$ at distance $k-1$ from $u$.
 Since $d(a,w')=2$, necessarily $z\ne a$. By $\INC(u,c)$ we have $a\sim z$ and there exists some $z'\sim a,z$ at distance $k-2$ from $u$. Hence $z'\in I(v,u)$ has distance 2 to
 $v$ and $w'$. Analogously, if $\PCu(u,cw')$  applies with respect to the neighbor $a\in I(c,u)$ of $c$, then we can find a pentagon
 $caz'zw'$ with $z'$ at distance $k-2$ from $u$. In this case also $z'\in I(a,u)\subset I(v,u)$. In both cases we have $a\sim z',v$ and $b\sim w',v$ and $z\sim z',w'$, thus
 we can apply Lemma \ref{lem: equilateral1} with $z'$ playing the role of $u'$ and $z$ playing the role of $c$ to deduce $a\sim b$.
\end{proof}

\medskip\noindent
{\bf Case 2:} $d(u,c)=k-1$.

\begin{proof}
 Since $d(u,a)=k-1$, we can apply $\TPCz(u,ac)$. If $\TC(u,ac)$ applies, then there exists $z\sim a,c$ at distance $k-2$ from $u$. Since $z\in I(v,u)$ and $d(v,z)=d(z,w')=2$, we can apply Lemma \ref{lem: equilateral1} with $z$ playing the role of $u'$ to conclude that $a,b,c$ are pairwise adjacent. So suppose that $\PCz(u,ac)$ applies and there exist a pentagon $au'z'zc$ with $d(u,u')=d(u,z)=k-2$ and $d(u,z')=k-3$. Then we are in position to apply Lemma \ref{lem: deuxpent} to the two pentagons $zz'u'ac$ and $acw'bv$. Since  $d(v,z')=3$,  one of these two pentagons has a universal vertex (i.e., a vertex adjacent to all vertices of the respective pentagon). Since we can suppose that $\TC(u,ac)$ does not apply, the pentagon $zz'u'ac$ cannot contain a universal vertex. Thus there exists a vertex $p$ adjacent to all vertices of the pentagon $acw'bv$. Since $u',z\in B_2(p)$ and $z'\in I(u',z)$, the convexity of $B_2(p)$ implies that $z'\in B_2(p)$. Consequently, $p\in I(v,z')\cap I(v,w')\subset I(v,u)\cap I(v,w)$, contrary to the assumption that $uvw$ is a metric triangle.
\end{proof}

In both cases we either obtained a contradiction or the desired conclusion that $a\sim b$.
\end{proof}

 \section{Helly theorem}
\label{sec: helly}

In this section, we prove a Helly theorem for convex sets in
CB-graphs. We follow the method from \cite{BaCh95} for weakly modular
graphs and we show how to adapt it to graphs with convex balls, thanks
to the results of previous section.

A finite subset of vertices $A\subseteq V$ is \emph{Helly independent} (\emph{h-independent} for short)  if
$$\bigcap_{a\in A}\mathrm{conv}\left(A\setminus \sg{a}\right) = \varnothing.$$
The \emph{Helly number} $h(G)$ of $G$ is the supremum  of the size of an $h$-independent set of vertices of $G$ \cite{vdV}. Equivalently, $h(G)$ is the smallest integer $k$ such that any finite family $\mathcal F$ of convex sets of $G$ has a nonempty intersection if and only if any subfamily of $\mathcal F$ with $k+1$ members has a nonempty intersection. For a positive integer $k$, let $h_k(G)$ be the supremum of the size of an $h$-independent set of $G$ of diameter at most $k$. Clearly $h_1(G)$ is just the clique number $\omega(G)$ of $G$.

A subset of vertices $S\subseteq V$ of $G$ is called a \emph{simplex} \cite{BaCh95} if any three vertices form an equilateral metric triangle and
for any four vertices $u,v,w,x\in V$,  we have $I(u, w)\cap I(v,x)=\varnothing$. Let $\sigma(G)$ be the supremum of the  size of a simplex of $G$ and $\sigma_k(G)$ be the supremum of the  size of a simplex of diameter at most  $k$.

It is easy to see that $h(G)$ and $\sigma(G)$ are always lower bounded by the clique number $\omega(G)=h_1(G)$ of $G$ and are upper
bounded by the  Hadwiger number $\eta(G)$ \cite{DuMe}  (the size of the largest complete graph that can be obtained by contracting the edges of $G$)
and that $h(G)$ and $\sigma(G)$ are not comparable in general. However, it is shown
in \cite{BaCh95} that $h(G)\leq \sigma(G)$ holds for graphs with equilateral metric triangles and that
$h(G)=\sigma(G)=\omega(G)=h_1(G)$ holds for weakly modular graphs. We adapt their proof scheme to establish  that $h(G)=h_2(G)$
for graphs with convex balls. Notice that in graphs with convex balls, $h(G)$  may be much larger than the clique number. The simplest such example is the 5-cycle which has clique number 2 and Helly number 3. 
Another example is  the Petersen graph, which has Helly number 4 (see Figure \ref{fig: Petersen} for an $h$-independent set of size 4) and clique number 2.

  \tikzexternaldisable
  \begin{figure}[h]
    \centering
    \begin{tikzpicture}[scale=1]
    \tikzstyle{every node}=[draw,circle,fill=black,minimum size=4pt,
                            inner sep=0pt]

    \node (1) at (90:2) [color=red] {};
    \node (2) at (18:2) [color=red] {};
    \node (3) at (306:2) {};
    \node (4) at (234:2) {};
    \node (5) at (162:2) {};
    \node (6) at (90:1) {};
    \node (7) at (18:1) {};
    \node (8) at (306:1) [color=red] {};
    \node (9) at (234:1) {};
    \node (10) at (162:1) [color=red] {};

    \draw[thick] (1) -- (2) -- (3) -- (4) -- (5) -- (1) -- (6) -- (8) -- (3);
    \draw[thick] (2) -- (7) -- (9) -- (4);
    \draw[thick] (8) -- (10) -- (5);
    \draw[thick] (6) -- (9);
    \draw[thick] (10) -- (7);

    \end{tikzpicture}
    \caption{The red vertices form an $h$-independent set of size $4$ of the Petersen graph.}  \label{fig: Petersen}
\end{figure}
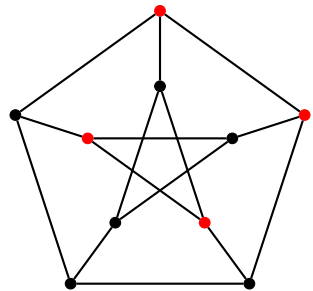

The main result of this section shows that the Helly number of a CB-graph can be defined locally:

\begin{theorem}
 \label{thm: hellydiam2}
 Let $G$ be a CB-graph. Then $h(G)=h_2(G)$.\end{theorem}

The proof of Theorem \ref{thm: hellydiam2} is based on  three lemmas. First we recall the following general result: 

\begin{lemma}[\cite{BaCh95}*{Lemma 5}]
\label{lem: LBaCh}
 Let $A\subseteq V$ be an $h$-independent set of a graph $G$. If $x \in I(u,v)\cap \mathrm{conv}\left(A\setminus \sg{u}\right)$ for some vertices $u,v\in A$,
 then the set $B:=\left(A\setminus\sg{v} \right)\cup \sg{x}$ is $h$-independent with $|B|=|A|$. In particular, $B$ is $h$-independent when
 $x$ is chosen from $I(u,v)\cap I(v,w)$ for distinct $u,v,w\in A$.
\end{lemma}

For the sake of simplicity, by a \emph{distance-minimal set} we mean any set $A\subseteq V$ such that for any two distinct $u, v \in A$ we have
$I(u,v)\cap \mathrm{conv}\left(A\setminus \sg{u}\right) = \sg{v}.$ We will use the operation of Lemma \ref{lem: LBaCh} to
transform  $h$-independent sets into  $h$-independent sets with ``nicer properties''.

\begin{lemma}\label{lem: hleqsigma}
  If a CB-graph $G$ contains an $h$-independent set $A$, then $G$
  contains a distance-minimal $h$-independent simplex of $G$ or an
  $h$-independent set of diameter at most $2$. In particular,
  $h(G)\leq \max\{ h_2(G), \sigma(G)\}$.
\end{lemma}

\begin{proof}
  For any vertex $u$ of $A$, let
  $\Delta(u,A):= \sum_{v\in A\setminus \sg{u}} d(u,v).$ Let also
  $\Delta(A):=\min_{z\in A}\Delta(z,A).$ Since $A$ is finite, these
  values are well defined.  We call any vertex $z\in A$ such that
  $\Delta(z,A)=\Delta(A)$ \emph{minimal}.  Pick any minimal vertex
  $z\in A$ and we make any other vertex of $A$ as close as possible to
  $z$.  Namely, if there exists some $u\in A\setminus \sg{z}$ such
  that $I(u,z)\cap \mathrm{conv}(A\setminus \sg{z})\neq \sg{u}$, then
  we pick any vertex $x\ne u$ from this intersection and set
  $B:= \left(A\setminus \sg{u}\right)\cup \sg{x}$. By Lemma \ref{lem:
    LBaCh}, the set $B$ is $h$-independent. Notice also that
  $\Delta(B)\leq \Delta(z, B)<\Delta(z,A)=\Delta(A)$. Hence, if we set
  $A:=B$ and apply this transformation to minimal vertices of $A$ as
  long as we decrease $\Delta(A)$, after a finite number of steps we
  will end up with an $h$-independent set $S$ of the same cardinality
  as the initial set $A$ and such that
  $I(u,z)\cap \mathrm{conv}(S\setminus\sg{z})=\sg{u}$ holds for any
  minimal vertex $z$ of $S$ and any other vertex $u \in S$. This
  implies that $I(u,v)\cap I(u,z)=\{ u\}$ and
  $I(v,u)\cap I(v,z)=\{ v\}$ for any three distinct vertices $u,v,z$
  of $S$ such that $z$ is minimal. We distinguish two cases.

 \medskip\noindent
 {\bf Case 1:} $z\sim u$ for a minimal vertex $z$ of $S$ and a vertex $u\in S\setminus \sg{z}$.

 \begin{proof}
By construction, for any $v\in S\setminus\sg{z,u}$, we have
   $I(v,z)\cap I(v,u)=\sg{v}$ and $I(u,z)\cap I(u,v)=\sg{u}$. In
   particular, we have $u\notin I(v,z)$. Consequently, $zuv$ is a metric triangle
   of $G$. By Theorem \ref{thm: metrictriangle}, $zuv$ is either an
   equilateral triangle or it lies on a pentagon, i.e.,
   $d(z,uv)=d(u,v)=2$. This implies that $S\subseteq B_2(z)$ and that
   $B_1(z)\cap S$ form a clique. We assert that
   $\mathrm{diam}(S)\leq 2$. Pick any $v, w \in S\setminus \sg{z}$.
   If $d(z,v)=2$ and $w\sim z$, then we saw that $d(v,w)=2$. Now, let
   $d(z,v)=d(z,w)=2$. If we assume that $d(v,w)\geq 3$, then by
   Theorem \ref{thm: metrictriangle} $zvw$ is not a metric triangle.
   By the property of the set $S$ there exists a vertex
   $y \in I(z, v)\cap I(z,w)$ different from $z$. Since
   $d(z,v)=d(z,w)=2$, $y$ must be adjacent to $v$ and $w$, contrary to
   the assumption $d(v,w)\ge 3$. Hence $d(v,w)\le 2$, establishing
   that $S$ has diameter at most 2. Thus in this case we have
   $|A|=|S|\leq h_2(G)$.
 \end{proof}

\medskip\noindent
 {\bf Case 2:} $z\nsim u$ for any minimal vertex $z$ of $S$ and any vertex $u\in S\setminus \sg{z}$.

\begin{proof}
  We assert that in this case $S$ is a distance-minimal simplex. The
  second condition in the definition of a simplex is always true in
  $h$-independent sets, hence we only need to show that any three
  vertices of $S$ form an equilateral metric triangle. Pick any
  triplet $z,u,v$ of vertices of $S$ such that $z$ is minimal.  We
  already know that $I(u,v)\cap I(u,z)=\{ u\}$ and
  $I(v,u)\cap I(v,z)=\{ v\}$. Let $y \in I(z,u)\cap I(z,v)$ such that
  $I(y,u)\cap I(y,v)=\sg{y}$. Then clearly $yuv$ is a metric
  triangle. We assert that $d(u,v)\le \min \{d(z,u), d(z,v)\}$ and
  that equality holds only if $z=y$ and $zuv$ is an equilateral
  triangle.
By Theorem \ref{thm: metrictriangle}, the metric triangle $yuv$ is
  equilateral or has type $(2,2,1)$. If $yuv$ is equilateral, then
  $d(z,u)=d(z,v)=d(u,v)$ if $z=y$ and $d(z,u)=d(z,v)>d(u,v)$ if
  $z\ne y$. Now assume that $yuv$ has type $(2,2,1)$. If $d(u,v)=1$,
  then $d(z,u)=d(z,v)\ge d(y,u)=d(y,v)=2$ and we are done. Now suppose
  that $d(u,v)=2$. Then $y$ is adjacent to one of the vertices $u,v$,
  say $y\sim v$. Since $z$ is not adjacent to any vertex of
  $S\setminus \{ z\}$, we must have $z\ne y$.  We assert that
  $\Delta(v,S)<\Delta(z,S)$. First notice that
  $d(u,v)=2 < d(u,y)+d(y,z)=d(u,z)$.  For any
  $w \in S\setminus \sg{z,u,v}$, let $y'\in I(z,v)\cap I(z,w)$ be such
  that $y'vw$ is a metric triangle. By Theorem \ref{thm:
    metrictriangle}, either $y'vw$ is equilateral, in which case
  $d(v,w)=d(y', w)\leq d(z,w)$, or $y'vw$ has type $(2,2,1)$ and,
  since $z\notin B_1(S\setminus\sg{z})$, we must have
  $d(v,w)\leq 2 \leq d(z,w)$.  Consequently,
  $\Delta(v,S)< \Delta(z,S)$, contradicting the minimality choice of
  $z$. This shows that $d(u,v)\le \min\{ d(z,u), d(z,v)\}$ and that
  equality holds only if $z =y$ and $zuv$ is equilateral. Since
  $\Delta(z)\le \Delta(u)$, all inequalities $d(u,v)\le d(z,v)$ are
  equalities. This proves that $z=y$ and that each metric triangle
  $zuv$ is equilateral. Consequently, all vertices of $S$ are minimal,
  yielding that all triplets $u,v,w$ of $S$ define equilateral metric
  triangles of size $k\ge 2$. In particular, $S$ is a simplex, whence
  $|S|\leq \sigma(G)$.
\end{proof}

Consequently, in Case 1 $S$ is an $h$-independent set of diameter at most 2 and in Case 2 $S$ is an $h$-independent distance-minimal simplex.
\end{proof}

Now we show that the simplex from the second case of the previous proof can be reduced to a clique of $G$ of the same size.

\begin{lemma}
 \label{lem: h-indepdiam2}
 Let $S\subseteq V$ be a distance-minimal $h$-independent simplex of a CB-graph $G$ and let $\diam(S)\ge 3$. Then $\mathrm{conv}(S)$
 contains a clique $C$ of size $|S|$.
\end{lemma}

\begin{proof} The case $|S|=2$ is trivial. Thus,  let $|S|\geq 3$.  Let  $k:=\mathrm{diam}(S) \geq 3$, which corresponds to the distance between any two distinct vertices of $S$. Pick  $u,v\in S$ and let $x$ be a neighbor of $v$ in the interval $I(u,v)$. For every $w\in S\setminus \sg{u, v}$, we can apply $\TPCz(w,vx)$. We partition $S\setminus \sg{u,v}$ into the sets  $S_1$ and $S_2$: $S_1$ is the set of all vertices $w$ of $S\setminus \sg{u,v}$ such that $\TC(w,vx)$ applies and $S_2:= S\setminus S_1$ (notice that for all  $w\in S_2$ the condition $\PCz(w,vx)$ applies).

\begin{claim} \label{cl:empty}
The set $S_2$ is empty.
\end{claim}
\begin{proof}
Suppose by way of contradiction that $w\in S_2$. Let $z$ be a vertex at distance 2 from $v$ and $x$ such that $z\in I(v,w)\cap I(x,w)$.
Let $t$ be a common neighbor of $z$ and $v$. Then we can apply Lemma \ref{lem: equilateral2} with the vertices $u,v,w,t,z,x$ playing the role of $u,v,w,b,w',a$, respectively.
Thus we get that $t\sim x$. This  implies that  $\TC(w,vx)$ applies, contrary to the assumption that $w\in S_2$. Hence $S_2= \varnothing$.
\end{proof}

Let $A_1$ be the set of all common neighbors $z$ of $v$ and $x$ such that there exists a vertex $w\in S_1$ such that $z\in I(v,w)\cap I(x,w)$. We call $z$ an \emph{imprint} of $w$ on the edge $vx$.

\begin{claim}
\label{clm: S1'}
The set $C:=A_1\cup \{ v,x\}$ induces a clique of size at least $|S|$.
\end{claim}

\begin{proof}
Pick any two vertices $z_1, z_2$ in $A_1$ and suppose that  $z_1,z_2$ are  imprints of the vertices $w_1,w_2\in S_1$, respectively. If $w_1=w_2$, then $z_1\sim z_2$
by $\INCz(w_1)$. Thus  assume that $w_1\neq w_2$. Assume by way of contradiction that $z_1\nsim z_2$. This implies that $x \in I(z_1, z_2)$. Since $z_1, z_2\in I(w_1, v)\cup I(w_2, v)\subseteq \mathrm{conv}(S\setminus \sg{u})$, we obtain that $ x\in \mathrm{conv}(S\setminus \sg{u})\cap I(u, v).$
This contradicts the fact that the set $S$ is distance-minimal. This establishes that $A_1$ induces a clique.  Since all vertices of $A_1$ are different from $v,x$ and are adjacent to $v$ and $x$, $C=A_1\cup \{ v,x\}$ also induces a clique.

From the definition of $S_1$ and since  any three vertices of $S$ form a metric triangle it follows that any two vertices of $S_1$ have different imprints, thus $|A_1|\geq |S_1|$.
By Claim \ref{cl:empty} the set $S_2$ is empty, yielding $|A_1|\geq |S_1|=|S\setminus\sg{u,v}|=|S|-2.$ Consequently, $|C|\geq |S|$.
\end{proof}

The desired result is now a direct consequence of Claim \ref{clm: S1'}.
\end{proof}

We now prove the equality $h(G)=h_2(G)$. Obviously, $h(G)\ge
h_2(G)$. We now establish the converse inequality.  Given a finite
$h$-independent set $A$ of size $k$, by Lemma \ref{lem: hleqsigma},
we can assume that $A$ is an $h$-independent set $S$ of the same size
and which either has diameter at most $2$ or is a distance-minimal
simplex.  In the second case, by Lemma \ref{lem: h-indepdiam2}, we
find a clique $C\subseteq\mathrm{conv}(S)$ of the same size as $S$ and
$A$. Since $C$ is an $h$-independent of diameter $1$, in both cases,
we have found an $h$-independent of diameter at most $2$ and of size
$k$. This establishes that $h(G) \leq h_2(G)$ when $h(G)$ is finite.

If $h(G)$ is infinite, then for each integer $k$ there exists an
$h$-independent set $A_k$ of size $k$. By the previous argument, we
can assume that each $A_k$ is an $h$-independent $S_k$ of size $k$ and
of diameter at most $2$. This implies that $h_2(G)$ is also infinite
and thus $h(G) = h_2(G)$. This finishes the proof of Theorem \ref{thm:
  hellydiam2}.

\medskip
We were not able to prove that any $h$-independent set of diameter 2 can be transformed into an $h$-independent simplex of diameter 2 and of the same size.We formulate this as an open question:

\begin{question} Is it true that for any CB-graph $G$, the equality $h_2(G)=\sigma_2(G)$ holds?
\end{question}

\section*{Acknowledgements}
We are grateful to the anonymous referee for a careful reading of all parts of the paper and numerous useful comments.
J.C.\ and V.C.\ were  supported  by ANR project DISTANCIA (ANR-17-CE40-0015). U.G. was supported by an internship grant of ``Pole Calcul" of LIS, Aix-Marseille Université and
by ANR project DISTANCIA (ANR-17-CE40-0015).

\bibliographystyle{plain}
\bibliography{biblio}

\end{document}